\documentclass[aps,prx,reprint,superscriptaddress,longbibliography,twocolumn,groupedaddress
]{revtex4-2}

\usepackage[dvipdfmx]{graphicx}
\graphicspath{{./}{./figure/}}

\usepackage[utf8]{inputenc}
\usepackage{CJKutf8}

\usepackage{amsmath,amssymb,bm,braket}
\usepackage{enumitem}
\usepackage{comment}
\usepackage{multirow}
\usepackage{tabularx}
\newcolumntype{Y}{&gt;{\centering\arraybackslash}X} %centering

\usepackage{xcolor}

\usepackage{diagbox}

\usepackage{amsthm}
\newtheorem{thm}{Main Theorem}
\newtheorem*{thm*}{Main Theorem}
\newtheorem{lem}[thm]{Lemma}
\newtheorem*{lem*}{Lemma}

\newtheorem*{conj*}{Conjecture}
\theoremstyle{definition}
\newtheorem*{dfn*}{Definition}

\usepackage{nameref}

\usepackage{mathtools}
\usepackage{cancel}

\newcommand{\eq}[1]{\begin{align} #1 \end{align}}

\newcommand{\nx}{\nonumber \\
}
\newcommand{\sumi}{\sum_{i=1}^N}
\newcommand{\sums}{\sum_{\rm{shift}}}

\renewcommand{\deg}{{\rm deg}}

\newcommand{\wha}{&~ \phantom{=}}
\newcommand{\whb}{&~ \phantom{=\sums}}

\newcommand{\half}{\tfrac{1}{2}}

\newcommand{\co}{\circ}
\newcommand{\bd}{b^\dagger}
\newcommand{\hb}{\bar{h}}
\newcommand{\lo}[1]{(#1 \text{-} \rm{local})}
\newcommand{\op}[1]{\underline{\smash{#1}}}
\newcommand{\und}[2]{\underbracket[0.5pt]{#1}_{\substack{#2}}}
\newcommand{\pk}{^{(k)}}
\newcommand{\bx}[1]{\overline{\underline{{\smash[b]{#1}}}}}
\newcommand{\akkun}{\overset{\star}{=}}
\newcommand{\lbl}[1]{\label{#1} {\magenta \texttt{[#1]}}}
\renewcommand{\lbl}[1]{\label{#1}}

\newcommand{\bA}{\boldsymbol{A}}
\newcommand{\bB}{\boldsymbol{B}}

\newcommand{\eref}[1]{Eq.~\eqref{#1}}

\newcommand{\lref}[1]{Lemma~\ref{t:#1}}

\newcommand{\sref}[1]{Sec.~\ref{s:#1}}
\newcommand{\ssref}[1]{Subsec.~\ref{s:#1}}
\newcommand{\apref}[1]{Appx.~\ref{s:#1}}
\newcommand{\mref}{\hyperref[t:main]{Main Theorem}}

\usepackage{hyperref}
\hypersetup{colorlinks=true, citecolor=magenta, linkcolor=blue, urlcolor=cyan}

\begin{document}
%\begin{CJK*}{UTF8}{min}
%\date{\today}
\title{Violating the All-or-Nothing Picture of Local Charges \\in Non-Hermitian Bosonic Chains}

\author{Mizuki Yamaguchi}
\email{yamaguchi-q@g.ecc.u-tokyo.ac.jp}
\affiliation{
%C Group, Department of Basic Science, Department of Multidisciplinary Sciences, 
Graduate School of Arts and Sciences, The University of Tokyo, 3-8-1 Komaba, Meguro, Tokyo 153-8902, Japan}

\author{Naoto Shiraishi}
\email{shiraishi@phys.c.u-tokyo.ac.jp}
\affiliation{
%C Group, Department of Basic Science, Department of Multidisciplinary Sciences, 
Graduate School of Arts and Sciences, The University of Tokyo, 3-8-1 Komaba, Meguro, Tokyo 153-8902, Japan}

\begin{abstract}
We present explicit counterexamples to a widespread empirical expectation that local commuting charges display all-or-nothing behavior. In the class of bosonic chains with symmetric nearest-neighbor hopping and arbitrary on-site terms (including non-Hermitian terms), we exhibit systems that possess $k$-local charges for some but not all $k$. Concretely, we construct non-Hermitian models with a 3-local charge but no other nontrivial local charges and models with $k$-local charges for all $k$ except $k=4$. These results show that the Grabowski--Mathieu integrability test based on 3-local charges is not universally applicable. We further give necessary and sufficient conditions for the existence of $k$-local charges in this class, yielding an exhaustive classification and uncovering additional integrable models.
\end{abstract}

\maketitle
\section{introduction}
\lbl{s:introduction}
Among interacting quantum many-body systems, exact solvability, often equated with quantum integrability, is a rare exception. Integrable models offer sharp analytical access to spectra, eigenstates, and correlation functions 
\cite{bethe1931zur,lieb1961two,lieb1963exact,mcguire1964study,yang1967exact,sutherland1968further,yang1970onedimensional,jimbo1980density,korepin1982calculation,jimbo1985qdifference,shastry1986exact,drinfeld1988quantum,haldane1988exact,slavnov1989calculation,essler1995determinant,jimbo1996quantum,gohmann2004integral}
and provide conceptual insights into strong correlations
\cite{girardeau1960relationship,hubbard1963electron,lieb1968absence,takahashi1972onedimensional,andrei1980diagonalization,wiegmann1981exact,deguchi2000thermodynamics,caux2006fourspinon,guan2013fermi}, criticality
\cite{onsager1944crystal,baxter1972partition,blote1986conformal},
nonequilibrium dynamics
\cite{zotos1997transport,kinoshita2006quantum,rigol2007relaxation,bertini2016transport,castro-alvaredo2016emergent,denardis2022correlation,doyon2025generalized},
and quantum field theory
\cite{coleman1975quantum,karowski1978exact,belavin1984infinite,affleck1986universal,zamolodchikov1990thermodynamic,destri1992new,smirnov1992form,bazhanov1997integrable,minahan2003betheansatz,beisert2005longrange,gromov2014quantum}. Their rarity and usefulness naturally raise a practical question: how can one tell whether a given Hamiltonian is integrable? In the literature, integrable models have often been identified heuristically, for instance via the Bethe ansatz
\cite{baxter1985exactly,korepin1993quantum,takahashi1999thermodynamics,gaudin2014bethe}, without a broadly applicable diagnostic.

Accumulated evidence has suggested an all-or-nothing contrast when focusing on local commuting charges. In known integrable models, one finds (at least) one local charge $Q_k$ supported on $k$ consecutive sites for each $k=3,4,5,\dots$
\cite{lax1968integrals,kulish1976factorization,takhtadzhan1979quantum,thacker1980quantum,kulish1982quantum,tetelman1982lorentz,reshetikhin1983method,sogo1983boost,sklyanin1992quantum,grabowski1994quantum,grabowski1995structure,faddeev1996how,hokkyo2025integrability}.
By contrast, in generic (nonintegrable) systems, no such finite-range charges are expected beyond obvious global symmetry charges, as suggested by macroscopic considerations and numerical studies
\cite{srednicki1995does,peres2002quantum,caux2005computation,rigol2008thermalization,caux2011remarks,polkovnikov2011colloquium,steinigeweg2013eigenstate,eisert2015quantum,dalessio2016quantum,gogolin2016equilibration,mori2018thermalization,bertini2021finitetemperature}. 
This all-or-nothing empirical picture motivated Grabowski and Mathieu to propose an integrability test based on whether a 3-local charge exists
\cite{grabowski1995integrability,gombor2021integrable}.
More recently, direct analyses of local charges and exhaustive classifications have provided rigorous proofs of this picture in several model classes
\cite{shiraishi2025dichotomy,yamaguchi2024complete,hokkyo2025absence}.

\begin{figure}[t]
\centering
\includegraphics[width=0.90\columnwidth]{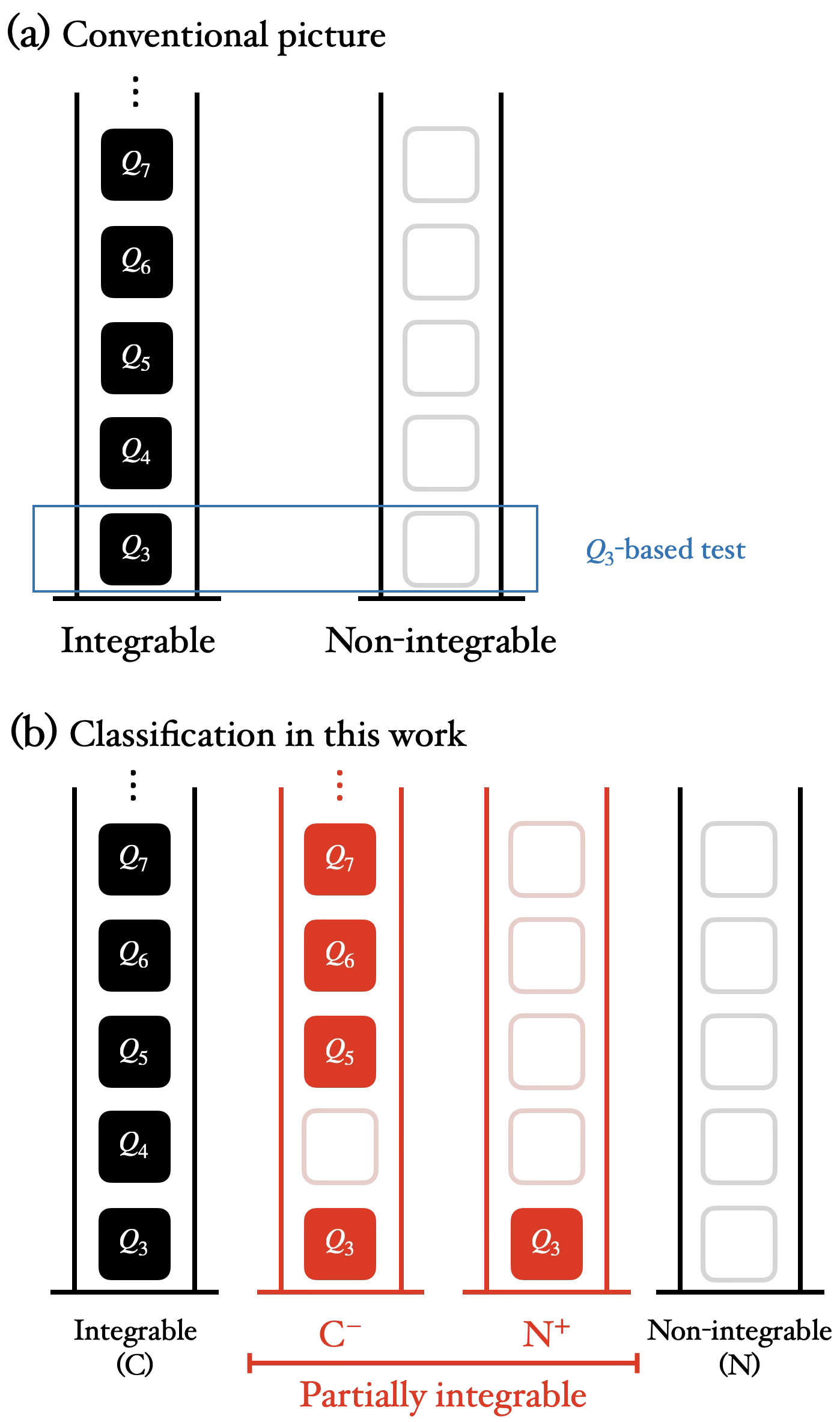}
\caption{Breakdown of the conventional all-or-nothing picture of integrability. Filled and unfilled boxes denote the presence and absence of local charges, respectively. (a) In the conventional view, integrable models admit a local charge $Q_k$ for every $k$, whereas non-integrable models admit none. Under this paradigm, the existence of the 3-local charge $Q_3$ serves as a diagnostic of integrability. (b) For the model class considered in this work, however, the classification is richer. In addition to the integrable case (Type~C) and the non-integrable case (Type~N), two partially integrable types appear: Type~C$^-$, in which $Q_4$ is absent despite the presence of $Q_3$ and $Q_{\geq 5}$, and Type~N$^+$, in which only $Q_3$ exists. Thus, the existence of $Q_3$ alone does not universally diagnose integrability.}
\lbl{f:breakdown}
\end{figure}

Local-charge structure has been less explored in non-Hermitian settings and in bosonic ones. In particular, for non-Hermitian Hamiltonians, non-normality preclude standard level-statistics probes
\cite{sa2020complex,suthar2022nonhermitian,xiao2022level,roccati2024diagnosing,akemann2025two,you2026limitations},
whereas for bosons the unbounded local Hilbert space makes systematic investigations substantially harder
\cite{kolovsky2004quantum,kollath2010statistical,wittmannw.2022interacting,anh-tai2023quantum}. 
Although both settings admit integrable examples with complete hierarchies of local charges
\cite{zheng2024exact,fring2013ptsymmetric,nakagawa2021exact,mao2023nonhermitian,wang2023scalefree,kattel2025spin}, 
comprehensive local-charge studies in these settings are still lacking.

Here we report explicit violations of the all-or-nothing picture of local charges in non-Hermitian bosonic chains. We construct \emph{partially integrable} counterexamples, meaning that only part of the local-charge hierarchy survives, within the translationally invariant class with symmetric nearest-neighbor hopping and general onsite non-Hermitian terms. Concretely, we identify (i) models with a 3-local charge but no other nontrivial finite-range charge, and (ii) models with $k$-local charges for all $k$ except $k=4$. These counterexamples show that the lowest-range ($k=3$) information need not propagate to higher ranges, thereby demonstrating that the Grabowski--Mathieu integrability test based on 3-local charges \cite{grabowski1995integrability} is not universally applicable.

Beyond constructing these counterexamples, we provide an exhaustive classification of this model family in terms of local charges, explicitly specifying which Hamiltonians are integrable, nonintegrable, or partially integrable. Partial integrability occurs only in the non-Hermitian sector, whereas the Hermitian subclass retains the all-or-nothing picture. The classification also uncovers several new integrable models and, in some cases, an occasional even--odd sensitivity in the local-charge structure. Taken together, these results provide a concrete map of integrability in this non-Hermitian bosonic setting, highlighting a richer structure of local charges than conventional expectations suggest.

The structure of this paper is as follows. In \sref{setup}, we introduce the model and state our main classification theorem. In \sref{proof1}, we develop tools for analyzing local charges and carry out the model-independent preprocessing. In \sref{proof2}, we give a complete analysis of the existence or absence of 3-local charges. In \sref{proof3}, we finish the classification by (i) a case-by-case analysis of models with a 3-local charge and (ii) proving in general that, if no 3-local charge exists, then no higher local charge exists either. Finally, in \sref{discussion}, we summarize our findings and discuss future directions.
\section{setup and main result}
\lbl{s:setup}

We consider a translationally invariant bosonic chain of length $N$ with periodic boundary conditions.
The Hamiltonian consists of symmetric nearest-neighbor hopping and a general on-site term,
\eq{
H = \sumi \left( \bd_i b_{i+1} + b_i \bd_{i+1} + g_i \right)
\lbl{H-gen}
}
where the on-site operator $g$ is identical on every site and may be non-Hermitian.
We write it as
\eq{
g =\sum_{x,y=0}^{\infty}c_{xy}(\bd)^x b^y,
\lbl{g-expand}
}
where the coefficients $c_{xy}$ are arbitrary complex numbers.

We call an operator $Q$ a \textit{(commuting) charge} if $[Q,H]=0$ holds.\footnote{In the literature on integrable systems, commuting charges are often referred to as \emph{conserved quantities}, regardless of whether the system is Hermitian or non-Hermitian. In this paper, however, we use the term commuting charge, since in non-Hermitian systems the meaning of ``conservation'' can be context-dependent, and we wish to avoid possible terminological confusion.}

\begin{dfn*}[$k$-local charge]
A charge $Q$ is called \textit{$k$-local} if it can be written as a sum of operators supported on contiguous intervals of length at most $k$, and at least one term has support of length exactly $k$.
\end{dfn*}

To date, only two extreme behaviors of $k$-local charges have been known: integrable models have them for all $k$, while nonintegrable models have none.
As in prior work 
\cite{shiraishi2019proof,yamaguchi2024proof}, $k$ is taken in the window $3\le k\le N/2$, which is the natural setting for distinguishing local-charge structures.\footnote{The cases $k=1,2$ are dominated by short-range charges that can appear even in nonintegrable models; for instance, $H$ itself is always a trivial 2-local charge, and the total particle number $\sum_i n_i$ may also be commutative. By contrast, for $k>N/2$ a periodic chain admits redundant nonlocal composites (e.g., powers of $H$).}

\begin{table*}[t!]
\def\arraystretch{1.5}
\begin{tabular}{|c|c|c||cccccc|}
\hline
Allowed coeffs. & Conditions & Type & 3 & 4 & 5 & 6 & $\cdots$ & $\tfrac{N}{2}$ \\ \hline
\multirow{2}{*}{
\shortstack{ $c_{40}(\neq0), c_{30}, c_{20} , c_{10} , c_{11}, c_{01}$ \\ 
with $c_{30}(c_{11}+2)=4c_{40}c_{01}$}} & $c_{11}=0$ & C 
& $\checkmark$ & $\checkmark$ & $\checkmark$ & $\checkmark$ & $\checkmark$ & $\checkmark$
\\ \cline{2-9}
& $c_{11} \neq 0$ & N$^+$
& $\checkmark$ & $-$ & $-$ & $-$ & $-$ & $-$
\\ \hline
its adjoint & \multicolumn{8}{c|}{as above} \\ \hline
\multirow{3}{*}{
$c_{30}(\neq 0), c_{20} , c_{10} , c_{11}, c_{01}$} & $c_{20}(c_{11}+2)=3c_{30}c_{01}$ & C 
& $\checkmark$ & $\checkmark$ & $\checkmark$ & $\checkmark$ & $\checkmark$ & $\checkmark$
\\ \cline{2-9}
& $c_{11} =1 $ & C 
& $\checkmark$ & $\checkmark$ & $\checkmark$ & $\checkmark$ & $\checkmark$ & $\checkmark$
\\ \cline{2-9}
& otherwise & C$^-$
& $\checkmark$ & $-$ & $\checkmark$ & $\checkmark$ & $\checkmark$ & $\checkmark$
\\ \hline
its adjoint & \multicolumn{8}{c|}{as above}\\ \hline
\multicolumn{1}{|c}{$c_{20}, c_{10} , c_{11}, c_{01}, c_{02}$} &  & C 
& $\checkmark$ & $\checkmark$ & $\checkmark$ & $\checkmark$ & $\checkmark$ & $\checkmark$
\\ \hline
\multicolumn{1}{|c}{otherwise} &  & N 
& $-$ & $-$ & $-$ & $-$ & $-$ & $-$
\\ \hline
\end{tabular}
\caption{
Classification of models based on uniform local charges.
The coefficients $c_{xy}$ are those defined in \eref{g-expand}, i.e., the expansion coefficients of the on-site term $g$.
The checkmark ($\checkmark$) means that this $k$-local charge exists, and the hyphen ($-$) means that it does not.
In particular, one can see that the Grabowski--Mathieu integrability test is not universally applicable, because the presence or absence of a 3-local charge alone does not predict the structure of higher local charges.
}
\lbl{table-classify-I}
\end{table*}

\begin{table*}[t!]
\def\arraystretch{1.5}
\begin{tabular}{|c|c|c||cccccc|}
\hline
Allowed coeffs. & Conditions & Type & 3 & 4 & 5 & 6 & $\cdots$ & $\tfrac{N}{2}$ \\ \hline
\multirow{2}{*}{
\shortstack{$c_{x0}$ for any $x$ ($\exists x \geq 3, c_{x0} \neq 0$),\\ $c_{02},c_{01}$}} & $c_{02} = 0 $ & SC & $\checkmark$ & $\checkmark$ & $\checkmark$ & $\checkmark$ & $\checkmark$ & $\checkmark$ \\ \cline{2-9}
& $c_{02} \neq 0$ & SN$^+$
& $\checkmark$ & $-$ & $-$ & $-$ & $-$ & $-$ \\ \hline
its adjoint & \multicolumn{8}{c|}{as above} \\ \hline
\multicolumn{1}{|c}{$c_{20}, c_{10}, c_{01}, c_{02}$} &  & SC 
& $\checkmark$ & $\checkmark$ & $\checkmark$ & $\checkmark$ & $\checkmark$ & $\checkmark$ \\ \hline
\multicolumn{1}{|c}{otherwise} &  & SN 
& $-$ & $-$ & $-$ & $-$ & $-$ & $-$ \\ \hline
\end{tabular}
\caption{
Classification of models based on staggered local charges.
The meanings of the symbols are the same as Table~\ref{table-classify-I}.
}
\lbl{table-classify-II}
\end{table*}

We analyze local charges separately in each momentum sector. This is justified because the one-site shift superoperator $T$ commutes with $[\bullet,H]$. Concretely, one may write $Q$ as a sum of eigenoperators of $T$ satisfying $T[Q]=e^{ip}Q$, and each component is itself a charge.

We are now ready to state our main result.

\begin{thm*}\phantomsection\lbl{t:main}
Consider the translationally invariant bosonic chain with symmetric nearest-neighbor hopping and a general on-site term, given by \eref{H-gen}.
Then the following statements hold.

\begin{itemize}
\item[(i)] For any Hamiltonian in this class, any local charge can be decomposed into a sum of local charges in either the $p=0$ sector or the $p=\pi$ sector.

\item[(ii)] Each Hamiltonian falls into exactly one of the following four types, based on its $p=0$-sector local charges:
\par\noindent \textup{\textbf{Type~C}}: A $k$-local charge exists for each $k$ with $3\le k\le N/2$.

\par\noindent \textup{\textbf{Type~N}}: No $k$-local charge exists for any $k$ with $3\le k\le N/2$.

\par\noindent \textup{\textbf{Type~N$^+$}}: A $3$-local charge exists, but no $k$-local charge exists for any $k$ with $4\le k\le N/2$.

\par\noindent \textup{\textbf{Type~C$^-$}}: A $3$-local charge exists and a $k$-local charge exists for each $k$ with $5\le k\le N/2$, but no $4$-local charge exists.

\item[(iii)] Each Hamiltonian falls into exactly one of the following three types, based on its $p=\pi$-sector local charges:
\par\noindent \textup{\textbf{Type~SC}}: A $k$-local charge exists for each $k$ with $3\le k\le N/2$.

\par\noindent \textup{\textbf{Type~SN}}: No $k$-local charge exists for any $k$ with $3\le k\le N/2$.

\par\noindent \textup{\textbf{Type~SN$^+$}}: A $3$-local charge exists, but no $k$-local charge exists for any $k$ with $4\le k\le N/2$.
\end{itemize}

The correspondence between Hamiltonians and types is summarized in Tables~\ref{table-classify-I} and \ref{table-classify-II}.
\end{thm*}

We refer to the $p=0$ and $p=\pi$ sectors as the \emph{uniform} and \emph{staggered} sectors, respectively.
From a mathematical-physics viewpoint, uniform local charges provide the most natural classification: each Hamiltonian falls into one of the four types in statement~(ii).
Among these, Type~C corresponds to a usual integrable system in which the full hierarchy of local charges exists (hence ``completely integrable''), whereas Type~N corresponds to a usual nonintegrable system.
By statement~(i), the only other possibility for local charges is the staggered sector, which provides a different classification into the three types in statement~(iii).

Among the uniform-sector types, Type~N$^+$ and Type~C$^-$ are especially notable as explicit counterexamples to the empirical all-or-nothing picture of local charges.
This all-or-nothing picture has long been widely believed \cite{rigol2008thermalization,caux2011remarks,eisert2015quantum,dalessio2016quantum,gogolin2016equilibration}.
In particular, Grabowski and Mathieu \cite{grabowski1995integrability} proposed an integrability test based on the existence of a 3-local charge, motivated by this expectation.
Later work \cite{gombor2021integrable} has put forward conjectures in the same spirit.
If the all-or-nothing picture were universally valid, checking for a 3-local charge would suffice to distinguish complete integrability from nonintegrability.
Our classification shows that this simplification fails, at least for non-Hermitian bosonic chains.
In the following, we introduce atypical models whose local-charge behavior is nonstandard.

\subsection{Partially integrable counterexample I: Type~N$^+$}
A Type~N$^+$ model has a 3-local charge but no $k$-local charges for $4\le k\le N/2$.
A representative example of the on-site term is
\eq{
g = c_{40} (\bd)^4 + c_{11} n,
}
with $c_{40} \neq 0$ and $c_{11} \neq 0$.
Its unique nontrivial (i.e., $3\le k\le N/2$) local charge is the following 3-local charge:
\eq{
Q_3 &= \sumi \bd_ib_{i+2}-b_i\bd_{i+2} \nx
\whb + 4c_{40}\left(\bd_i(\bd_{i+1})^3-(\bd_i)^3\bd_{i+1}\right) \nx
\whb +4c_{11}(\bd_ib_{i+1}-b_i\bd_{i+1}).
}
This is the simplest representative. In general, Type~N$^+$ is characterized by the on-site term of the form
\eq{
g = c_{40} (\bd)^4 + c_{30} (\bd)^3 + c_{20} (\bd)^2 + c_{11} n + \frac{c_{30}(2+c_{11})}{4c_{40}} b
}
with $c_{40}\neq 0$ and $c_{11} \neq 0$, or by its adjoint (Hermitian conjugate).

This model is partially integrable in the sense that the existence of a 3-local charge does not lead to any higher local charges.
This is the first report of local-charge behavior that differs from both conventional integrable and nonintegrable systems.

We provide the proof of the absence of $k$-local charges for $4\leq k\leq N/2$ in \lref{degree4-N+} (\ssref{model-4}).

\subsection{Partially integrable counterexample II: Type~C$^-$}
A Type~C$^-$ model has a 3-local charge and $k$-local charges for all $5\le k\le N/2$, but no 4-local charge.
A representative example of the on-site term is
\eq{
g = c_{30} (\bd)^3 + c_{11} n + c_{01} b
}
with $c_{30} \neq 0$, $c_{11} \neq 1$, and $c_{01} \neq 0$.
This model admits the following 3-local and 5-local charges:
\eq{
Q_3 &= \sumi \bd_i b_{i+2} - b_i \bd_{i+2} \nx
\whb + 3c_{30}\left( \bd_i (\bd_{i+1})^2 + (\bd_i)^2 \bd_{i+1} \right) \nx
\whb + (3c_{11}-2) \left( \bd_i b_{i+1} - b_i \bd_{i+1} \right),}
\eq{
Q_5 &= \sumi \bd_i b_{i+4} - b_i \bd_{i+4} \nx
\whb + 3c_{30} \left( \bd_i (\bd_{i+3})^2 - (\bd_i)^2 \bd_{i+3} \right) \nx
\whb - 6c_{30} \left( \bd_i \bd_{i+2} \bd_{i+3} - \bd_i \bd_{i+1} \bd_{i+3} \right) \nx
\whb + 2 (3c_{11}-2) \left( \bd_i b_{i+3} - b_i \bd_{i+3} \right) \nx
\whb + 3c_{30} (3c_{11}-2) \left( \bd_i (\bd_{i+2})^2 - (\bd_i)^2 \bd_{i+2} \right) \nx
\whb -3c_{30} (9c_{11}^2 - 6c_{11} - 1) \left( \bd_i (\bd_{i+1})^2 - (\bd_i)^2 \bd_{i+1} \right) \nx
\whb - (27c_{11}^3 - 36c_{11}^2 +6c_{11} + 4) \left(\bd_i b_{i+1} - b_i \bd_{i+1}\right).
}
Accordingly, $k$-local charges for $k\ge 6$ also exist, although we do not display their explicit forms here.
Nevertheless, no 4-local charge exists.

Note that this is the simplest representative. In general, Type~C$^-$ is characterized by the on-site term of the form
\eq{
g=c_{30}(\bd)^3+c_{20}(\bd)^2+c_{10}\bd +c_{11}n + c_{01}b
}
with $c_{30} \neq 0$, $c_{01} \neq \frac{c_{20}(c_{11}+2)}{3c_{30}}$, and $c_{11} \neq 1$, or by its adjoint.

This model is ``almost'' completely integrable in the sense that it possesses the full hierarchy of local charges except for the single missing 4-local charge.
This is the first report of such a near-integrable yet genuinely partially integrable behavior, together with our Type~N$^+$ example in the previous subsection.

We show how to construct $k$-local charges for $k=3$ and $5\leq k\leq N/2$ in \lref{degree3-C-} (\ssref{model-3}).
We prove the absence of 4-local charges in \lref{degree3-4local} (\ssref{model-3}).

\subsection{New completely integrable systems}
We begin with two simple completely integrable examples in our class:
\eq{
g &= c_{40}(\bd)^4, \\
g &= c_{30}(\bd)^3 + c_{11} n.
}
Each of these models belongs to Type~C, i.e., it admits a full hierarchy of $k$-local charges for all $3\le k\le N/2$.

In general, completely integrable (Type~C) models in the present class are precisely those with on-site terms $g$ of the following four forms:
\eq{
g &= c_{40}(\bd)^4+c_{30}(\bd)^3+c_{20}(\bd)^2+c_{10}\bd+\frac{c_{30}}{2c_{40}}b, 
\lbl{catalog-c-1}\\
g &= c_{30}(\bd)^3+c_{20}(\bd)^2+c_{10}\bd +c_{11}n + \frac{c_{20}(c_{11}+2)}{3c_{30}}b , 
\lbl{catalog-c-2}\\
g &= c_{30}(\bd)^3 + c_{20}(\bd)^2 + c_{10}\bd + n + c_{01}b, 
\lbl{catalog-c-3}\\
g &= c_{20} (\bd)^2 + c_{10} \bd + c_{11} n + c_{01} b + c_{02} b^2,
\lbl{catalog-c-4}
}
together with their adjoints.
To the best of our knowledge, the models with on-site terms \eqref{catalog-c-1}--\eqref{catalog-c-3} are new integrable systems reported here for the first time.

The significance of this discovery is methodological.
Traditionally, integrable systems were often discovered in a top-down manner: one first finds an explicit solution (for instance via an $R$-matrix), and then derives the associated hierarchy of local charges.
By contrast, our approach starts from the local charges themselves and identifies integrable Hamiltonians bottom-up.

This bottom-up perspective builds on a framework introduced in 2019 for direct analysis of $k$-local charges without relying on any solvable structure \cite{shiraishi2019proof}.
So far, this direct-analytic approach has been used to (i) prove nonintegrability by showing the absence of local charges
%arXiv-order
\cite{chiba2024proof,park2025nonintegrability,park2025proof,park2025graphtheoretical,shiraishi2024absence,chiba2024exact,yamaguchi2024complete,yamaguchi2024proof,hokkyo2025absence,shiraishi2024s1,chiba2025proof,shiraishi2025complete,hokkyo2025rigorous,futami2025absence1,shiraishi2025dichotomy,futami2025absence,
%hokkyo2025integrability,
fan2025absence,fan2026minimal},
or (ii) obtain explicit expressions for local charges in already-known integrable models
\cite{nozawa2020explicit,fukai2023all,fukai2024proof}.
Here we show that this approach can go further: it can also be used to discover new integrable systems.

We describe how to construct these local charges in \lref{degree4-C} (\ssref{model-4}) and in \lref{degree3-Ca} and \lref{degree3-Cb} (\ssref{model-3}).
Understanding whether these models admit an explicit solution (e.g., an $R$-matrix description) is left as a future research.

\subsection{Even--odd sensitivity from staggered charges}
There exist models that lie at the intersection of Type~N and Type~SC: they admit no nontrivial uniform (translation-invariant) local charges, yet they do admit staggered $k$-local charges for all $k$.
A representative example is
\eq{
g = c_{50} (\bd)^5,
}
for which the staggered $k$-local charges can be written as
\eq{
Q_k &= \sumi (-1)^i \bd_i \bd_{i+k-1}.
}

All previously known integrable models have uniform charges; whether they also have staggered ones depends on the model.
Our example exhibits an anomalous even--odd behavior under periodic boundary conditions: staggered charges exist only for even $N$, meaning that the model has infinitely many local charges when $N$ is even but none when $N$ is odd.
This is the first example of a model in which the local-charge structure depends drastically on the parity of the system size.

\subsection{The Bose-Hubbard model (standard nonintegrable example)}
We finally comment on the conventional Bose--Hubbard model. In our notation, its on-site term consists of a two-body interaction and a chemical potential,
\eq{
g = \und{c_{22}}{=U/2} n(n-1) + \und{c_{11}}{=-\mu} n.
}
In our classification, for $U\neq 0$ (i.e., $c_{22}\neq 0$) the model falls into Type~N in the uniform sector and Type~SN in the staggered sector
(see Tables~\ref{table-classify-I} and \ref{table-classify-II}). It admits no nontrivial local charges beyond the obvious ones ($H$ and the total particle number $\sum_i n_i$).
Integrability is recovered only in special limits such as $U=0$ (free bosons) or the hard-core limit.

Historically, Haldane proposed a Bethe-ansatz-type solution for this model \cite{haldane1980solidification}, but Choy and Haldane later showed that this construction fails \cite{haldane1981erratum,choy1982failure}.
Subsequent numerical studies found quantum-chaotic spectral statistics \cite{kolovsky2004quantum}, in line with nonintegrable behavior.
In this paper, we settle this issue within our framework by proving the absence of nontrivial local charges for $U\neq 0$.
\section{Proof part 1: General treatment}
\lbl{s:proof1}

The rest of this article is devoted to proving \mref.
The purpose of the present section is to establish a general framework to analyze local commuting charges: we first introduce notation and an expansion for local quantities, and then derive model-independent constraints on the possible form of $k$-local charges (\lref{step1} and \lref{step2}).

The subsequent sections (\sref{proof2} and \sref{proof3}) complete the proof by combining this general framework with model-specific analyses.
In \sref{proof2}, we perform a complete classification of Hamiltonians \eqref{H-gen} according to the presence or absence of a 3-local charge (\lref{3local}).
In \sref{proof3}, we (i) determine the full local-charge structure of the models that do admit a 3-local charge (including the partially integrable Types~N$^+$ and~C$^-$), and (ii) prove that the absence of a 3-local charge implies the absence of any $k$-local charge for all $3\le k\le N/2$ (\lref{absence3k}), placing all remaining models into Type~N.
The staggered-sector charges are treated separately in \apref{staggered}.

\subsection{Main idea of the proof}

A key idea in our analysis is to expand local quantities in a suitable basis and to rewrite the commutativity condition $[Q,H]=0$ as a system of linear equations for the expansion coefficients of $Q$.
This approach was originally developed as a direct method to prove the absence of nontrivial local charges
\cite{shiraishi2019proof}, though related ideas for the special case of 3-local charges also appeared earlier
\cite{grabowski1995integrability}.
While most applications have focused on spin systems, in this paper we adapt this method to bosonic chains.

Note that we use the term \emph{quantity} $Q$ without assuming $[Q,H]=0$ and reserve the term \emph{charge} for commuting ones.
A \emph{$k$-local quantity} is defined in a similar way to a $k$-local charge in \sref{setup}.

We choose a single-site operator basis, denoted by
\eq{
\op{xy} := (\bd)^x b^y ,
}
where $x$ and $y$ are nonnegative integers (we also write this as $\op{x,y}$).
Using this basis, we expand a $k$-local quantity $Q_k$ (a candidate for a $k$-local charge) as
\eq{
Q_k = \sum_{\bA; |\bA| \leq k} q_{\bA} \bA.
}
Here $\bA$ is a symbol sequence of $\op{xy}$'s equipped with spatial positions, whose initial and final symbols are not $\op{00}=I$, and $|\bA|$ is the length of the symbol sequence.
Examples are $\op{20}_3\op{10}_4\op{01}_5=(\bd_3)^3\bd_4b_5$ and $\op{31}_2\op{00}_3\op{22}_4=(\bd_2)^3b_2(\bd_4)^2b_4^2$.
The basis of $k$-local quantities is the set of all symbol sequences $\bA$ with $|\bA|\leq k$.

We also expand the commutator $[Q_k,H]$ in the same basis,
\eq{
[Q_k,H] = \sum_{\bB; |\bB| \leq k+1} r_{\bB} \bB.
}
Here we used the fact that the commutator of a $k$-local quantity and the 2-local Hamiltonian is at most $k+1$-local.
Since the commutator is linear in $Q_k$, each coefficient $r_{\bB}$ is a linear function of the unknowns $\{q_{\bA}\}$.
Thus the commutativity condition $[Q_k,H]=0$ is equivalent to imposing $r_{\bB}=0$ for all $\bB$, which yields a system of linear equations for $\{q_{\bA}\}$.

We solve these equations in a suitable order.
If a nontrivial solution exists, it produces a nontrivial local charge.
If the only solution has $q_{\bA}=0$ for all $|\bA|=k$, this proves the absence of $k$-local charges.

Concretely, we incorporate the constraints $r_{\bB}=0$ in decreasing order of $|\bB|$, from $|\bB|=k+1$ down to $|\bB|=1$.
In Step $s$ ($1\leq s \leq k+1$), we require that the $(k+2-s)$-local part of $[Q_k,H]$ vanishes, i.e., $r_{\bB}=0$ for all $|\bB|=k+2-s$.
$Q_k$ satisfies all conditions up to Step $s$ if and only if
\eq{
[Q_k , H] = \lo{[k+1-s]}
\lbl{step-iff}
}
where $\lo{l}$ denotes an at-most-$l$ local quantity.
If no $Q_k$ satisfies this condition at some step, it immediately follows that no $k$-local charge exists.

In this section we carry out Step~1 (analysis of $r_{\bB}=0$ with $|\bB|=k+1$) and Step~2 (analysis of $r_{\bB}=0$ with $|\bB|=k$), which are model-independent and hold for most onsite terms $g$ in \eref{g-expand}. A few exceptional cases are treated separately.

\subsection{Symbols and useful relations}

We introduce some notation used throughout the rest of this paper.

We denote the commutator of two operators $\alpha$ and $\beta$ by
\eq{
\alpha \co \beta := [\alpha,\beta].
}
If the order of nested commutators does not need to be distinguished, that is, if $(\alpha \co \beta) \co \gamma = \alpha \co (\beta \co \gamma)$ holds in the situation at hand, we may omit parentheses and write simply $\alpha \co \beta \co \gamma$.
The same convention applies to nested commutators involving four or more operators.

We define the following nested commutators built from the hopping terms:
\eq{
 h_{1:k} := \bd_1 b_k + (-1)^k b_1 \bd_k
 = h_{1:2} \co h_{2:3} \co \cdots \co h_{k-1:k},
}
and similarly
\eq{
 \hb_{1:k} := \bd_1 b_k - (-1)^k b_1 \bd_k.
}
Here, the subscript $1:k$ indicates that its contiguous support is $1 \le i \le k$.

Let $T$ denote the translation superoperator by one site, and let $\sums$ denote the translation sum.
For example, $\sums a_1 := \sumi T^i a_1 = \sumi a_i$.
With this notation, \eref{H-gen} is simply written as
\eq{
H = \sums (h_{1:2} + g_1).
}

Finally, we record a useful commutator identity.
We define the degree of $\op{xy}$ as $x+y$, and the degree of $g$ as
\eq{
\deg(g) := \sup\{x+y\,|\,c_{xy} \neq 0\}.
}
Then the commutator between $\op{xy}$ and $\op{x'y'}$ can be written as
\begin{widetext}
\eq{
\op{xy} \co \op{x'y'}
&= \sum_{1 \le z \le \min(x+x',y+y')}
 z! \left(
\begin{pmatrix} y \\ z \end{pmatrix}
\begin{pmatrix} x' \\ z \end{pmatrix} -
\begin{pmatrix} x \\ z \end{pmatrix}
\begin{pmatrix} y' \\ z \end{pmatrix}
\right)
\op{x+x'-z,y+y'-z} \nx
&= (yx'-xy') \op{x+x'-1,y+y'-1} + (\text{lower terms}) \lbl{comxy-formula}
 ,
}
\end{widetext}
where the degree of the lower terms is at most $x+x'+y+y'-2$.
In many arguments below, we will only need the leading term in \eref{comxy-formula}.

\subsection{Analysis of Step~1}
\lbl{s:step1}

In Step~1, we analyze relations $r_{\bB}=0$ with $|\bB|=k+1$ and derive some constraints on a parameter family of ${q_{\bA}}$.
The conclusion of Step~1 analysis is summarized in the following Lemma:

\begin{lem}[Step~1 analysis]\lbl{t:step1}
Let $k$ be an integer with $2 \leq k \leq N/2$, and let $Q_k$ be a $k$-local quantity.
Then, the following two conditions are equivalent.
\begin{itemize}
\item $[Q_k, H]$ is an at-most-$k$-local quantity.
\item $Q_k$ can be written as
\eq{
Q_k &= \sumi ( q_{+-} \bd_i b_{i+k-1} + q_{-+} b_i \bd_{i+k-1}) \nx
\wha + \sumi (-1)^i (q_{++} \bd_i \bd_{i+k-1} + q_{--} b_i b_{i+k-1} ) \nx
\wha + \lo{k-1} \lbl{Qk-step1}
}
for some constants $q_{+-}$, $q_{-+}$, $q_{++}$, and $q_{--}$.
\end{itemize}
\end{lem}

We postpone the proof to \apref{proof-step1} and here explain the idea behind the proof.

We suppose that $[Q_k, H]$ is an at-most-$k$-local quantity.
A simple but important observation is the following.
Suppose a $k+1$-local operator $\bB$ is generated from some commutator between a $k$-local term in $Q_k$ and either $\bd_i b_{i+1}$ or $b_i \bd_{i+1}$.
Then the same $\bB$ must also be generated by another commutator contribution.
Otherwise, $\bB$ would appear with a nonzero coefficient $r_{\bB}$ and $[Q_k, H]$ would contain $k+1$-local terms, contrary to the assumption.

Owing to this, we first find that two ends of any $k$-local term in $Q_k$ must be $b$ or $\bd$.
Consider the case $k=3$ as an example.
If $Q_3$ has $b^2_1 \bd_2 \bd_3$ with a nonzero coefficient $q_{b^2_1 \bd_2 \bd_3}$, then $\bB = b^2_1 \bd_2 \bd_4$ is generated only by $[b^2_1 \bd_2 \bd_3, b_3 \bd_4]$, and thus $[Q_3, H]$ becomes a $4$-local quantity, which contradicts the supposition.

Through a similar reasoning, we find that all operators other than the two ends of any $k$-local term in $Q_k$ must be identities.
Consider again the case $k=3$ as an example.
If $Q_3$ contains $b_1 b_2 \bd_3$ with a nonzero coefficient $q_{b_1 b_2 \bd_3}$, then the $4$-local operator $\bB' = b_1 b_2 \bd_4$ in $[Q_3, H]$ is generated by exactly two commutators:
$[b_1 b_2 \bd_3, b_3 \bd_4]$ and $[b^2_2 \bd_4, b_1\bd_2]$.
Therefore, in order for the coefficient $r_{\bB'}$ to vanish, the coefficients of $b_1 b_2 \bd_3$ and $b^2_2 \bd_4$ in $Q_3$ must be proportional.
However, as we have already shown, $b^2_2 \bd_4$ has zero coefficient in $Q_3$, and hence $b_1 b_2 \bd_3$ must also have zero coefficient, contrary to the assumption.
Similar arguments confirm the desired result.

These assertions impose that a possible form of the $k$-local terms of $Q_k$ is \eref{Qk-step1}.

\subsection{Symmetry argument}
\lbl{s:symmetry}

Before proceeding to Step~2, we use the translation and inversion symmetries of \eref{H-gen} to reduce the analysis of the existence of local charges to symmetry sectors.

Since the commutator superoperator $[\bullet,H]$ commutes with the translation superoperator $T$, any charge can be decomposed into a linear combination of $T$-eigenoperator charges.
Hence it suffices to analyze local charges in each $T$-sector.
From \eref{Qk-step1}, the $T$-eigenvalue can only be $\pm 1$.
We call the $T=+1$ sector the \emph{uniform sector} and the $T=-1$ sector the \emph{staggered sector}.
In terms of \eref{Qk-step1}, $q_{++}=q_{--}=0$ corresponds to the uniform sector, while $q_{+-}=q_{-+}=0$ corresponds to the staggered sector.

Within the uniform sector ($T=+1$), the superoperators $[\bullet,H]$, $T$, and spatial inversion $P$ commute.
Thus we may decompose the operator space into $P$ sectors.
Equivalently, it suffices to consider only uniform inversion-symmetric charges and uniform inversion-antisymmetric charges.
A uniform inversion-symmetric charge takes the form
\eq{
Q_k &= \sums \bd_1 b_k + b_1 \bd_k + \lo{k-1} \lbl{Qk-uniform-sym}
}
and a uniform inversion-antisymmetric charge takes the form
\eq{
Q_k &= \sums \bd_1 b_k - b_1 \bd_k + \lo{k-1} \lbl{Qk-uniform-anti}.
}
Here we normalize $Q_k$ so that the leading $k$-local terms have coefficient $\pm1$.

We treat the uniform sector in the main text (setting $q_{++}=q_{--}=0$) and the staggered sector in \apref{staggered}.

\subsection{Analysis of Step~2}
\lbl{s:step2}

In Step~2 we analyze the relations $r_{\bB}=0$ with $|\bB|=k$ and derive further constraints on a parameter family of $q_{\bA}$, as in Step~1.
Although these constraints do not depend on the detailed coefficients of $g$, we will distinguish cases according to the degree $\deg(g)$.
The low-degree cases $\deg(g)\le 2$ are treated separately in \apref{degree2}.

The Step~2 analysis, together with Step~1, yields the following lemma, whose proof is presented in \apref{proof-step2}.

\begin{lem}[Step~2 analysis]\lbl{t:step2}
Let $k$ be an integer with $3 \leq k \leq N/2$, and let $Q_k$ be a uniform $k$-local quantity that is either inversion-symmetric or inversion-antisymmetric.
Then, the following two conditions are equivalent for $g$ with $\deg(g)\geq 3$.
\begin{itemize}
\item $[Q_k,H]$ is an at-most-$k-1$-local quantity.
\item $Q_k$ can be written, up to an overall normalization, as
\eq{
Q_k &= \sums \left( h_{1:k} + \sum_{0 \leq m \leq k-1} h_{1:m} \co g_m \co h_{m:k-1} + t \hb_{1:k-1} \right) \nx
\wha +\lo{k-2}
\lbl{Qk-step2}
}
for some constant $t$.
\end{itemize}
\end{lem}

Here, we promise that superoperators $h_{1:1} \co$ and $\co h_{k-1:k-1}$ mean an identity operator.
Examples are
\eq{
Q_3 &= \sums h_{1:3} + g_1 \co h_{1:2} + h_{1:2} \co g_2 + t\hb_{1:2} + \lo{1} , \lbl{exampleQ3}\\
Q_4 &= \sums h_{1:4} + g_1 \co h_{1:3} + h_{1:2} \co g_2 \co h_{2:3} + h_{1:3} \co g_3 \nx \whb + t\hb_{1:3} + \lo{2} .
\lbl{exampleQ4}
}
Recall that we restrict attention to the uniform sector here.

\lref{step2} implies that a $k$-local charge $Q_k$ should be inversion symmetric for even $k$ and inversion antisymmetric for odd $k$ in the case of $g$ with $\deg(g)\geq 3$.
Indeed, the leading $k$-local term $h_{1:k}$ in \eref{Qk-step2} is $\bd_1 b_k + b_1 \bd_k$ for even $k$ and $\bd_1 b_k - b_1 \bd_k$ for odd $k$, which matches the leading terms in \eref{Qk-uniform-sym} and \eref{Qk-uniform-anti}, respectively.
\section{Proof part 2: Complete analysis of 3-local charges}
\lbl{s:proof2}

In this section, we focus on $k=3$ and provide a complete classification of all Hamiltonians of the form of \eref{H-gen} according to the presence or absence of 3-local charges (\lref{3local}).
In fact, as shown later in \lref{absence3k} (\ssref{GMC1}), all models without 3-local charges are proven to be non-integrable (i.e., they admit no $k$-local charge for any $3\leq k\leq N/2$).
Models with 3-local charges are candidates for complete integrability.
By explicitly constructing local charges, we show in \ssref{model-4} and \ssref{model-3} that some of them are completely integrable whereas others are only partially integrable.

\subsection{Criterion for the existence of 3-local charges}
\lbl{s:proof2-1-local}

As \eref{exampleQ3} shows, a candidate 3-local charge $Q_3$ must take the form
\eq{
Q_3 = \sums h_{1:3} + g_1 \co h_{1:2} + h_{1:2} \co g_2 + t \hb_{1:2}+s_1 ,
\lbl{Q3-finform}
}
where $t$ is an a priori unknown scalar coefficient, and $\sums s_1$ denotes the (a priori unknown) 1-local part of $Q_3$.

Since this choice of $Q_3$ already satisfies Steps~1--2, we proceed to Step~3 and analyze the 2-local terms in $[Q_3,H]$.
Restricting to the 2-local terms, we compute $[Q_3,H]$ as
\eq{
[Q_3 , H]
&= \sums h_{1:3} \co H \nx
\whb + (g_1 \co h_{1:2} + h_{1:2} \co g_2 + t \hb_{1:2}) \co H \nx
\whb + \und{s_1 \co H}{=s_1 \co h_{1:2} - h_{0:1} \co s_1} + \lo{1} \nx
&= \sums (g_1 \co h_{1:2}) \co g_1 + (g_1 \co h_{1:2}) \co h_{1:2} \nx
\whb + (h_{1:2} \co g_2) \co g_2 + (h_{1:2} \co g_2) \co h_{1:2} \nx
\whb + t (-g_1 \co \hb_{1:2} + \hb_{1:2} \co g_2) \nx
\whb + t \und{\hb_{1:2} \co h_{1:2}}{=\lo{1}} \nx
\whb + s_1 \co h_{1:2} - h_{1:2} \co s_2 + \lo{1}.
}
The Step~3 condition means $[Q_3, H]=\lo{1}$, so the above relation on sites 1 and 2 reads 
\eq{
\wha 
(g_1 \co h_{1:2}) \co g_1 
+ (g_1 \co h_{1:2}) \co h_{1:2} \nx
\wha
+ (h_{1:2} \co g_2) \co g_2
+ (h_{1:2} \co g_2) \co h_{1:2} \nx
\wha 
+ t ( -g_1 \co \hb_{1:2} + \hb_{1:2} \co g_2) 
+ s_{1} \co h_{1:2} 
- h_{1:2} \co s_{2} \nx
&= \lo{1}. \lbl{3local-eq-pre}
}

By the inversion antisymmetry of $Q_3$, its 1-local part must vanish, so we have $s_1=0$.
Moreover, since $H$ is inversion symmetric, $[Q_3,H]$ is also inversion-antisymmetric and hence has no 1-local term; therefore, Step~3 is the last nontrivial step in the analysis.
If there exists $t$ satisfying \eref{3local-eq-pre}, then the Hamiltonian has a 3-local charge (obtained by substituting this $t$ into \eref{Q3-finform}); otherwise, it has no 3-local charge.

Therefore, the existence of a 3-local charge is equivalent to the solvability in $t$ of
\eq{
\wha 
\und{\und{(g_1 \co h_{1:2}) \co g_1}{\text{n$_1$ term}} + \und{(h_{1:2} \co g_2) \co g_2}{\text{n$_2$ term}}}{\text{n terms}} \nx
\wha
+ \und{\und{(g_1 \co h_{1:2}) \co h_{1:2}}{\text{l$_1$ term}}
+ \und{(h_{1:2} \co g_2) \co h_{1:2}}{\text{l$_2$ term}}}{\text{l terms}} \nx
\wha 
\und{\und{-tg_1 \co \hb_{1:2}}{\text{t$_1$ term}}
+\und{t \hb_{1:2} \co g_2}{\text{t$_2$ term}}}{\text{t terms}} \nx
&= \lo{1} .
\lbl{3local-eq}
}
Here, the first two terms of \eref{3local-eq} are nonlinear in $g$ (the \emph{n-terms}), the next two are linear in $g$ (the \emph{l-terms}), and the last two contain the unknown parameter $t$ (the \emph{t-terms}).
Note that, although $[Q_3,H]$ has no 1-local part, the right-hand side of \eref{3local-eq} need not vanish, because the shift sum of nontrivial 1-local terms may cancel.
For instance, $\bd_1-\bd_2$ is a nonzero 1-local operator whose shift sum vanishes.

\subsection{Classification by the existence of a 3-local charge}
\lbl{s:proof2-classify}

We now analyze \eref{3local-eq} and determine for which Hamiltonians there exists a solution for $t$.
We first summarize the result.

\begin{lem}[Complete list of models with 3-local charges]\lbl{t:3local}
A Hamiltonian in the form of \eref{H-gen} has a 3-local charge if and only if $g$ is of one of the following forms, or any Hermitian conjugate thereof:
\begin{enumerate}[label=(\roman*)]
\item $g = c_{40} (\bd)^4 + c_{30} (\bd)^3 + c_{20} (\bd)^2 +c_{10} \bd + c_{11} n + c_{01} b$, with $c_{30}(c_{11}+2) = 4c_{40}c_{01}$;
\item $g = c_{30} (\bd)^3 + c_{20} (\bd)^2 +c_{10} \bd + c_{11} n + c_{01} b$;
\item $g = c_{20} (\bd)^2+ c_{10} \bd + c_{11} n + c_{01} b + c_{02} b^2$.
\end{enumerate}
\end{lem}

This lemma implies that 3-local charges appear only in Hamiltonians of the form~\eqref{H-gen} with special $g$ of degree $4$ or $3$, and with arbitrary $g$ of degree at most $2$.
We shall provide explicit 3-local charges in \eref{explicit-4} and \eref{explicit-3} (\ssref{proof2-displace}).

\begin{proof}
The if part is immediate by substituting the above forms of $g$ into \eref{3local-eq}.

We now prove the only-if part.
For brevity, we treat here only the case of $g$ with $\deg(g)\geq 5$; the remaining cases $\deg(g)\leq 4$ are treated separately in \apref{proof-3local}.

In this case, there exist $x$ and $y$ such that $c_{xy} \neq 0$ and $x+y = \deg(g) \geq 5$.
The contribution of the l$_2$-term (see \eref{3local-eq}) of $\op{xy}$ is calculated as
\begin{widetext}
\eq{
\wha (h_{1:2} \co c_{xy} \op{xy}_2) \co h_{1:2} \nx
&= (( \op{10}_{1} \op{01}_2 + \op{01}_1 \op{10}_2 ) \co c_{xy} \op{xy}_2) \co (\op{10}_{1} \op{01}_2 + \op{01}_1 \op{10}_2 ) \nx
&= c_{xy} (x \cdot \op{10}_1 \op{x-1,y}_2 - y \cdot \op{01}_1 \op{x,y-1}_2) \co (\op{10}_{1} \op{01}_2 + \op{01}_1 \op{10}_2 ) \nx
&= c_{xy}  (x\cdot (-(x-1) \cdot \op{20}_1 \op{x-2,y}_2 + y \cdot \op{11}_1 \op{x-1,y-1}_2)  -y\cdot ( -x \cdot \op{11}_1 \op{x-1,y-1}_2 +(y-1) \cdot \op{02}_1 \op{x,y-2}_2)) + \lo{1} \nx
&= c_{xy}  (-x(x-1) \op{20}_1 \op{x-2,y}_2 + 2xy \op{11}_1 \op{x-1,y-1}_2 - y(y-1) \op{02}_1 \op{x,y-2}_2) +\lo{1} ,
\lbl{cal-l2}
}
\end{widetext}
where we used the general commutator identity $[ \alpha_1 \beta_2 , \gamma_1 \delta_2 ] = \half [\alpha_1,\gamma_1]\{\beta_2,\delta_2\} + \half\{\alpha_1,\gamma_1\}[\beta_2,\delta_2]$ and the anticommutation relations $\{\op{10},\op{10}\} = 2\cdot \op{20}$, $\{\op{01},\op{01}\} = 2\cdot \op{02}$, and $\{\op{10},\op{01}\} = 2\cdot \op{11} + \op{00}$.
Note that in the last line of \eref{cal-l2}, the first term vanishes for $x=0$ or $1$, the second term vanishes for $x=0$ or $y=0$, and the last term vanishes for $y=0$ or $1$.

Since $Q_3$ is a charge, the coefficients of
$\op{20}_1\op{x-2,y}_2$, $\op{11}_1\op{x-1,y-1}_2$, and $\op{02}_1\op{x,y-2}_2$
in the expansion of $[Q_3,H]$ must vanish.
We claim that these three operators can only come from the l$_2$-term of $\op{xy}$.
In other words, there are no contributions from the n-terms, the t-terms, the l$_1$-term, or the l$_2$-terms of $\op{x'y'}$ with $(x',y')\neq(x,y)$.

The reason is a degree count.
Each of the three operators has degree $2$ on site~1 and degree at least $3$ on site~2 (since $x+y\geq 5$).
By contrast, every n-term and every t-term has degree $1$ on one of the sites (see Eqs.~\eqref{formula_n1}, \eqref{formula_n2}, \eqref{formula_t1}, and \eqref{formula_t2}), hence they cannot appear.
Exchanging sites~1 and~2 in \eref{cal-l2} shows that the l$_1$-term has degree $2$ on site~2, not on site~1, and is therefore excluded as well.

It remains to consider the l$_2$-terms of $\op{x'y'}$ with $(x',y')\neq(x,y)$.
In the l$_2$-term of $\op{x'y'}$, each operator $\op{pq}_1\op{rs}_2$ satisfies $p+r=x'$ and $q+s=y'$.
Hence these l$_2$-terms cannot contain
$\op{20}_1\op{x-2,y}_2$, $\op{11}_1\op{x-1,y-1}_2$, or $\op{02}_1\op{x,y-2}_2$, unless $(x',y')=(x,y)$.

It follows that a 3-local charge exists only if all of
\begin{alignat}{2}
r_{\op{20}_1\op{x-2,y}_2} &=& -x(x-1)c_{xy} &= 0, \\
r_{\op{11}_1\op{x-1,y-1}_2} &=& 2xyc_{xy} &= 0, \\
r_{\op{02}_1\op{x,y-2}_2} &=& -y(y-1)c_{xy} &= 0
\end{alignat}
are satisfied.
However, for $x+y\geq 5$ with $c_{xy}\neq 0$, at least one of $-x(x-1)$, $2xy$, and $-y(y-1)$ is nonzero, so the above three equations cannot be satisfied simultaneously.
This establishes the absence of 3-local charges for $g$ with $\deg(g)\geq 5$.
\end{proof}

\subsection{Simplifying the Hamiltonian by a displacement transformation}
\lbl{s:proof2-displace}

\lref{3local} shows that the remaining admissible $g$ still carry many terms.
To handle the models more conveniently, we transform Hamiltonians using the displacement operator $D(z):= \exp(z \bd - z^* b)$ with an arbitrary complex parameter $z$.
The displacement operator $D(z)$ is a unitary operator that transforms the creation and annihilation operators as
\eq{
D(z) \bd D^\dagger (z) &= \bd+ z^*, \\
D(z) b   D^\dagger (z) &= b + z .
}
Note that the global displacement $H \mapsto D(z)^{\otimes N} H D^\dagger(z)^{\otimes N}$ does not change the presence or absence of $k$-local charges.

By a suitable global displacement, we can eliminate some coefficients of $g$.
First, recalling the degree-$4$ case (i) in \lref{3local}, we may write the Hamiltonian as
\eq{
H &= \sums h_{1:2} + c_{40} (\bd_1)^4 + c_{30} (\bd_1)^3 + c_{20} (\bd_1)^2 \nx
\whb + c_{10} \bd_1 + c_{11} n_1 + c_{01} b_1
\lbl{model4-raw}
}
with $c_{40} \neq 0$ and $c_{30}(c_{11}+2) = 4c_{40}c_{01}$.
We apply a displacement operator with $z = -\frac{c_{30}}{4c_{40}} = -\frac{c_{01}}{c_{11}+2}$ to remove the terms with coefficients $c_{30}$ and $c_{01}$, which yields
\eq{
H &= \sums h_{1:2} + c_{40} (\bd_1)^4 + c_{20} (\bd_1)^2 + c_{10} \bd_1 + c_{11} n_1.
\lbl{model4}
}
Note that, under this transformation, $c_{40}$ and $c_{11}$ remain unchanged, whereas $c_{20}$ and $c_{10}$ may change.

Thus, without loss of generality, we consider the Hamiltonian \eqref{model4} instead of the original Hamiltonian \eqref{model4-raw}.
We now write down the corresponding 3-local charge, given by
\eq{
Q_3 &= \sums \bd_1 b_3 - b_1 \bd_3 + 4 c_{11} \left( \bd_1 b_2 -  b_1 \bd_2 \right) \nx
\whb - c_{40} \left( (\bd_1)^3 \bd_2 - \bd_1 (\bd_2)^3 \right).
\lbl{explicit-4}
}

Likewise, with the degree-$3$ case (ii) in \lref{3local} in mind, a general Hamiltonian with a 3-local charge reads
\eq{
H = \sums h_{1:2} + c_{30} (\bd_1)^3 + c_{20} (\bd_1)^2 + c_{10} \bd_1 + c_{11} n_1 + c_{01} b_1
\lbl{model3-raw}
}
with $c_{30} \neq 0$.
Applying a displacement operator with $z=-\frac{c_{20}}{3c_{30}}$, we eliminate $c_{20}$, yielding
\eq{
H = \sums h_{1:2} + c_{30} (\bd_1)^3 + c_{10} \bd_1 + c_{11} n_1 + c_{01} b_1.
\lbl{model3}
}
Note that, under this transformation, $c_{30}$ and $c_{11}$ do not change, whereas $c_{10}$ and $c_{01}$ may change.
More explicitly, $c_{01}$ is transformed as $c_{01} \gets c_{01} - \frac{c_{20}(c_{11}+2)}{3c_{30}}$.

Thus, without loss of generality, we consider the Hamiltonian \eqref{model3} instead of the original Hamiltonian \eqref{model3-raw}.
The corresponding 3-local charge is given by
\eq{
Q_3 &= \sums \bd_1 b_3 - b_1 \bd_3  - 3 c_{30} \left( (\bd_1)^2 \bd_2 - \bd_1 (\bd_2)^2 \right) \nx
\whb + (3 c_{11} - 2) \left( \bd_1 b_2 - b_1 \bd_2 \right).
\lbl{explicit-3}
}
\section{Proof part 3: Model-specific treatments}
\lbl{s:proof3}

\subsection{Models of degree 4}
\lbl{s:model-4}

In this subsection, we study the Hamiltonian \eqref{model4}, which is obtained from \eqref{model4-raw} by a displacement transformation.
This is a degree-4 Hamiltonian with a 3-local charge, and we determine the existence of $k$-local charges for all $k\geq 4$.
In brief, the result depends on $c_{11}$: if $c_{11}=0$, $k$-local charges exist for all $k$ (\lref{degree4-C}); if $c_{11} \neq 0$, no $k$-local charge exists for any $k\geq 4$ (\lref{degree4-N+}).
Accordingly, \eqref{model4} is Type~C for $c_{11}=0$ and Type~N$^+$ for $c_{11} \neq 0$.

\subsubsection{Case of $c_{11}=0$ (Type~C)}

\begin{lem}\lbl{t:degree4-C}
The Hamiltonian \eqref{model4-raw} with $c_{11}=0$ has a $k$-local charge for any $3 \leq k \leq N/2$.
\end{lem}

\begin{proof}
The displaced Hamiltonian is \eref{model4} with $c_{11}=0$, namely,
\eq{
H &= \sums h_{1:2} + c_{40} (\bd_1)^4 + c_{20} (\bd_1)^2 + c_{10} \bd_1.
}
We explicitly construct a $k$-local charge as
\eq{
Q_k &= \sums h_{1:k} + h_{1:k-2} + (1+(-1)^k) (c_{20} \bd_1 \bd_{k-1} + c_{10} \bd_1) \nx
\wha + c_{40} \sum_{0 \leq n \leq [k/2]-1} D_{k-n-1,n},
\lbl{exact-degree4-C}
}
where $D_{l,n}$ ($0\leq n\leq l-1$, $l\geq 2$) is defined as
\eq{
D_{l,n} := 24 \sums \sum_{1 \leq m \leq l-n} (-1)^{l-n-m} \sigma_{1,m,n+m,l} \bd_1 \bd_m \bd_{n+m} \bd_l
\lbl{D-def}
}
with the symmetry factor $\sigma_{i_1,i_2,i_3,i_4}$ for integers $i_1 \leq i_2 \leq i_3 \leq i_4$ given by
\eq{
\sigma_{i_1,i_2,i_3,i_4} := 
\begin{cases}
1 & (i_1 < i_2 < i_3 < i_4) \\
\tfrac{1}{2!} & (i_1 = i_2 < i_3 < i_4) \\
\tfrac{1}{2!} & (i_1 < i_2 = i_3 < i_4) \\
\tfrac{1}{2!} & (i_1 < i_2 < i_3 = i_4) \\
\tfrac{1}{2!2!} & (i_1 = i_2 < i_3 = i_4) \\
\tfrac{1}{3!} & (i_1 = i_2 = i_3 < i_4) \\
\tfrac{1}{3!} & (i_1 < i_2 = i_3 = i_4) \\
\tfrac{1}{4!} & (i_1 = i_2 = i_3 = i_4) \\
\end{cases} .
\lbl{sigma-def}
}

To ensure that $Q_k$ commutes with $H$, we use the following relations:
\eq{
\left[ \sums h_{1:k}, H \right]&= c_{40} E_{k,0} \nx \wha+ 2 (1+(-1)^k) \sums c_{20} \bd_1 \bd_k + c_{10} \bd_1, \\
\left[ \sums h_{1:k-2}, H \right] &= c_{40} E_{k-2,0} \nx \wha+ 2 (1+(-1)^k) \sums c_{20} \bd_1 \bd_{k-2} + c_{10} \bd_1, \\
\left[ \sums \bd_1 \bd_{k-1}, H \right] &= -2 \sums \bd_1 \bd_k + \bd_1 \bd_{k-2}, \\
\left[ \sums \bd_1, H \right] &= -2 \sums \bd_1,
}
\eq{
\left[ D_{l,n} , H\right] &= \begin{cases}
-E_{l+1,n} +E_{l,n+1} + 2E_{l-1,n} \\ \hfill (n=0) \\
-E_{l+1,n} -3E_{l,n-1} +E_{l,n+1} +2E_{l-1,n} \\ \hfill  (n=1) \\
-E_{l+1,n} -E_{l,n-1} +E_{l,n+1} +2E_{l-1,n} \\ \hfill  (2\leq n \leq l-1) \\
\end{cases},
}
where we define
\eq{
E_{l,n} := 
12 \sums (-1)^{l-n}(\bd_1)^2 \bd_{n+1} \bd_l + \bd_1 \bd_{l-n} (\bd_l)^2.
}
For convenience, we set $E_{l,l-1} = E_{l,l} = 0$.
Substituting these into $[Q_k, H]$, we directly verify the commutativity of $Q_k$.
\end{proof}

\subsubsection{Case of $c_{11}\neq 0$ (Type~N$^+$)}

\begin{lem}\lbl{t:degree4-N+}
The Hamiltonian \eqref{model4-raw} with $c_{11} \neq 0$ has no $k$-local charge for any $4 \leq k \leq N/2$.
\end{lem}
\begin{proof}
Here we show only the absence of 4-local charges.
The absence of $k$-local charges for general $5 \leq k \leq N/2$ is shown in \apref{proof-degree4-N+} in a similar manner to the argument below.

The displaced Hamiltonian is \eref{model4} with $c_{11}\neq 0$.
We claim that $Q_4$ satisfies Step~3 (i.e., $[Q_4 , H] = \lo{2}$) if and only if it can be written as
\eq{
Q_4 &= \sums 
\bd_1 b_4 + b_1 \bd_4 + 6c_{11} \left( \bd_1 b_3 + b_1 \bd_3 \right) + t \left( \bd_1 b_2 + b_1 \bd_2 \right) \nx
\whb + 2c_{20} \left(\bd_1 \bd_3 + 5c_{11} \bd_1 \bd_2\right) \nx
\whb + 8c_{40}c_{11} \left( (\bd_1)^3 \bd_2 + \bd_1 (\bd_2)^3 \right) \nx
\whb + 4c_{40} \left( (\bd_1)^3 \bd_3 -3 \bd_1(\bd_2)^2\bd_3 + \bd_1(\bd_3)^3 \right) \nx
\whb + 6c_{40} (\bd_1)^2(\bd_2)^2  + \lo{1}
\lbl{Q4-N+}
}
with a constant $t$, up to an overall constant factor.
A direct calculation confirms the if part, namely that the right-hand side of \eref{Q4-N+} satisfies Step~3.
The only-if part, namely that any $Q_4$ satisfying Step~3 must be of the form \eqref{Q4-N+}, is proved below.

Let $Q_4$ be as in \eref{Q4-N+}, and let $Q'_4$ be any other quantity satisfying Step~3.
Since there is no inversion-antisymmetric solution of Step~3 here, $Q'_4$ is inversion symmetric.
By the Step~2 analysis (see also \eref{exampleQ4}), the 4-local part of any quantity satisfying Step~3 is fixed up to an overall normalization, so we choose it so that $Q'_4$ shares the same 4-local part as $Q_4$ and define $\Delta:=Q'_4-Q_4$, which is then at most 3-local.
By linearity, it satisfies
\eq{
[\Delta , H] = \lo{2} .
\lbl{delta-N+}
}
We analyze $\Delta$ by cases according to the size of its support, using the general constraints obtained in the Step~2 analysis (\lref{step2}).

If $\Delta$ is 3-local, then \eref{delta-N+} together with \eref{step-iff} implies that $\Delta$ is a Step~2 solution for $k=3$; however, \lref{step2} states that no inversion-symmetric Step~2 solution exists for $k=3$, which is a contradiction.

If $\Delta$ is a 2-local quantity, then \eref{delta-N+} implies that $\Delta$ meets the Step~1 condition for $k=2$; hence \lref{step1} implies that $\Delta$ is a linear combination of $h_{1:2}$, $\hb_{1:2}$, and 1-local terms.
In the inversion-symmetric sector this reduces to $\Delta \propto h_{1:2} + \lo{1}$, which is absorbed by the arbitrariness of $t$ in \eref{Q4-N+}.

Consequently, in all cases $Q'_4$ satisfying Step~3 is shown to be expressed as \eref{Q4-N+}.

\bigskip

We next show that $Q_4$ given in \eref{Q4-N+} is not a charge.
Indeed, the coefficient of $(\bd_1)^2 (\bd_2)^2$ in $[Q_4 ,H]$ is
\eq{
r_{(\bd_1)^2 (\bd_2)^2}
&= \und{-24 c_{40}c_{11} -24 c_{40}c_{11}}{\text{(the third line of \eref{Q4-N+})} \\\co \sums h_{1:2}}
\und{-24 c_{40}c_{11}}{\ \ \text{(the fifth line of \eref{Q4-N+})} \\\co \sums c_{11} n_1} \nx
&= -72 c_{40}c_{11} \neq 0.
\lbl{cc-N+}
}
Since the commutator of any 1-local part in \eref{Q4-N+} with $H$ does not generate $(\bd_1)^2 (\bd_2)^2$, this nonzero coefficient cannot be canceled.
Hence, $[Q_4,H]$ is nonzero, and $Q_4$ is not a local charge.
\end{proof}

\subsection{Models of degree 3}
\lbl{s:model-3}

In this subsection, we consider the Hamiltonian \eqref{model3} (the displaced Hamiltonian of \eqref{model3-raw}), a degree-3 Hamiltonian with a 3-local charge, and determine the existence of $k$-local charges for all $k\geq 4$.
There are two types of models: those that have $k$-local charges for all $k\geq 4$ (Type~C), and those that have $k$-local charges for all $k\geq 5$ but not for $k=4$ (Type~C$^-$).
The former class further splits into two subtypes.

We first determine the condition when a 4-local charge exists (\lref{degree3-4local}).
Next, we split the case with a 4-local charge into two subtypes and show that they admit $k$-local charges for all $k\geq 3$ (\lref{degree3-Ca} and \lref{degree3-Cb}).
Finally, we treat the models without a 4-local charge and show that they admit $k$-local charges for all $k\geq 5$ (\lref{degree3-C-}).

\subsubsection{Classification by the existence of a 4-local charge}

\begin{lem}\lbl{t:degree3-4local}
The Hamiltonian \eqref{model3-raw} has a $4$-local charge if and only if one of $c_{20}(c_{11}+2)=3c_{30}c_{01}$ or $c_{11}=1$ holds.
\end{lem}

The conditions $c_{20}(c_{11}+2)=3c_{30}c_{01}$ and $c_{11}=1$ in the original Hamiltonian \eqref{model3-raw} are respectively equivalent to the conditions $c_{01}=0$ and $c_{11}=1$ in the displaced Hamiltonian \eqref{model3}.
In the proof below, we work with the displaced Hamiltonian \eqref{model3} and show that it admits a 4-local charge if and only if either $c_{01}=0$ or $c_{11}=1$ holds.

\begin{proof}
Consider the following 4-local quantity
\eq{\lbl{tildeQ4}
\tilde{Q}_4 &= \sums \bd_1 b_4 + b_1 \bd_4 
+ (\tfrac{9}{2}c_{11}-3) \left( \bd_1 b_3 + b_1 \bd_3 \right) \nx
\whb + (\tfrac{9}{2}c_{11}^2-3) \left( \bd_1 b_2 + b_1 \bd_2 \right) + (3c_{11}-2) n_1 \nx
\whb - 3 c_{30} c_{01} \left( \bd_1 \bd_2 + (c_{11}-2) (\bd_1)^2 \right) \nx
\whb + 3c_{30} \left( (\bd_1)^2 \bd_2 - 2 \bd_1 \bd_2 \bd_3 + \bd_1 (\bd_3)^2 \right) \nx
\whb + (\tfrac{9}{2} c_{11} - 3) c_{30} \left( (\bd_1)^2 \bd_2 + \bd_1 (\bd_2)^2 \right) \nx
\whb + (-2c_{11}^2 + 7c_{11} - 2) c_{10} \bd_1 \nx
\whb + (9c_{11} - 6) c_{01} b_1 .
}
The commutator between $\tilde{Q}_4$ and $H$ reads
\eq{
[\tilde{Q}_4, H]
&= 2(c_{11}-1)c_{01} \cdot \nx
\wha \left( \sums 3c_{30}c_{01}\bd_1 -(c_{11}+2)^2 b_1 -(c_{11}+2) c_{10}I \right). \nx 
\lbl{tildeQ4-degree3}
}
Hence, $\tilde{Q}_4$ is a 4-local charge if either $c_{01}=0$ or $c_{11}=1$ holds.

\bigskip

We next show that if neither $c_{01}=0$ nor $c_{11}=1$ holds, then this Hamiltonian has no 4-local charge.
Suppose, for contradiction, that $\tilde{Q}_4 + \Delta$ is a 4-local charge. Then $\Delta$ satisfies
\eq{
[\Delta ,H] &= -[\tilde{Q}_4,H] \nx &= \lo{1}.
\lbl{delta-degree3}
}
By the Step~2 analysis and the corresponding relations \eqref{exampleQ4}, it suffices to consider $\Delta$ that is at most 3-local and inversion symmetric (see the Type~N$^+$ case in \ssref{model-4}).
We analyze $\Delta$ by cases according to the size of its support, using the general constraints obtained in the Step~2 and Step~1 analyses (\lref{step2} and \lref{step1}).

If $\Delta$ is 3-local, then \eref{delta-degree3} implies that $\Delta$ is a Step~2 solution for $k=3$; however, \lref{step2} states that no inversion-symmetric Step~2 solution exists for $k=3$, which is a contradiction.

If $\Delta$ is 2-local, then $\Delta$ is a linear combination of $\sums h_{1:2}$ and 1-local terms; hence, by shifting $\tilde{Q}_4$ by a multiple of $H$, we may reduce to the following 1-local case.

Finally, we consider the case that $\Delta$ is 1-local.
A short calculation shows that, in order to satisfy \eref{delta-degree3}, $\Delta$ should be a linear combination of $\sums \bd_1$, $\sums b_1$, and $\sums n_1$.
At this stage, we examine the condition $[\Delta, H] = -[\tilde{Q}_4, H]$.
Recalling that $[\tilde{Q}_4, H]$ contains only terms involving $\bd$ and $b$, we conclude that $\Delta$ contains no terms proportional to $\sums b_1$ or $\sums n_1$.
Indeed, $\bigl[\sums b_1, H\bigr] = 3 c_{30} \sums (\bd_1)^2 + \cdots$ and $\bigl[\sums n_1, H\bigr] = 3 c_{30} \sums (\bd_1)^3 + \cdots$ would introduce terms that do not appear in $[\tilde{Q}_4, H]$.
Thus, $\Delta$ is constrained to $\Delta \propto \sums \bd_1$.
However, the commutator between $\sums \bd_1$ and $H$ reads
\eq{
\left[ \sums \bd_1, H \right] = - \sums ((c_{11}+2) \bd_1 + c_{01}I),
}
which cannot cancel all terms in \eref{tildeQ4-degree3}.
More precisely, if $c_{11}\neq -2$ then the $\sums b_1$ term remains with a nonzero coefficient.
If $c_{11}=-2$ then the $\sums \bd_1$ term remains with a nonzero coefficient.
This completes the proof.
\end{proof}

\subsubsection{Case of $c_{20}(c_{11}+2) = 3c_{30}c_{03}$ (Type~C)}

\begin{lem}\lbl{t:degree3-Ca}
The Hamiltonian~\eqref{model3-raw} with $c_{20}(c_{11}+2) = 3c_{30}c_{03}$ has a $k$-local charge for any $3 \leq k \leq N/2$.
\end{lem}

Through the displacement transformation, this model is equivalent to
\eq{
H = \sums h_{1:2} + c_{30} (\bd_1)^3 + c_{10} \bd_1 + c_{11} n_1,
\lbl{model3-Ca}
}
where the condition $c_{20}(c_{11}+2) = 3c_{30}c_{03}$ in \eref{model3-raw} corresponds to $c_{01}=0$ in \eref{model3}.
Note that this Hamiltonian has already been shown to have 3-local and 4-local charges, as demonstrated in \lref{3local} (its concrete form is given in \eref{explicit-3}) and \lref{degree3-4local} (its concrete form is given in \eref{tildeQ4}), respectively.

Unlike the case of Type~C models of degree~4, where we explicitly construct all $k$-local charges in \eref{exact-degree4-C}, we do not have a concise expression for the $k$-local charges in the present case.
Therefore, instead of providing an explicit construction, we prove the existence of $k$-local charges by introducing a recurrence algorithm that determines all coefficients consistently.
Here we provide only an intuitive picture of the algorithm.
Its full proof will be presented in \apref{proof-degree4-C}.

\bigskip

We take the simplest setup,
\eq{\lbl{deg3-g-simple-ex}
g=c_{30}(\bd_1)^3+c_{11}n_1 ,
}
as an example.
Here we explain how our algorithm determines the coefficients $q_{\bd_1\bd_m\bd_l}$ for $1\leq m\leq l, l \leq k-1$.

The condition that the coefficient of $\sums \bd_1 \bd_m \bd_l$ (with $m\neq 1,2,l-1,l$) in $[Q_k,H]$ vanishes implies
\eq{
0&=-q_{\bd_1\bd_m \bd_{l+1}}-q_{\bd_1\bd_{m+1} \bd_{l+1}} \nx
\wha -q_{\bd_1\bd_{m-1} \bd_l}-3c_{11}q_{\bd_1\bd_m \bd_l} -q_{\bd_1\bd_{m+1}\bd_l} \nx
\wha -q_{\bd_1\bd_m \bd_{l-1}}-q_{\bd_1\bd_{m-1} \bd_{l-1}}. \lbl{deg3-rec-simple}
}
For $m=2$ and $m=l-1$ the relation is similar, with some coefficients differing. The edge cases $m=1$ and $m=l$ are treated separately below.

Thanks to Steps~1 and~2, the coefficients with the largest support, namely those with $l=k-1$, are already fixed. We therefore determine the remaining coefficients $q_{\bd_1\bd_m\bd_l}$ recursively in $l$, starting from $l=k-2$ and proceeding down to $l=1$.

Fix $l$ and assume that $q_{\bd_1\bd_m\bd_{l+1}}$ and $q_{\bd_1\bd_m\bd_{l}}$ are already known for all $m$.
Then \eref{deg3-rec-simple} gives linear relations linking neighboring coefficients in $m$ at support size $l-1$, i.e., $q_{\bd_1\bd_m\bd_{l-1}}$ and $q_{\bd_1\bd_{m-1}\bd_{l-1}}$.
Accordingly, the coefficients $q_{\bd_1\bd_m\bd_{l-1}}$ can be determined successively in $m$, up to an overall choice.
In a fixed inversion-symmetry sector, inversion (anti-)symmetry imposes a central constraint ($m\sim l/2$), which fixes this remaining overall choice.
When $k$ is odd, $Q_k$ is inversion antisymmetric, so the central coefficients obey
$q_{\bd_1\bd_{(l+1)/2} \bd_l}=0$ for odd $l$ and $q_{\bd_1\bd_{l/2} \bd_l}=-q_{\bd_1\bd_{l/2+1} \bd_l}$ for even $l$.
The case of even $k$ is analogous.
Repeating this for $l=k-2,k-3,\dots,1$ yields a set of coefficients $q_{\bd_1\bd_m \bd_l}$ for all $l$ and $m$.

One may be concerned about the consistency of the commutator constraints at the edges ($m=1$ and $m=l$).
As an example of an edge case, consider $\sums (\bd_1)^2\bd_l$.
The condition that $\sums (\bd_1)^2\bd_l$ appears with zero coefficient in $[Q_k,H]$ implies
\eq{
0 &= -q_{(\bd_1)^2\bd_{l+1}}-q_{\bd_1\bd_2\bd_{l+1}} \nx
\wha -3c_{11}q_{(\bd_1)^2\bd_l}-q_{\bd_1\bd_2\bd_l} \nx
\wha -q_{(\bd_1)^2\bd_{l-1}}+3c_{30}q_{b_1\bd_l} .
}
Importantly, this relation is the only nontrivial constraint involving $q_{b_1\bd_l}$.
Hence, we can always choose $q_{b_1\bd_l}$ so that this constraint is satisfied.

The aforementioned argument is for \eref{deg3-g-simple-ex}, i.e., $c_{10}=0$ case.
In the case with $c_{10}\neq 0$, the coefficients $q_{b_1\bd_j}$ have contribution to $\sums \bd_1$ in $[Q_k, H]$ in the form of the following constraint:
\eq{
0=c_{10}\left( \sum_{2\leq l \leq k-1} q_{b_1\bd_l}+q_{\bd_1b_l}\right) 
+c_{10} q_{\bd_1b_1}
-(c_{11}+2)q_{\bd_1} .
}
Since this relation is the only nontrivial constraint containing $q_{\bd_1}$, as long as $c_{11}\neq -2$ we can employ this as an adjustable parameter so that this constraint is satisfied.
The case of $c_{11}=-2$ is treated separately, and we can provide a consistent construction even in this case.

\subsubsection{Case of $c_{11}=1$ (Type~C)}
\newcommand{\CC}[2]{\left(\substack{#1 \\ #2}\right)}
\newcommand{\LL}[2]{\left(\substack{#1 \\ #2}\right)'}

\begin{lem}\lbl{t:degree3-Cb}
The Hamiltonian~\eqref{model3-raw} with $c_{11}=1$ has a $k$-local charge for any $3 \leq k \leq N/2$.
\end{lem}

\begin{proof}
Unlike the rest of this subsection, we analyze here the model
\eq{
H = \sums h_{1:2} + c_{30} (\bd_1)^3 + c_{20} (\bd_1)^2 + c_{10} \bd_1 + n_1
\lbl{model3-Cb}
}
rather than \eref{model3}.
This model is obtained by applying the displacement operator to \eref{model3-raw} with $z = -\frac{c_{01}}{c_{11}+2} = -\frac{c_{01}}{3}$.
The advantage of this model is that we can explicitly write down its $k$-local charges.

For simplicity, we first treat the case $c_{20}=0$ in the Hamiltonian~\eqref{model3-Cb} and write down the $k$-local charge explicitly.
The $k$-local charge for odd $k\geq 3$ is given by
\begin{widetext}
\eq{
Q_k &= \sums \sum_{\max(\frac{k-3}{2},2)\leq l \leq k} \left( q_{\bd_1 b_l} \bd_1 b_l + q_{b_1 \bd_l} b_1 \bd_l \right)
+ \sum_{\frac{k+1}{2}\leq l \leq k-1} \sum_{1 \leq m \leq l} q_{\bd_1 \bd_m \bd_l} \bd_1 \bd_m \bd_l
\lbl{exact-Cb-odd}
}
with
\eq{
q_{\bd_1 b_l} &= - q_{b_1 \bd_l} = \CC{\frac{k-3}{2}}{k-3-l} + \CC{\frac{k-3}{2}}{k-2-l} + \CC{\frac{k-3}{2}}{k-1-l}
\\
q_{\bd_1 \bd_m \bd_l}
&= \begin{cases}
6 (-1)^m \CC{\frac{k-1}{2}-m}{k-1-l} \sigma_{1,m,l} c_{30} & (1 \leq m \leq l - \tfrac{k-1}{2}) \\
6 (-1)^{l-m} \CC{\frac{k-3}{2}-l+m}{k-1-l} \sigma_{1,m,l} c_{30} & (\tfrac{k+1}{2} \leq m \leq l) \\
0 & (\text{otherwise})
\end{cases}.
}
Here we adopt the convention that the binomial coefficient $\CC{i}{j}$ equals zero unless $0 \leq j \leq i$ holds.

The $k$-local charge of the Hamiltonian~\eqref{model3-Cb} with $c_{20}=0$ for even $k\geq 4$ is given by
\eq{
Q_k &= \sums \sum_{\max(\frac{k-2}{2},2)\leq l \leq k} \left( q_{\bd_1 b_l} \bd_1 b_l + q_{b_1 \bd_l} b_1 \bd_l \right)
+ \sum_{\frac{k}{2}\leq l \leq k-1} \sum_{1 \leq m \leq l} q_{\bd_1 \bd_m \bd_l} \bd_1 \bd_m \bd_l
+ q_{n_1} n_1 + q_{\bd_1} \bd_1
\lbl{exact-Cb-even}
}
with
\eq{
q_{\bd_1 b_l} &= q_{b_1 \bd_l} = \half \LL{\frac{k-2}{2}}{k-3-l} + \half \LL{\frac{k-2}{2}}{k-2-l} + \half \LL{\frac{k-2}{2}}{k-1-l} \\
q_{n_1} &= \delta_{k,4} \\
q_{\bd_1 \bd_m \bd_l}
&= \begin{cases}
6 (-1)^{\frac{k}{2}-1} & ((l,m) = (k-1,\frac{k}{2}))\\
3 (-1)^{m-1} \LL{\frac{k}{2}-m}{k-1-l} \sigma_{1,m,l} c_{30} & (1 \leq m \leq l - \tfrac{k-2}{2})\\
3 (-1)^{l-m} \LL{\frac{k-2}{2}-l+m}{k-1-l} \sigma_{1,m,l} c_{30} & (\tfrac{k}{2} \leq m \leq l) \\
0 & (\text{otherwise})
\end{cases} \\
q_{\bd_1} &= 3 \cdot 2^{\frac{k-4}{2}} c_{10}
\lbl{qux-Cb}
}
\begin{comment}
Analysis carried out with the $c_{01}$ model kept as is:
p_{\bd_1 \bd_l} &= - c_{30}c_{01} \sum_{1 \leq m \leq l} \frac{q_{\bd_1 \bd_m \bd_l}}{\sigma_{1,m,l}}
\nx \wha
- c_{30}c_{01} \sum_{2 \leq m \leq k-l} \left( 
\frac{q_{\bd_1 \bd_m \bd_{l+m-1}}}{\sigma_{1,m,l+m-1}} +
\frac{q_{\bd_1 \bd_l \bd_{l+m-1}}}{\sigma_{1,l,l+m-1}}
\right) \nx
q_{\bd_1 \bd_l} &= \sigma_{1,l} \left( \half p_{\bd_1 \bd_l} - q_{\bd_1 \bd_{l+2}} - q_{\bd_1 \bd_{l+1}} \right) \nx
q_{\bd_1} &= \frac{c_{10} q_{(\bd_1)^2}}{c_{01}c_{30}} \nx
q_{b_1} &= \frac{q_{(\bd_1)^2}}{c_{30}}
\end{comment}
\end{widetext}
Here $\LL{i}{j} := \CC{i}{j} + \CC{i-1}{j}$ is the $(i,j)$ entry of the reversed Lucas triangle~\cite{numbers}.

\bigskip

We next consider the general case of the Hamiltonian~\eqref{model3-Cb} with $c_{20}\neq 0$.
In what follows, we shall refer to the shift sum of an operator containing $x$ creation operators $\bd$ and $y$ annihilation operators $b$ as a {\it $\bx{xy}$ quantity}.
For example, $\sums \bd_1 \bd_m \bd_l$ and $\sums (\bd_1)^3$ are $\bx{30}$ quantities, and $\sums h_{1:k}$, $\sums \hb_{1:k}$, and $\sums n_1$ are $\bx{11}$ quantities.

We now construct $k$-local charges.
Let $\mathring{Q}_k$ denote the $k$-local charge for $c_{20}=0$, given by \eref{exact-Cb-odd} when $k$ is odd and by \eref{exact-Cb-even} when $k$ is even.
We first consider odd $k$.
In this case, $\mathring{Q}_k$ still serves as a $k$-local charge for $c_{20}\neq 0$.
The reason is as follows.
The charge $\mathring{Q}_k$ consists of $\bx{30}$ terms and $\bx{11}$ terms.
However, a $\bx{30}$ quantity commutes with $\sums (\bd_1)^2$, and the commutator between a $\bx{11}$ quantity and $\sums (\bd_1)^2$ produces a $\bx{20}$ quantity, which is always canceled for odd $k$ due to the inversion antisymmetry of $[\mathring{Q}_k,H]$.

We next consider $k$-local charges with even $k$.
In this case, the commutator of $\mathring{Q}_k$ with the Hamiltonian~\eqref{model3-Cb} yields a sum of $\bx{20}$ quantities.
These terms can be canceled by adding the following correction term to $\mathring{Q}_k$:
\eq{
q_{\bd_1 \bd_l} = \LL{\frac{k}{2}}{k-1-l} c_{20} \left( = \frac{2c_{20}}{3c_{30}} q_{(\bd_1)^2 \bd_l} = \frac{2c_{20}}{3c_{30}} q_{\bd_1 (\bd_l)^2}\right) .
\lbl{quux-Cb}
}
In other words, the following $k$-local quantity
\eq{
Q_k &= \sums \sum_{\max(\frac{k-2}{2},2)\leq l \leq k} \left( q_{\bd_1 b_l} \bd_1 b_l + q_{b_1 \bd_l} b_1 \bd_l \right) \nx
\whb + \sum_{\frac{k}{2}\leq l \leq k-1} \sum_{1 \leq m \leq l} q_{\bd_1 \bd_m \bd_l} \bd_1 \bd_m \bd_l \nx
\whb + \sum_{\frac{k}{2} \leq l \leq k-1} q_{\bd_1 \bd_l} \bd_1 \bd_l \nx
\whb + q_{n_1} n_1 + q_{\bd_1} \bd_1
}
with \eref{qux-Cb} and \eref{quux-Cb} serves as the $k$-local charge of the Hamiltonian~\eqref{model3-Cb} for even $k$.
\end{proof}

\subsubsection{The remaining case: $c_{20}(c_{11}+2) \neq 3c_{30}c_{03}$ and $c_{11}\neq 0$ (Type~C$^-$)}

\begin{lem}\lbl{t:degree3-C-}
The Hamiltonian~\eqref{model3-raw} with $c_{20}(c_{11}+2) \neq 3c_{30}c_{03}$ and $c_{11}\neq 0$ (and $c_{30} \neq 0$)
has $k$-local charges for any $5 \leq k \leq N/2$.
\end{lem}

We have already shown the existence of a 3-local charge and the absence of 4-local charges in \lref{3local} and \lref{degree3-4local}, respectively.
Thus, the above lemma completes the classification of $k$-local charges and shows that this model is Type~C$^-$.

The proof of \lref{degree3-C-} is similar to that of \lref{degree3-Ca}; we present an algorithm to determine all coefficients in $Q_k$.
We postpone the full proof to \apref{proof-degree3-C-} and present only the key idea here.

We first reduce the above Hamiltonian to the displaced Hamiltonian \eqref{model3} with $c_{01} \neq 0$ and $c_{11} \neq 1$ (and $c_{30} \neq 0$).
The $k$-local charges $Q_k$ consist of $\bx{30}$, $\bx{11}$, $\bx{20}$, $\bx{10}$, and $\bx{01}$ quantities.
The main idea of the algorithm used to determine these coefficients is the same as that in \lref{degree3-Ca}; in fact, the coefficients of the $\bx{30}$ quantities are already fixed there.

It remains to determine the coefficients of the $\bx{20}$ quantities, which do not appear in \lref{degree3-Ca}.
If $k$ is odd, $Q_k$ is inversion antisymmetric; hence all coefficients of the $\bx{20}$ quantities vanish automatically.
If $k$ is even, we determine them by requiring that the $\bx{20}$ terms in $[Q_k,H]$ vanish.
The condition that the coefficient of $\sums \bd_1\bd_l$ in $[Q_k,H]$ vanishes implies
\eq{
0&= 
-c_{01} \sum_{1 \leq m \leq l} q_{\bd_1\bd_m\bd_l} 
-c_{01} \sum_{l \leq m \leq k-1} q_{\bd_1\bd_l\bd_m} \nx
\wha - c_{01} \sum_{1 \leq m \leq k-l} q_{\bd_1\bd_m \bd_{l+m-1}}
 \nx
\wha + c_{20}(q_{\bd_1 b_l}+q_{b_1 \bd_l})
- 2c_{11}q_{\bd_1\bd_l}
- 2(q_{\bd_1\bd_{l+1}}+q_{\bd_1\bd_{l-1}}) .
}
Since the coefficients of the $\bx{30}$ and $\bx{11}$ quantities have already been determined in the previous step, the above constraint becomes a set of linear equations for the $\bx{20}$ coefficients. In particular, it can be viewed as a three-term recurrence in $l$ relating $q_{\bd_1\bd_{l-1}}$, $q_{\bd_1\bd_l}$, and $q_{\bd_1\bd_{l+1}}$.
The solution of these equations determines the $\bx{20}$ coefficients, and this can be done consistently for all even $k$.

Finally, we examine the coefficients of the $\bx{10}$ and $\bx{01}$ quantities in $[Q_k,H]$.
Unfortunately, the above procedure yields nonzero values for these coefficients.
For all even $k$, the resulting $k$-local quantity $Q_k'$ satisfies
\eq{
[Q_k' ,H]
&\propto \sums -2 \left( c_{01} + \frac{c_{10}\tilde{c}} {c_{01}} \right)  \bd_1 + 2\tilde{c} b_1
\lbl{degree3-N--1local}
}
with $\tilde{c}:={(c_{11}+2)^2}/{3c_{30}}$.
Since the right-hand side of \eref{degree3-N--1local} is nonzero, $Q_k'$ itself is not a charge.

Nevertheless, for even $k\geq 6$ we can construct a genuine $k$-local charge. Indeed, in this case the commutator \eqref{degree3-N--1local} consists only of 1-local terms, and the two 1-local operators $\bd_1$ and $b_1$ appear with proportional coefficients. Therefore, this remaining 1-local contribution can be canceled by taking a suitable linear combination of $Q_k'$ and $Q_4'$. Concretely, we choose $s_k$ such that
\eq{
[Q_k'+s_k Q_4' ,H]=0 .
}
We define $\tilde{Q}'_k:=Q_k'+s_k Q_4'$, which yields a $k$-local charge for each even $k\geq 6$. Combined with the odd-$k$ construction, this shows the existence of $k$-local charges for all $k\geq 5$.
Note that this argument does not work for $k=4$; indeed, we have already shown that no 4-local charge exists.

\subsection{Absence of 3-local charges implies nonintegrability}
\lbl{s:GMC1}

In this subsection, we prove a lemma showing that, for Hamiltonians of the form \eqref{H-gen}, the absence of a 3-local charge implies the absence of any $k$-local charge with $3 \leq k \leq N/2$.
This lemma directly implies that all models not listed in \lref{3local} are of Type~N.
The proof applies the method of Ref.~\cite{hokkyo2025rigorous} for general spin systems to the bosonic setting considered here, and the model restriction allows for a stronger conclusion.

We first treat the case $k=4$, which serves as a prototype for general $k$.
\begin{lem}\lbl{t:absence34}
For any Hamiltonian of the form \eqref{H-gen}, if it does not have a $3$-local charge, then it does not have a $4$-local charge either.
\end{lem}

\begin{proof}
We prove the lemma in the following lines:
\eq{
\wha \text{existence of a 4-local charge} \nx
&\Rightarrow \text{existence of } Q_4 \text{ satisfying } [Q_4, H] = \lo{2}
\nx
&\Rightarrow \text{existence of } Q_3 \text{ satisfying } [Q_3, H] = \lo{1} \nx
&\Leftrightarrow \text{existence of a 3-local charge}
\lbl{hie-34}
}
The first implication is immediate.
The equivalence between the third line (Step~3 for $k=3$; \eref{3local-eq}) and the fourth line is shown in \ssref{proof2-1-local}.
Therefore, it suffices to show that the existence of $Q_4$ satisfying Step~3 implies the existence of $Q_3$ satisfying Step~3, as we show below.

Restating the result of Steps~1 and~2 (see \eref{exampleQ4}), we have
\eq{
Q_4 &= \sums 
{h_{1:4}}
\nx \whb +
{g_1 \co h_{1:3} + h_{1:2} \co g_2 \co h_{2:3} + h_{1:3} \co g_3 + t \hb_{1:3}}
\nx \whb + 
{s_{1:2}}
+ \lo{1},
}
where $t$ is an undetermined parameter, and $s$ is an undetermined 2-local operator.

In Step~3 for $k=4$, we check whether there exist $t$ and $s$ such that $[Q_4, H] = \lo{2}$.
We compute this commutator by a long but straightforward calculation:
\begin{widetext}
\eq{
\wha [Q_4, H] \nx
&= \sums h_{1:4} \co H + (g_1 \co h_{1:3} + h_{1:2} \co g_2 \co h_{2:3} + h_{1:3} \co g_3 + t \hb_{1:3}) \co H + \und{s_{1:2} \co H}{=s_{1:2} \co h_{2:3} - h_{0:1} \co s_{1:2} + \lo{2}} + \lo{2} \nx
&= \sums 
(g_1 \co h_{1:3}) \co g_1 
+ \und{(g_1 \co h_{1:3}) \co g_2}{=0} 
+ \cancel{(g_1 \co h_{1:3}) \co g_3}
+ \und{(g_1 \co h_{1:3}) \co h_{1:2}}{=\half((g_1 \co h_{1:2}) \co h_{1:2}) \co h_{2:3}\\ \text{(\eref{useful-1})}} 
+ \und{(g_1 \co h_{1:3}) \co h_{2:3}}{=\lo{2}}\nx
\whb 
- g_1 \co h_{1:2} \co g_2 \co h_{2:3}
+ \und{(h_{1:2} \co g_2 \co h_{2:3}) \co g_2}
{=\half((h_{1:2}\co g_2)\co g_2)\co h_{2:3} + \half h_{1:2}\co((g_2\co h_{2:3})\co g_2)\\\text{(\eref{useful-3})}} 
+ h_{1:2} \co g_2 \co h_{2:3} \co g_3 \nx
\whb
+ \und{(h_{1:2} \co g_2 \co h_{2:3}) \co h_{1:2}}
{=((h_{1:2}\co g_2)\co h_{1:2})\co h_{2:3} + \lo{2}\\\text{(\eref{useful-4})}} 
+ \und{(h_{1:2} \co g_2 \co h_{2:3}) \co h_{2:3}}{=h_{1:2}\co((g_2\co h_{2:3})\co h_{2:3})+\lo{2}\\\text{(\eref{useful-5})}}  \nx
\whb 
- \cancel{g_1 \co h_{1:3} \co g_3} 
+ \und{(h_{1:3} \co g_3) \co g_2}{=0}
+ (h_{1:3} \co g_3) \co g_3
+ \und{(h_{1:3} \co g_3)\co h_{1:2}}{=\lo{2}}
+ \und{(h_{1:3} \co g_3) \co h_{2:3}}
{=\half h_{1:2}\co ((h_{2:3} \co g_3) \co h_{2:3})\\\text{(\eref{useful-2})}} \nx
\whb 
-t g_1 \co \hb_{1:3}
+t \und{\hb_{1:3} \co g_2}{=0}
+t \hb_{1:3}\co g_3 + s_{1:2} \co h_{2:3} - h_{1:2} \co s_{2:3} + \lo{2} \nx
&= \sums \und{\left( 
(g_1 \co h_{1:2} ) \co g_1
+ \half (g_1 \co h_{1:2} ) \co h_{1:2}
- g_1 \co h_{1:2} \co g_2
+ \half ( h_{1:2} \co g_2 ) \co g_2
+ (h_{1:2} \co g_2) \co h_{1:2}
- t g_1 \co h_{1:2} + s_{1:2} 
\right)}
{=: \chi_{1:2}}
\co h_{2:3}\nx
\whb + h_{1:2} \co 
\und{\left( 
\half (g_2 \co h_{2:3}) \co g_2 
+ g_2 \co h_{2:3} \co g_3 
+ (g_2 \co h_{2:3}) \co h_{2:3} 
+ (h_{2:3} \co g_3) \co g_3 + \half (h_{2:3} \co g_3) \co h_{2:3} 
+ t \hb_{2:3} \co g_3 - s_{2:3}
\right)}{=:\psi_{2:3}}\nx
\whb + \lo{2}\nx
&= \sums \chi_{1:2} \co h_{2:3} + h_{1:2} \co \psi_{2:3} + \lo{2} .
\lbl{cal-absence34}
}
\end{widetext}

We now examine the condition for Step~3 to admit a solution, i.e.,
\eq{
\chi_{1:2} \co h_{2:3} + h_{1:2} \co \psi_{2:3} = \lo{2} .
\lbl{k=4-step3}
}
We observe that any 3-local term in $\chi_{1:2} \co h_{2:3}$ must be canceled by the same 3-local term in $h_{1:2} \co \psi_{2:3}$. This implies that such a 3-local term must be expressible both as (some term) $\co\, h_{2:3}$ and as $h_{1:2}\,\co$ (some term).
Hence, we can express the above two terms using a 1-local operator $\omega$ as
\eq{
\chi_{1:2} \co h_{2:3} &= h_{1:2} \co \omega_2 \co h_{2:3} + \lo{2} , \lbl{chi-34-pre} \\
-h_{1:2} \co \psi_{2:3} &= h_{1:2} \co \omega_2 \co h_{2:3} + \lo{2} . \lbl{psi-34-pre}
}
In addition, since the superoperator $\co h_{2:3}$ (resp. $h_{1:2}\co$) is injective (i.e., there is no nontrivial 1-local operator $\delta$ satisfying $\delta_2 \co h_{2:3} = 0$ (resp. $h_{1:2} \co \delta_2 = 0$)), we further have
\eq{
\chi_{1:2} &= h_{1:2} \co \omega_2 + \lo{1}, \lbl{chi-34}\\
\psi_{2:3} &= - \omega_2 \co h_{2:3} + \lo{1}, \lbl{psi-34}
}
which directly imply
\eq{
\chi_{1:2} + \psi_{1:2} - h_{1:2} \co \omega_2 + \omega_1 \co h_{1:2} = \lo{1},
}
where we used the one-site shifted version of \eref{psi-34}.
Finally, substituting the definitions of $\chi_{1:2}$ and $\psi_{1:2}$ from \eref{cal-absence34}, we arrive at
\eq{
\wha
\tfrac{3}{2} (g_1 \co h_{1:2}) \co g_1
+ \tfrac{3}{2} (g_1 \co h_{1:2}) \co h_{1:2} \nx
\wha + \tfrac{3}{2} (h_{1:2} \co g_2) \co g_2
+ \tfrac{3}{2} (h_{1:2} \co g_2) \co h_{1:2} \nx
\wha + t( - g_1 \co \hb_{1:2} + \hb_{1:2} \co g_2) + \omega_1 \co h_{1:2} - h_{1:2} \co \omega_2 \nx
&= \lo{1}.
\lbl{chi-psi-34}
}
We now find that the existence of $t$ and $\omega$ satisfying \eref{chi-psi-34} is equivalent to the existence of $t$ and $s$ satisfying \eref{3local-eq-pre}, namely Step~3 for $k=3$, upon setting $t\to \frac32 t$ and $\omega\to \frac32 s$.
In summary, if $Q_4$ satisfies Step~3 for $k=4$, then \eref{chi-psi-34} holds, and hence \eref{3local-eq-pre} follows.
This shows that the second line of \eref{hie-34} implies its third line, completing the proof of the contraposition of the desired lemma: if a 4-local charge exists in a model, then a 3-local charge also exists.
\end{proof}

We now state the lemma for general $k$.
\begin{lem}\lbl{t:absence3k}
For any Hamiltonian of the form \eqref{H-gen}, if it does not have a $3$-local charge, then it does not have $k$-local charges for all $3 \leq k \leq N/2$.
\end{lem}

\begin{proof}
Here we only explain the outline of the proof.
The complete proof is presented in \apref{proof-GMC1}.

We show that for any $4 \leq k \leq N/2$, the existence of $Q_k$ satisfying Step~3 implies the existence of $Q_{k-1}$ satisfying Step~3.
If this holds, then we can apply this implication recursively:
\eq{
\wha \text{existence of a $k$-local charge} \nx
&\Rightarrow \text{existence of } Q_k \text{ satisfying } [Q_k , H] = \lo{k-2}
\nx
&\Rightarrow \text{existence of } Q_{k-1} \text{ satisfying } [Q_{k-1}, H] = \lo{k-3} \nx
&\Rightarrow \text{existence of } Q_{k-2} \text{ satisfying } [Q_{k-2}, H] = \lo{k-4} \nx
&\Rightarrow \cdots \nx
&\Rightarrow \text{existence of } Q_4 \text{ satisfying } [Q_4 , H] = \lo{2} \nx
&\Rightarrow \text{existence of } Q_3 \text{ satisfying }  [Q_3 ,H] = \lo{1} \nx
&\Leftrightarrow \text{existence of a 3-local charge}.
}
This implies that the existence of a $k$-local charge for some $4 \leq k \leq N/2$ entails the existence of a 3-local charge.

We prove the above statement in parallel with the proof of \lref{absence34} (for $k=4$).
First, the analyses of Steps~1 and~2 constrain $Q_k$ to be of the form
\eq{
Q_k &= \sums h_{1:k} \nx
\whb + \sum_{0 \leq m \leq k-1} h_{1:m} \co g_m \co h_{m:k-1} + t^{(k)} \hb_{1:k-1} \nx
\whb + s^{(k)}_{1:k-2} + \lo{k-3},
\lbl{prep-absence3k}
}
where $t^{(k)}$ is an undetermined parameter and $s^{(k)}$ is an undetermined $k-2$-local operator.
In Step~3 for $k$, we examine whether there exist $t^{(k)}$ and $s^{(k)}$ satisfying $[Q_k, H] = \lo{k-2}$.

In a similar manner to \eref{cal-absence34}, we find that, introducing $k-2$-local operators $\chi\pk$ and $\psi\pk$, the commutator between $Q_k$ in \eref{prep-absence3k} and the Hamiltonian can be written as
\eq{
[Q_k, H] &= \sums \chi_{1:k-2}\pk \co h_{k-2:k-1} + h_{1:2} \co \psi_{2:k-1}\pk \nx  \wha + \lo{k-2}.
}
The derivation of this relation is presented in \apref{proof-GMC1}.
Hence, the Step~3 condition, $[Q_k, H] = \lo{k-2}$, reads
\eq{
\chi_{1:k-2}\pk \co h_{k-2:k-1} + h_{1:2} \co \psi_{2:k-1}\pk = \lo{k-2}.
\lbl{chi-psi-h}
}
In a similar manner to Eqs.~\eqref{chi-34-pre} and \eqref{psi-34-pre}, we can show that \eref{chi-psi-h} means that $\chi\pk$ and $\psi\pk$ are expressed as
\eq{
\chi_{1:k-2}\pk \co h_{k-2:k-1} &= h_{1:2} \co \omega_{2:k-2}\pk \co h_{k-2:k-1} \nx \wha + \lo{k-2} \\
h_{1:2} \co \psi_{2:k-1}\pk &= - h_{1:2} \co \omega_{2:k-2}\pk \co h_{k-2:k-1} \nx \wha + \lo{k-2},
}
for some $k-3$-local operator $\omega\pk$.
In a similar manner to the derivation of Eqs.~\eqref{chi-34} and \eqref{psi-34}, and recalling the injectivity of the superoperator $\co h$, the above two relations imply
\eq{
\chi_{1:k-2}\pk &= h_{1:2} \co \omega_{2:k-2}\pk + \lo{k-3}, \lbl{chi-k}\\
\psi_{2:k-1}\pk &= - \omega_{2:k-2}\pk \co h_{k-2:k-1} + \lo{k-3}. \lbl{psi-k}
}
Combining \eref{chi-k} and the one-site translated version of \eref{psi-k}, we arrive at
\eq{
\wha \chi_{1:k-2}\pk + \psi_{1:k-2}\pk - h_{1:2} \co \omega_{2:k-2}\pk + \omega_{1:k-3}\pk \co h_{k-3:k-2} \nx &= \lo{k-3}.
\lbl{chi-psi-3k}
}

As shown in \apref{proof-GMC1}, \eref{chi-psi-3k} is equivalent to the Step~3 condition for $k-1$, i.e., there exist appropriate $s^{(k-1)}$ and $t^{(k-1)}$ such that $[Q_{k-1}, H] = \lo{k-3}$ is satisfied, where $Q_{k-1}$ is defined as \eref{prep-absence3k}.
Hence, if the Step~3 condition for $k$ is satisfied, then the Step~3 condition for $k-1$ is also satisfied, which completes the proof.
\end{proof}
\section{discussion}
\lbl{s:discussion}
We analyzed local commuting charges in a translationally invariant bosonic chain with symmetric hopping and a general on-site term that may be non-Hermitian. Our classification shows that the all-or-nothing picture fails: in addition to the usual integrable (Type~C) and nonintegrable (Type~N) types, we find two partially integrable types in the uniform sector, namely Type~N$^+$ and Type~C$^-$. To our knowledge, these are the first examples of translationally invariant nearest-neighbor systems violating the all-or-nothing rule. This implies that the Grabowski--Mathieu test based on 3-local charges \cite{grabowski1995integrability} is not universally applicable. Partial integrability appears only in the non-Hermitian case, while the Hermitian subclass still obeys all-or-nothing. We also completed the classification of staggered local charges and showed that, beyond uniform charges, staggered charges are the only possible local ones in this class.

It is worth clarifying the boundary of applicability of the all-or-nothing rule. Prior work has established this rule in several low-spin classes \cite{yamaguchi2024complete,hokkyo2025absence} and in SU(2)-symmetric spin chains \cite{shiraishi2025dichotomy}, whereas our counterexamples arise in non-Hermitian bosonic chains. At least from the perspective of our proof, violations of the rule seem to owe more to the infinite-dimensional local bosonic Hilbert space than to non-Hermiticity. Two questions remain open: whether a breakdown can occur in Hermitian bosonic systems beyond our setting, and whether spin chains obey the all-or-nothing rule. Extending the rule to broader settings while simultaneously searching for further counterexamples should sharpen where this boundary lies.

Our findings do not refute the philosophy behind the Grabowski--Mathieu test. More precisely, the present results challenge only the universality of 3-locality as the decisive diagnostic range. In our class, the presence or absence of a 5-local charge already determines the entire tail of the hierarchy: it decides whether $k$-local charges exist for every $k\ge 6$. This motivates the following conjecture.
\begin{conj*}[Generalized Grabowski--Mathieu test]
For each natural class of one-dimensional nearest-neighbor systems, there exists a finite integer $M\ge 3$ such that, for every Hamiltonian in the class, the existence of an $M$-local charge is equivalent to the existence of $k$-local charges for all $M+1\le k\le N/2$.
\end{conj*}
\noindent
In previously studied subclasses, $M=3$ has been rigorously established, whereas our class realizes $M=5$, indicating that the decisive range is class-dependent.

Finally, Type~C$^-$ poses a particularly sharp puzzle: it has infinitely many local charges, yet no standard exact-solution framework is currently available for it. In regular Yang--Baxter integrable nearest-neighbor chains, the transfer matrix generates local commuting charges at every range $k \ge 3$
\cite{tetelman1982lorentz}. The same pattern also appears in known Yang--Baxter integrable models with non-difference-type $R$-matrices, in which the higher charges can be constructed via generalized boosts
\cite{links2001ladder,deleeuw2021yangbaxter}. From this viewpoint, Type~C$^-$ falls outside the Yang--Baxter / transfer-matrix paradigm. This raises a natural question: do these models admit an exact-solution mechanism beyond the conventional framework, or do they instead realize a genuinely new class with infinitely many local charges but no exact solution method? An analogous question arises for the new Type~C models found by our bottom-up analysis: although they possess a complete local-charge tower, it remains open whether they admit any exact solution method.

\paragraph*{Acknowledgments}
We thank Masaya Kunimi for suggesting the initial problem that inspired this study, and Yuuya Chiba and Fuga Ishii for their constructive comments.
We also thank Atsuo Kuniba, Hosho Katsura, Akihiro Hokkyo, Hiroaki Shimeno, and Hal Tasaki for valuable discussions.

This work was supported by JSPS KAKENHI Grant Number JP25KJ0815 and JST ERATO Grant Number JPMJER2302.

\bibliography{main}

\clearpage

\widetext

\newcommand{\tk}{t^{(k)}}
\newcommand{\sk}{s^{(k)}}
\newcommand{\dk}{d^{(k)}}
\newcommand{\ek}{e^{(k)}}
\newcommand{\fk}{f^{(k)}}
\newcommand{\ddk}{\delta^{(k)}}
\newcommand{\erk}{\acute{\epsilon}^{(k)}}
\newcommand{\elk}{\grave{\epsilon}^{(k)}}

\newcommand{\addl}[1]{\und{#1}{\text{add to }\chi^{(k)}_{1:k-2}}}
\newcommand{\addr}[1]{\und{#1}{\text{add to }\psi^{(k)}_{2:k-1}}}

\appendix

\section{Step~1 analysis (Proof of \lref{step1})}
\lbl{s:proof-step1}

In this section, we perform Step~1 analysis and prove \lref{step1}.
More precisely, for a $k$-local quantity $Q_k$, we show the equivalence of the following two statements:
\begin{itemize}
\item $Q_k \co H$ is an at-most-$k$-local quantity, i.e., $Q_k \co H$ does not have a $k+1$-local parts.
\item $Q_k$ is expressed as
\eq{Q_k = \sumi ( q_{+-} \bd_i b_{i+k-1} + q_{-+} b_i \bd_{i+k-1}) + 
\sumi (-1)^i (q_{++} \bd_i \bd_{i+k-1} + q_{--} b_i b_{i+k-1} ) + \lo{k-1}
\lbl{Qk-step1-a}
}
with four constants $q_{+-}$, $q_{-+}$, $q_{++}$, and $q_{--}$.
\end{itemize}
Since the case $k=3$ has been addressed in the main text, we now establish the result for general $k$.

\begin{proof}
The implication from the latter to the former is straightforward.
A key observation is that the $k+1$-local parts of $Q_k \co H$ arise only from pairs of terms in the expansions of $Q_k$ and $H$ whose supports overlap on exactly one site.
Hence, for example, $\bd_i b_{i+k-1}$ in $Q_k$ can generate $k+1$-local parts only in the commutators with $h_{i+k-1:i+k}$ and $h_{i-1:i}$.
A direct computation of the $k+1$-local parts of \eref{Qk-step1-a} reveals that all such contributions cancel exactly, so that no $k+1$-local term remains.

We now show that the former implies the latter.
The proof is based on three observations.
For any operator $\bA$ with a nonzero coefficient, we have:
(i) The first and last symbols of $\bA$ are either $\op{10}$ ($=\bd$) or $\op{01}$ ($=b$).
(ii) All interior symbols of $\bA$ (i.e., all symbols except the first and last) are $\op{00}$ ($=I$).
(iii) The coefficients of $\bd_i b_{i+k-1}$ and $b_i \bd_{i+k-1}$ are independent of $i$ (i.e., they are translation invariant), whereas the coefficients of $\bd_i \bd_{i+k-1}$ and $b_i b_{i+k-1}$ acquire a factor of $(-1)^{i}$ (i.e., they are staggered).

We first show that, for $\bA$ to have a nonzero coefficient, its first and last symbols must be either $\op{10}$ or $\op{01}$.
Consider $\bA = \op{x_Ly_L}_{i} \op{x_{L'}y_{L'}}_{i+1} \dots \op{x_{R}y_{R}}_{i+k-1}$.
Suppose, for contradiction, that the first symbol $\op{x_L y_L}$ is neither $\op{10}$ nor $\op{01}$.
In this case, the commutator of $\bA$ with the interaction term $h_{i+k-1:i+k}$ produces a $(k+1)$-local operator $\bB$ as follows:
\eq{
\bA \co h_{i+k-1:i+k}
&= y_R \cdot \op{x_Ly_L}_{i} \op{x_{L'}y_{L'}}_{i+1} \dots \op{x_{R},y_{R}-1}_{i+k-1} \op{01}_{i+k}  - x_R \cdot \op{x_Ly_L}_{i} \op{x_{L'}y_{L'}}_{i+1} \dots \op{x_{R}-1,y_{R}}_{i+k-1} \op{10}_{i+k},
\lbl{step1-Ah}
}
where we have used the relations $\op{x_Ry_R} \co \op{10} = y_R \cdot \op{x_R-1,y_R}$ and $\op{x_Ry_R} \co \op{01} = -x_R \cdot \op{x_R,y_R-1}$ at site $i+k-1$.
Now, observe that the operators $\op{x_Ly_L}_{i} \op{x_{L'}y_{L'}}_{i+1} \dots \op{x_{R},y_{R}-1}_{i+k-1} \op{01}_{i+k}$ and $\op{x_Ly_L}_{i} \op{x_{L'}y_{L'}}_{i+1} \dots \op{x_{R}-1,y_{R}}_{i+k-1} \op{10}_{i+k}$ appearing on the right-hand side of \eref{step1-Ah} cannot be generated from commutators of any other pair of terms, under the assumption that $\op{x_L y_L}\neq \op{10},\op{01}$ (see \ssref{step1} for an example).
Hence, we obtain
\begin{alignat}{3}
r_{\op{x_Ly_L}_{i} \op{x_{L'}y_{L'}}_{i+1} \dots \op{x_{R},y_{R}-1}_{i+k-1} \op{01}_{i+k}} &=& y_R q_{\bA} &= 0\lbl{step1-r01},\\
r_{\op{x_Ly_L}_{i} \op{x_{L'}y_{L'}}_{i+1} \dots \op{x_{R}-1,y_{R}}_{i+k-1} \op{10}_{i+k}} &=& -x_R q_{\bA} &= 0.
\lbl{step1-r02}
\end{alignat}
Since the rightmost operator $\op{x_R y_R}$ is not $\op{00}$, $q_{\bA}=0$ holds if $\op{x_L y_L}$ is neither $\op{10}$ nor $\op{01}$.
A symmetric argument for the last symbol shows that $\op{x_R y_R}\neq \op{10},\op{01}$ implies $q_{\bA}=0$ as well.

Next, we turn to the interior symbols and show that they must all be $\op{00}$.
Consider, as a representative case, an operator $\bA'$ with first symbol $\op{x_L y_L}_i = \op{10}_i$ and last symbol $\op{x_R y_R}_{i+k-1} = \op{01}_{i+k-1}$;
the other cases, with the first/last symbols being $\op{10}$ or $\op{01}$ in different combinations, can be treated analogously.
Among the possible commutators that produce $\bB' = \bA' \co h_{i+k-1:i+k} = \op{10}_{i} \op{x_{L'}y_{L'}}_{i+1} \dots \op{00}_{i+k-1} \op{01}_{i+k}$, 
there are exactly two contributing pairs: $(\bA',h_{i+k-1:i+k})$ itself and $(\tilde{\bA}',h_{i:i+1})$, where $\tilde{\bA}' := \op{x_L+1,y_L}_{i+1} \dots \op{00}_{i+k-1} \op{01}_{i+k}$.
Hence, we obtain
\eq{
r_{\bB'} = q_{\bA'} - (x_L+1) q_{\tilde{\bA}'} = 0
\lbl{step1-prop} .
}
Since we have already shown that the coefficient of $\tilde{\bA}'$ vanishes unless its first symbol $\op{x_{L'}+1, y_{L'}}$ is either $\op{10}$ or $\op{01}$.
Therefore, $q_{A'}$ can be nonzero only when the second symbol of $\bA'$, namely $\op{x_{L'} y_{L'}}$, is $\op{00}$.

A similar argument applies to the third symbol.
Assume $\bA'' = \op{10}_i \op{00}_{i+1} \op{x_{L''}y_{L''}}_{i+2} \dots \op{01}_{i+k-1}$.
Then, by the same reasoning as in \eref{step1-prop}, the coefficient of $\bA''$ is proportional to that of
$\tilde{\bA}''=\op{10}_{i+1} \op{x_{L''} y_{L''}}_{i+2} \dots \op{00}_{i+k-1} \op{01}_{i+k}$.
Since we have already established that $q_{\tilde{\bA}''}=0$ unless the second symbol of $\tilde{\bA}''$ (i.e., $\op{x_{L''} y_{L''}}$) is $\op{00}$, it follows that $q_{\bA''}$ must vanish unless $\op{x_{L''} y_{L''}}=\op{00}$ holds.

Proceeding inductively, one shows that every interior symbol of $\bA$ must be $\op{00}$.
Thus, combining this with the previous step, we conclude that an operator $\bA$ with a nonzero coefficient must have its first and last symbols equal to $\op{10}$ or $\op{01}$, and all of its interior symbols equal to $\op{00}$.

Finally, we derive the relations satisfied by the remaining coefficients of $\bA$, by showing that they are restricted to the form given in \eref{Qk-step1-a}.
For instance, the operator $\bB = \bd_i b_{i+k}$ can be generated by exactly two pairs of operators, namely $(\bd_i b_{i+k-1},h_{i+k-1:i+k})$ and $(\bd_{i+1} b_{i+k},h_{i:i+1})$, which yield
\eq{
q_{\bd_i b_{i+k-1}} - q_{\bd_{i+1} b_{i+k}} = 0
}
In a similar manner, by considering the pairs that produce $\bB = b_i \bd_{i+k}$, $\bd_i \bd_{i+k}$, and $b_i b_{i+k}$, we obtain the relations
\eq{
q_{\bd_i b_{i+k-1}} &= q_{\bd_{i+1} b_{i+k}}, \\
q_{b_i \bd_{i+k-1}} &= q_{b_{i+1} \bd_{i+k}}, \\
q_{\bd_i \bd_{i+k-1}} &= - q_{\bd_{i+1} \bd_{i+k}}, \\
q_{b_i b_{i+k-1}} &= - q_{b_{i+1} b_{i+k}}.
}

In conclusion, the Step~1 analysis shows that the condition $[Q_k,H]$ being at most $k$-local is equivalent to $Q_k$ taking the restricted form \eref{Qk-step1-a}.
\end{proof}

\section{Step~2 analysis (Proof of \lref{step2})}
\lbl{s:proof-step2}

In this section, we perform Step~2 analysis and prove \lref{step2}.

More precisely, we show the equivalence of the following two statements on a $k$-local quantity $Q_k$ that is an eigenoperator of spatial inversion $P$ (i.e., $P[Q_k]=\pm Q_k$) and an on-site operator $g$ with $\deg(g)\ge 3$:
\begin{itemize}
\item $Q_k \co H$ is an at-most-$k-1$-local quantity, i.e., $Q_k \co H$ does not have $k+1$-local parts and $k$-local ones.
\item $Q_k$ is expressed as
\eq{
Q_k = \sums \left( h_{1:k} + \sum_{0 \leq m \leq k-1} h_{1:m} \co g_m \co h_{m:k-1} + t \hb_{1:k-1}  \right) + \lo{k-2}
\lbl{Qk-step2-a}
}
up to a scalar multiple.
\end{itemize}

\begin{proof}
The implication from the latter to the former is straightforward.
A key observation is that $k$-local parts in $Q_k \co H$ arise in only two ways: from pairs consisting of a $k$-local term of $Q_k$ and an on-site term $g$ in $H$ (when their supports overlap), and from a $(k-1)$-local term of $Q_k$ and an interaction term $h$ in $H$ (when they overlap on exactly one site).
Substituting \eref{Qk-step2-a}, we obtain
\eq{
Q_k \co H
&= \sums h_{1:k} \co (g_1 + g_k) \nx
\whb + \sum_{0 \leq m \leq k-1} (h_{1:m} \co g_m \co h_{m:k-1}) \co (h_{0:1} + h_{k-1:k}) \nx
\whb + t\hb_{1:k-1} \co (h_{0:1} + h_{k-1:k}) + \lo{k-1} \nx
&= \sums - g_1 \co h_{1:k} + h_{1:k} \co g_k \nx
\whb +\sum_{0 \leq m \leq k-1} -\und{h_{0:m} \co g_m \co h_{m:k-1}}{=T^{-1} h_{1:m+1} \co g_{m+1} \co h_{m+1:k}} +
h_{1:m} \co g_m \co h_{m:k} \nx
\whb +t ( - \und{\hb_{0:k-1}}{=T^{-1} \hb_{1:k}} + \hb_{1:k} )
+ \lo{k-1} \nx
&= \lo{k-1}
\lbl{step2-QH}
}
In the second line, we also use that $h_{1:1}$ and $h_{m:m}$ are the identity operator, so the boundary $g \co h$ terms and the $h \co g \co h$ terms cancel pairwise in the translational sum.
Here we use the identity for translational sums, which holds for any $a_{1:l}$,
\eq{
\sums T^{-1} a_{1:l} = \sums a_{1:l}.
}

We now show that the former implies the latter in two parts:
(i) We establish $q_{\bd_1 b_k} = (-1)^k q_{b_1 \bd_k}$, that is, the coefficient of $\hb_{1:k}$ vanishes.
(ii) Any quantity in the same spatial inversion sector as $h_{1:k}$ must take the form \eqref{Qk-step2-a}.

We begin with part (i), proving $q_{\bd_1 b_k} = (-1)^k q_{b_1 \bd_k}$.
The idea is to focus on several relations of the form $r_{\bB}=0$ for operators $\bB$ of length $k$, and then combine them with appropriate coefficients to obtain the desired result.

First, fix an operator $\op{xy}$ with $c_{xy}\neq 0$ and $x+y=\deg(g)(\geq 3)$, and consider all the commutators that produce $\bB_{(1)} = \op{x,y-1}_1 \op{01}_k$.
An operator $\bB$ of length $k$ can arise either from a pair consisting of a $k$-local term in $Q$ (namely, $\op{10}_1 \op{01}_k$ or $\op{01}_1 \op{10}_k$) and an on-site term $g$ in $H$, or from a pair consisting of a $(k-1)$-local term in $Q$ (which is arbitrary at this point) and an interaction term $h$ in $H$.
Taking into account that the degree of $\op{x,y-1}$, the first symbol of $\bB_{(1)}$, is at least two, we find that the only contributing commutators are
$\op{10}_1 \op{01}_k \co \op{xy}_1$ ($= -y \cdot \op{x,y-1}_1 \op{01}_k$)
and
$\op{x,y-1}_1 \op{01}_{k-1} \co \op{10}_{k-1} \op{01}_k$.
As a result, we obtain the relation
\eq{
r_{\bB_{(1)}}
= -y c_{xy} q_{\op{10}_1 \op{01}_k}
  + q_{\op{x,y-1}_1 \op{01}_{k-1}} = 0.
}

Considering a general family of operators $\bB_{(m)}$ ($1 \leq m \leq k$) as
\eq{
\bB_{(m)} = 
\begin{cases}
\op{x,y-1}_1 \op{01}_k & (m=1) \\
\op{01}_1 \op{x,y-2}_m \op{01}_k & (2\leq m \leq k-1) \\
\op{01}_1 \op{x,y-1}_k & (m=k)
\end{cases} ,
}
we obtain
\eq{
r_{\bB_{(m)}}
= 
\begin{cases}
-y c_{xy} q_{\op{10}_1 \op{01}_k}
+ q_{\op{x,y-1}_1 \op{01}_{k-1}} & (m=1) \\
(y-1) q_{\op{x,y-1}_2\op{01}_k}+ q_{\op{01}_1\op{x,y-2}_2\op{01}_{k-1}} & (m=2) \\
q_{\op{01}_2 \op{x,y-2}_m \op{01}_k} + q_{\op{01}_1 \op{x,y-2}_{m-1} \op{01}_{k-1}} & (3 \leq m \leq k-2) \\
q_{\op{01}_2 \op{x,y-2}_{k-1} \op{01}_k} + (y-1) q_{\op{01}_1 \op{x,y-1}_{k-1}} & (m=k-1) \\
q_{\op{01}_2 \op{x,y-1}_k} - y c_{xy} q_{\op{01}_1 \op{10}_k} & (m=k)
\end{cases}.
\lbl{step2-r}
}
Here, for $\bB_{(m)}$ with $2 \leq m \leq k-1$, the contributions arise from two pairs consisting of $k-1$-local terms of $Q$ and the interaction terms of $H$.
The contributions to $\bB_{(m)}$ are similar to those for $\bB_{(1)}$.

Substituting \eref{step2-r} into $-(y-1)r_{\bB_{(0)}} + \sum_{2 \leq m \leq k-1} (-1)^m r_{\bB_{(m)}} + (-1)^k (y-1) r_{\bB_{(k)}} = 0$ (which holds since each $r_{\bB_{(m)}}$ vanishes), we find that many terms cancel, yielding
\eq{
y(y-1) c_{xy} \left( q_{\op{10}_1 \op{01}_k} - (-1)^k q_{\op{01}_1 \op{10}_k}\right) = 0.
}
Here, we have used the translation symmetry of $Q_k$, namely $q_{a_{1:l}} = q_{a_{2:l+1}}$.
From this, we obtain the desired relation $q_{\op{10}_1 \op{01}_k} = (-1)^k q_{\op{01}_1\op{10}_k}$, which shows that the $k$-local part of a charge $Q_k$ is restricted to be a scalar multiple of $h_{1:k}$.
Thus, since we only consider eigenoperators of spatial inversion, it follows that for odd $k$ only the antisymmetric sector is allowed, whereas for even $k$ only the symmetric sector is allowed.

Next, we prove part (ii).
First, recall from \eref{step2-QH} that $Q_k$ of the form \eqref{Qk-step2-a} satisfies $Q_k \co H = \lo{k-1}$.
We then investigate a general $k$-local quantity $Q_k + \Delta$ that also satisfies $(Q_k + \Delta) \co H = \lo{k-1}$.
Here, since $\Delta$ is the difference between two $k$-local quantities, it is itself an at-most-$k$-local quantity and satisfies
\eq{
\Delta \co H = \lo{k-1}.
\lbl{delta-step2}
}

By Step~1, the $k$-local part of $\Delta$ is restricted to the form \eref{Qk-step1-a}. Thus, we may assume that $\Delta$ is at most $k-1$-local.
If $\Delta$ is $k-1$-local, applying Step~1 with $k$ replaced by $k-1$ shows that a translationally invariant $\Delta$ is a linear combination of $h_{1:k-1}$, $\hb_{1:k-1}$, and at-most-$k-2$-local terms.
Among the first two, the one that belongs to the same spatial inversion sector as $h_{1:k}$ is $\hb_{1:k-1}$. Therefore,
\eq{
\Delta \propto  \hb_{1:k-1} + \lo{k-2}.
}
This freedom can be absorbed into the arbitrariness of $t$ in \eref{Qk-step2-a}. Thus, the condition that $Q_k \co H$ is at-most-$k-1$-local is equivalent to \eref{Qk-step2-a} (up to an overall constant factor).
\end{proof}

\section{Classification of the existence/absence of a 3-local charge (Proof of \lref{3local})}
\lbl{s:proof-3local}
In this section, we prove \lref{3local}, namely, we clarify the necessary and sufficient condition on $g$ for the existence of a 3-local charge.
In particular, we will show below that models other than those listed in the statement of \lref{3local} do not admit a 3-local charge.
We have already demonstrated in \eref{explicit-4} and \eref{explicit-3} that the listed models possess a 3-local charge. The case $\deg(g) \leq 2$ is addressed later in \apref{degree2}.

As established in the main text, a model admits a 3-local charge if and only if $g$ satisfies \eref{3local-eq}.
It has already been confirmed in the main text (\ssref{model-4}) that \eref{3local-eq} never holds for $\deg(g)\geq 5$.
We therefore treat separately the cases $\deg(g)=4$ and $\deg(g)=3$ in what follows.

\begin{proof}
Substituting $g = \sum_{x,y} c_{xy} \op{xy}$ into \eref{3local-eq}, we obtain
\eq{
\wha \sum_{x,y} \sum_{x',y'} c_{xy} c_{x'y'} ( (\op{xy}_1 \co h_{1:2}) \co \op{x'y'}_1 + (h_{1:2} \co \op{xy}_2) \co \op{x'y'}_2) \nx
\wha + \sum_{x,y} c_{xy} ( (\op{xy}_1 \co h_{1:2}) \co h_{1:2} + (h_{1:2} \co \op{xy}_2) \co h_{1:2}) \nx
\wha + \sum_{x,y} c_{xy} ( -t \op{xy}_1 \co \hb_{1:2} + t \hb_{1:2} \co \op{xy}_2 ) \nx
&= \lo{1}.
\lbl{ccxy}
}
Here, let us recall that the six terms appearing on the left-hand side of \eref{ccxy} were referred to as the n$_1$ term, n$_2$ term, l$_1$ term, l$_2$ term, t$_1$ term, and t$_2$ term.
The following formulas, which give the explicit expressions of these six terms, will be useful in the subsequent case-by-case analysis:
\eq{
&\text{n$_1$ term}& 
( \op{xy}_1 \co h_{1:2} ) \co \op{x'y'}_1
&= - x(yx'-(x-1)y') \op{x+x'-2,y+y'-1}_1 \op{10}_2 \nx
& & \wha + y((y-1)x'-xy') \op{x+x'-1,y+y'-2}_1 \op{01}_2
+(\text{lower terms}),
\lbl{formula_n1} \\
&\text{n$_2$ term}& 
( h_{1:2} \co \op{xy}_2 ) \co \op{x'y'}_2
&= x(yx'-(x-1)y') \op{10}_1 \op{x+x'-2,y+y'-1}_2 \nx
& & \wha -y((y-1)x'-xy') \op{01}_1 \op{x+x'-1,y+y'-2}_2
+(\text{lower terms}),
\lbl{formula_n2} \\
&\text{l$_1$ term}& 
(\op{xy}_1 \co h_{1:2}) \co h_{1:2}
&= x(x-1) \op{x-2,y}_1 \op{20}_2 - 2xy \op{x-1,y-1}_1 \op{11}_2 + y(y-1) \op{x,y-2}_1 \op{02}_2 +\lo{1},
\lbl{formula_l1} \\
&\text{l$_2$ term}& 
(h_{1:2} \co \op{xy}_2) \co h_{1:2}
&= -x(x-1) \op{20}_1 \op{x-2,y}_2 + 2xy \op{11}_1 \op{x-1,y-1}_2 - y(y-1) \op{02}_1 \op{x,y-2}_2 +\lo{1},
\lbl{formula_l2} \\
&\text{t$_1$ term}& 
-t\op{xy}_1 \co \hb_{1:2} 
&= -tx \op{x-1,y}_1 \op{10}_2 - ty \op{x,y-1}_1 \op{01}_2,
\lbl{formula_t1} \\
&\text{t$_2$ term}& 
t \hb_{1:2} \co \op{xy}_2 
&= tx \op{10}_1 \op{x-1,y}_2 + ty \op{01}_1 \op{x,y-1}_2.
\lbl{formula_t2} \\
}

\subsection{Models of degree 4}
In this subsection, we clarify the condition for the existence of a 3-local charge in the case $\deg(g)=4$, that is,
\eq{
g = \sum_{x+y \le 4} c_{xy}\,\op{xy},
\qquad
\text{with at least one pair $(x,y)$ satisfying $x+y=4$ and $c_{xy}\neq 0$.}
}
In what follows, we use the fact that $r_{\bB}=0$ holds for all $\bB$ with $|\bB|=2$ on the left-hand side of \eref{ccxy} to derive constraints on the coefficients $\{c_{xy}\}$.

First, we have
\eq{
r_{\op{20}_1 \op{11}_2} &= 
\und{-6c_{31}}{\text{l}_1(\op{31})}
\und{-6c_{31}}{\text{l}_2(\op{31})}
= 0 \\
r_{\op{11}_1 \op{02}_2} &= 
\und{+6c_{13}}{\text{l}_1(\op{13})}
\und{+6c_{13}}{\text{l}_2(\op{13})} = 0
}
which immediately imply $c_{31}=c_{13}=0$.
With this in hand, we obtain
\eq{
r_{\op{10}_1 \op{41}_2} &=
\und{-24c_{40}c_{22}}{\text{n}_2(\op{40},\op{22})}
\und{+16c_{22}c_{40}}{\text{n}_2(\op{22},\op{40})}=0 \\
r_{\op{10}_1 \op{23}_2} &=
\und{-48c_{40}c_{04}}{\text{n}_2(\op{40},\op{04})}
\und{+4c_{22}^2}{\text{n}_2(\op{22},\op{22})}=0 \\
r_{\op{10}_1 \op{05}_2} &=
\und{-8c_{22}c_{04}}{\text{n}_2(\op{22},\op{04})}=0
}
Here, if we assume $c_{22} \neq 0$, then from these three equations above we obtain $c_{40} = 0$, $c_{40}c_{04} \neq 0$, and $c_{04} = 0$, which is a contradiction. 
Hence we find $c_{22} = 0$.
Combining this with the second equation above, we obtain $c_{40}c_{04} = 0$, so at least one of $c_{40}$ and $c_{04}$ must vanish.
Together with $c_{31}=c_{13}=0$, this implies that only one degree-4 coefficient can survive.

In what follows, we treat the case $c_{40} \neq 0$ and $c_{04} = 0$.
The alternative case $c_{04} \neq 0$ and $c_{40} = 0$ can be treated analogously by complex conjugation (equivalently, exchanging $x$ and $y$), thus we omit this case.
From
\eq{
r_{\op{10}_1\op{40}_2} &= 
\und{-12c_{40}c_{21}}{\text{n}_2(\op{40},\op{21})}
\und{+8c_{21}c_{40}}{\text{n}_2(\op{21},\op{40})} = 0, \\
r_{\op{10}_1\op{31}_2} &= 
\und{-24c_{40}c_{12}}{\text{n}_2(\op{40},\op{12})}
\und{+8c_{12}c_{40}}{\text{n}_2(\op{12},\op{40})} = 0, \\
r_{\op{10}_1\op{22}_2} &= 
\und{-36c_{40}c_{03}}{\text{n}_2(\op{40},\op{03})} = 0,
}
we find $c_{21}=c_{12}=c_{03}=0$.
Giving this,
\eq{
r_{\op{10}_1\op{21}_2} &= 
\und{-24c_{40}c_{02}}{\text{n}_2(\op{40},\op{02})}
= 0
}
implies $c_{02}=0$ as well.

At this point, the remaining coefficients are $c_{40}(\neq 0), c_{30}, c_{20}, c_{10}, c_{11}$, and $c_{01}$.
Finally, we derive a single relation among these coefficients:
\eq{
r_{\op{10}_1\op{30}_2} &=
\und{-12c_{40}c_{11}}{\text{n}_2(\op{40},\op{11})}
\und{+4c_{11}c_{40}}{\text{n}_2(\op{11},\op{40})}
\und{+4tc_{40}}{\text{t}_2(\op{40})} \nx
&= 4c_{40}(-2c_{11}+t) = 0\\
r_{\op{10}_1\op{20}_2} &=
\und{-12c_{40}c_{01}}{\text{n}_2(\op{40},\op{01})}
\und{-6c_{30}c_{11}}{\text{n}_2(\op{30},\op{11})}
\und{+3c_{11}c_{30}}{\text{n}_2(\op{11},\op{30})}
\und{+6c_{30}}{\text{l}_1(\op{30})}
\und{+3tc_{30}}{\text{t}_2(\op{30})} \nx
&= -12c_{40}c_{01} + 3c_{30}(-c_{11}+2+t) = 0.
}
Eliminating $t$, we obtain
\eq{
c_{30} (c_{11} + 2) = 4 c_{40} c_{01} .
}

Thus, $g$ of degree 4 that admits a 3-local charge is of the form
\eq{
g = c_{40} \op{40} + c_{30} \op{30} + c_{20} \op{20} + c_{10} \op{10} + c_{11} \op{11} + c_{01} \op{01}
\text{ with } c_{30} (c_{11} + 2) = 4 c_{40} c_{01} ,} or its conjugate counterpart:
\eq{
g = c_{04} \op{04} + c_{03} \op{03} + c_{02} \op{02} + c_{01} \op{01} + c_{11} \op{11} + c_{10} \op{10}
\text{ with } c_{03} (c_{11} + 2) = 4 c_{04} c_{10} .}

\subsection{Models of degree 3}
In this subsection, we clarify the condition for the existence of a 3-local charge in the case $\deg(g)=3$, that is,
\eq{
g = \sum_{x+y \le 3} c_{xy}\,\op{xy},
\qquad
\text{with at least one pair $(x,y)$ satisfying $x+y=3$ and $c_{xy}\neq 0$.}
}

First, we have
\eq{
r_{\op{10}_1\op{21}_2} &=
\und{-12c_{30}c_{12}}{\text{n}_2(\op{30},\op{12})}
\und{+2c_{21}^2}{\text{n}_2(\op{21},\op{21})}
\und{+6c_{12}c_{30}}{\text{n}_2(\op{12},\op{30})}=0,\\
r_{\op{10}_1\op{12}_2} &=
\und{-18c_{30}c_{03}}{\text{n}_2(\op{30},\op{03})}
\und{-2c_{21}c_{12}}{\text{n}_2(\op{21},\op{12})}
\und{+4c_{12}c_{21}}{\text{n}_2(\op{12},\op{21})}=0,\\
r_{\op{10}_1\op{03}_2} &=
\und{-6c_{21}c_{03}}{\text{n}_2(\op{21},\op{03})}
\und{+2c_{12}^2}{\text{n}_2(\op{12},\op{12})}=0.
}
From the three equations above, we obtain the following relations among the degree-3 coefficients:
\eq{
c_{21}^2 &= 3c_{30}c_{12} \lbl{3local-c1},\\
c_{21}c_{12} &= 9c_{30}c_{03} \lbl{3local-c2},\\
c_{12}^2 &= 3c_{21}c_{03}\lbl{3local-c3}.
}

Next, we have
\eq{
r_{\op{10}_1\op{20}_2} &=
\und{-6c_{30}c_{11}}{\text{n}_2(\op{30},\op{11})}
\und{+3c_{11}c_{30}}{\text{n}_2(\op{11},\op{30})}
\und{+4c_{21}c_{20}}{\text{n}_2(\op{21},\op{20})}
\und{-2c_{20}c_{21}}{\text{n}_2(\op{20},\op{21})}
\und{+6c_{30}}{\text{l}_1(\op{30})}
\und{+3tc_{30}}{\text{t}_2(\op{30})}\nx
&= + 2c_{21}c_{20} -3c_{30}c_{11} +(6+3t)c_{30} =0
\lbl{3local-r1},\\
r_{\op{10}_1\op{11}_2} &=
\und{-12c_{30}c_{02}}{\text{n}_2(\op{30},\op{02})}
\und{+2c_{11}c_{21}}{\text{n}_2(\op{11},\op{21})}
\und{+4c_{12}c_{20}}{\text{n}_2(\op{12},\op{20})}
\und{-4c_{20}c_{12}}{\text{n}_2(\op{20},\op{12})}
\und{-4c_{21}}{\text{l}_1(\op{21})}
\und{+2tc_{21}}{\text{t}_2(\op{21})}\nx
&= +2c_{21}c_{11} - 12c_{30}c_{02} + (-4+2t)c_{21} =0
\lbl{3local-r2},\\
r_{\op{10}_1\op{02}_2} &=
\und{-4c_{21}c_{02}}{\text{n}_2(\op{21},\op{02})}
\und{+2c_{12}c_{11}}{\text{n}_2(\op{12},\op{11})}
\und{+c_{11}c_{12}}{\text{n}_2(\op{11},\op{12})}
\und{-6c_{20}c_{03}}{\text{n}_2(\op{20},\op{03})}
\und{+2c_{12}}{\text{l}_1(\op{12})}
\und{+tc_{12}}{\text{t}_2(\op{12})}\nx
&=- 6c_{03}c_{20} + 3c_{12}c_{11} - 4c_{21}c_{02} + (2+t)c_{12} =0.
\lbl{3local-r3}
}
Substituting these into
\eq{
9c_{30}c_{03}r_{\op{10}_1\op{20}_2}
-c_{21}^2r_{\op{10}_1\op{11}_2}
+3c_{30}c_{21}r_{\op{10}_1\op{02}_2}=0
}
(since $r_{\bB} = 0$ holds for any $\bB$), we arrive at
\eq{
(-27c_{11}+54+27t)c_{30}^2c_{03}
+(-2c_{11}+4-2t)c_{21}^3
+(9c_{11}+6+3t)c_{30}c_{21}c_{12}=0.
}
Combined with \eref{3local-c1}, \eref{3local-c2}, and \eref{3local-c3}, this gives
\eq{
8c_{21}^3=0,
}
which implies $c_{21}=0$.

Substituting this back into \eref{3local-c1}, \eref{3local-c2}, and \eref{3local-c3}, we obtain $c_{12}=0$ and $c_{30}c_{03}=0$.
Thus, among the degree-3 coefficients, only one of $c_{30}$ or $c_{03}$ can be nonzero.
In the case $c_{30}\neq 0$ and $c_{03}=0$, \eref{3local-r2} further implies $c_{\op{02}}=0$.

Hence, $g$ of degree 3 admitting a 3-local charge must take the form
\eq{
g=c_{30}\op{30}+c_{20}\op{20}+c_{10}\op{10}+c_{11}\op{11}+c_{01}\op{01} ,
}
or its conjugate counterpart,
\eq{
g=c_{03}\op{03}+c_{02}\op{02}+c_{01}\op{01}+c_{11}\op{11}+c_{10}\op{10} .
}
\end{proof}

\section{Detailed argument on Type~N$^+$ (Proof of \lref{degree4-N+})}
\lbl{s:proof-degree4-N+}
We now complete the proof that the Hamiltonian \eqref{model4} with $c_{11} \neq 0$ and $c_{40} \neq 0$ belongs to Type~N$^+$, i.e., it possesses a 3-local charge but no $k$-local charges for $4 \leq k \leq N/2$.
In the main text, we have already shown the existence of a 3-local charge (\eref{explicit-4}) and the absence of 4-local charges (\ssref{model-4}).
In what follows, we make the heuristic discussion in \ssref{model-4} precise by establishing the absence of $k$-local charges for general $k \geq 5$.

\begin{proof}
First, the condition that $Q_k$ satisfies Step~3, i.e., $Q_k \co H = \lo{k-2}$ holds, is equivalent to the statement that $Q_k$ can be written (up to an overall constant factor) in the following form for some constant $t$:
\eq{
Q_k &= 
\sums h_{1:k} + 2(k-1)c_{11} \hb_{1:k-1} + th_{1:k-2} \nx
\whb +(1+(-1)^k) c_{20} (\bd_1 \bd_{k-1} + (2k-3)c_{11}\bd_1 \bd_{k-2}) \nx
\whb + \sum_{1\leq m \leq k-2} 48(-1)^{k-m+1}(k-2m-1)\sigma_{1,m,m,k-2} c_{40}c_{11} \bd_1 (\bd_m)^2 \bd_{k-2} \nx
\wha + c_{40} D_{k-1,0} 
+ c_{40} D_{k-2,1} 
+ \lo{k-3}  ,
\lbl{Qk-proof-N+}
}
where $\sigma_{i_1,i_2,i_3,i_4}$ and $D_{l,n}$ are defined in Eqs.~\eqref{sigma-def} and \eqref{D-def}, respectively.
The equivalence can be established by an argument similar to that used in the case $k=4$ in the main text.

Next, we show that this $Q_k$ is not a charge, regardless of the form of its $k-3$-local and lower-local parts.
The coefficient of $\bd_1 \bd_m \bd_{m+1} \bd_{k-2}$ in the expansion of $Q_k \co H$ is calculated as
\eq{
r_{\bd_1\bd_m\bd_{m+1}\bd_{k-2}} =
\begin{cases}
\displaystyle
24(-1)^k(k-9)c_{40}c_{11} -\frac{q_{\bd_1\bd_m\bd_{m+1}\bd_{k-3}}}{2\sigma_{1,m,m+1,k-3}}& (m=1) \\
\displaystyle
192(-1)^{k-m}c_{40}c_{11} 
-\frac{q_{\bd_1\bd_{m-1}\bd_m\bd_{k-3}}}{\sigma_{1,m-1,m,k-3}}
-\frac{q_{\bd_1\bd_m\bd_{m+1}\bd_{k-3}}}{\sigma_{1,m,m+1,k-3}}
& (2 \leq m \leq k-4) \\
\displaystyle
24(k-9)c_{40}c_{11} 
-\frac{q_{\bd_1\bd_{m-1}\bd_m\bd_{k-3}}}{2\sigma_{1,m-1,m,k-3}}
& (m=k-3)
\end{cases}
}
In order to eliminate all the undetermined coefficients of $k-3$-local operators (namely, $q_{\bd_1 \bd_m \bd_{m+1} \bd_{k-3}}$ for $1 \leq m \leq k-4$), we evaluate the following weighted sum:
\eq{
\sum_{1\leq m \leq k-3} \frac{(-1)^{k-3-m}r_{\bd_1\bd_m\bd_{m+1}\bd_{k-2}}}{2 \sigma_{1,m,m+1,k-2}} = -48(k-1) c_{40}c_{11} \neq 0
}
Since the commutation condition requires $r_{\bB}=0$ for all $\bB$, we conclude that the $Q_k$ given in \eref{Qk-proof-N+} is not a charge.
Therefore, the Hamiltonian~\eqref{model4} (and hence the Hamiltonian~\eqref{model4-raw}) with $c_{11} \neq 0$ and $c_{40} \neq 0$ admits no $k$-local charges for $4 \leq k \leq N/2$.
\end{proof}

\section{Algorithmic proof of integrability (Proof of \lref{degree3-Ca})}
\lbl{s:proof-degree4-C}

We now complete the proof that the Hamiltonian \eqref{model3-Ca} is completely integrable (Type~C), i.e, it possesses $k$-local charges for all $3 \leq k \leq N/2$.
In the main text, we have already established the existence of 3-local and 4-local charges (\eref{explicit-3} and \lref{degree3-4local}).
In what follows, going beyond the simple setup in \ssref{model-3}, we revisit the $k\geq 5$ construction in a general setting and give a complete proof.

\begin{proof}
Unlike in the proof for the degree-4 Type~C model (\lref{degree4-C}) and for the other degree-3 Type~C model (\lref{degree3-Cb}), it is difficult to write down the local charges of the Hamiltonian \eqref{model3-Ca} in closed form.
Instead, for each $k$, we present an algorithm to construct a $k$-local charge and prove its validity.

We shall refer to the translation sum of an operator containing $x$ occurrences of $\bd$ and $y$ occurrences of $b$ as an $\bx{xy}$ quantity.
For example, $\sums \bd_1 \bd_2 \bd_3$ and $\sums (\bd_1)^3$ are $\bx{30}$ quantities; $\sums h_{1:k}$, $\sums \hb_{1:k}$, and $\sums n_1$ are $\bx{11}$ quantities.
A $\bx{10}$ quantity is simply $\sums \bd_1$ up to a constant multiple.

We construct the local charge $Q_k$ of the Hamiltonian \eqref{model3-Ca} as a linear combination of $\bx{30}$, $\bx{11}$, and $\bx{10}$ quantities.

Here, we note the following fact:
\eq{
\left(\bx{30}~\text{quantities}\right) \co H &= \left(\bx{30}~\text{quantities}\right)
\lbl{bx-quantity-1} \\
\left(\bx{11}~\text{quantities}\right) \co H &= \left(\bx{30}~\text{quantities}\right) + \left(\bx{10}~\text{quantities}\right)
\lbl{bx-quantity-2} \\
\left(\bx{10}~\text{quantities}\right) \co H &= \left(\bx{10}~\text{quantities}\right)
\lbl{bx-quantity-3}
}
This follows from two simple observations:
(i) The Hamiltonian \eqref{model3-Ca} is a sum of $\bx{30}$, $\bx{11}$, and $\bx{10}$ quantities.
(ii) The commutator of $\bx{xy}$ quantities and $\bx{x'y'}$ quantities yields $\bx{x+x'-z,y+y'-z}$ quantities for positive integers $z$.
Owing to Eqs.~\eqref{bx-quantity-1}--\eqref{bx-quantity-3}, the only nontrivial constraints on $Q_k$ come from the requirement that the coefficients $r_{\bB}$ of the $\bx{30}$ and $\bx{10}$ quantities vanish in $Q_k \co H = \sum_{\bB} r_{\bB} \bB$.

In what follows, we determine the coefficients $q_{\bA}$ of $Q_k = \sum_{\bA} q_{\bA} \bA$, starting from those with larger $|\bA|$, where $|\bA|$ denotes the length of $\bA$.
From the analysis up to Step~2 (i.e., \lref{step2}), we already know that $Q_k$ can be written in the following form:
\eq{
Q_k
&= \und{\sums \left( h_{1:k}
+ \sum_{1 \leq m \leq k-1} 6(-1)^{k-1-m} \sigma_{1,m,k-1} c_{30} \bd_1 \bd_m \bd_{k-1} \right)}{=:Q_k'}
+ t \hb_{1:k-1} + \lo{k-2} .
\lbl{QkQ'k}
}
with the symmetry factor $\sigma_{i_1,i_2,i_3}$ for integers $i_1 \leq i_2 \leq i_3$ given by
\eq{
\sigma_{i_1,i_2,i_3} := 
\begin{cases}
1 & (i_1 < i_2 < i_3) \\
\tfrac{1}{2!} & (i_1 = i_2 < i_3) \\
\tfrac{1}{2!} & (i_1 < i_2 = i_3) \\
\tfrac{1}{3!} & (i_1 = i_2 = i_3)
\end{cases} .
}
We first compute the coefficients of $k-2$-local parts of $Q_k$ together with the constant $t$ in \eref{QkQ'k}.
For this purpose, we use the condition that the coefficients of $k-1$-local parts of $Q_k \co H$ (namely, $r_{\bB}$ with $|\bB|=k-1$) must vanish.
Among the operators $\bB$ of length $k-1$, it suffices to consider those corresponding to $\bx{30}$ quantities, which are calculated as
\eq{
r_{\bd_1 \bd_m \bd_{k-1}} 
&= \begin{cases}
3c_{30} q_{b_1 \bd_{k-1}} - q_{(\bd_1)^2 \bd_{k-2}}
+ \text{const} & (m=1) \\
\displaystyle 
- \frac{q_{\bd_1\bd_{m-1}\bd_{k-2}}}{\sigma_{1,m-1,k-2}}
- \frac{q_{\bd_1\bd_{m}\bd_{k-2}}}{\sigma_{1,m,k-2}}
+ \text{const} & (2 \leq m \leq k-2) \\
3c_{30} q_{\bd_1 b_{k-1}} - q_{\bd_1 (\bd_{k-2})^2}
+ \text{const} & (m=k-1)
\end{cases} .
\lbl{rbbb}
}
Here, the constant term arises from the expansion of $Q_k' \co H$ ($Q_k'$ is defined in \eref{QkQ'k}).
The equations $r_{\bd_1 \bd_m \bd_{k-1}}=0$ ($1 \leq m \leq k-1$) involve $k$ unknown coefficients,
namely $q_{\bd_1 \bd_m \bd_{k-2}}$ ($1 \leq m \leq k-2$), $q_{\bd_1 b_{k-1}}$, and $q_{b_1 \bd_{k-1}}$.
Since there are $k-1$ equations and the coefficient matrix has full rank, the system admits a solution for any given constant term.
Note also that a solution satisfying the parity condition under spatial inversion always exists, since the constant term in \eref{rbbb} is compatible with this condition.
The explicit construction of such a solution is straightforward.

Solving the equations $r_{\bd_1 \bd_m \bd_{k-1}}=0$ in \eref{rbbb} yields
\eq{
q_{\bd_1 \bd_m \bd_{k-2}} &= (9c_{11}-6) c_{30} (-1)^{k-m+1} (k-2m-1) \sigma_{1,m,k-2}, \\
t \left(= q_{\bd_1 b_{k-1}} = (-1)^k q_{b_1 \bd_{k-1}} \right) &= (k-1) (\tfrac{3}{2}c_{11}-1).
}
Substituting these coefficients, the general solution at Step~3, i.e., a solution $Q_k$ satisfying $Q_k \co H = \lo{k-2}$, can be expressed as
\eq{
Q_k &= \und{Q_k' + 
\sums t \hb_{1:k-1}
+\sum_{1 \leq m \leq k-2} q_{\bd_1 \bd_m \bd_{k-2}} \bd_1 \bd_m \bd_{k-2}}
{=:Q_k''}
+ t' h_{1:k-2} + \lo{k-3}.
}

Subsequently, by adjusting $t'$ together with the coefficients of the quantities $\bA$ with $|\bA|=k-3$ so as to cancel the coefficients of $\bB$ with $|\bB|=k-2$, we obtain the solution at Step~4.

By continuing this algorithmic procedure step by step, we eventually obtain the solution up to Step~$k$, namely an operator $\tilde{Q}_k$ satisfying $\tilde{Q}_k \co H = \lo{1}$.
In the present setting, the quantities $\bB$ are restricted to the $\bx{30}$ and $\bx{11}$ sectors (see Eqs.~\eqref{bx-quantity-1}--\eqref{bx-quantity-3}), and therefore $\tilde{Q}_k \co H$ must be a linear combination of the unique 1-local $\bx{30}$ quantity $\sums (\bd_1)^3$ and the unique 1-local $\bx{11}$ quantity $\sums n_1$:
\eq{
\tilde{Q}_k \co H \in \operatorname{span} \left\{\sums (\bd_1)^3,\ \sums n_1\right\}.
}
As the final step, a genuine $k$-local charge $Q_k$ can be constructed by adding an appropriate 1-local operator to $\tilde{Q}_k$.
Indeed, since
\eq{
\sums n_1 \co H &= \sums 3c_{30} (\bd_1)^3 + c_{10} \bd_1, \\
\sums \bd_1 \co H &= \sums (-c_{11}-2) \bd_1
}
hold, the contribution from the commutator $\tilde{Q}_k \co H$ can always be cancelled out, provided $c_{11}\neq -2$ (with $c_{30}\neq 0$ as assumed in this section).

In the exceptional case $c_{11}=-2$, the adjustment by $\sums \bd_1$ is no longer available, since $\sums \bd_1 \co H$ is zero.
At this point, we can simply construct the charge as a linear combination of $\tilde{Q}_k$ and $\sums n_1$.
The reason is as follows.
Using Eqs.~\eqref{r-bx-op1} and \eqref{r-bx-op4} (as displayed in \ssref{proof-degree3-C-}; we only use established relations, so no circularity is involved) at $c_{01}=0$ and $c_{11}=-2$, we obtain
\eq{
r_{\bx{30}} &= 3 c_{30} q_{\bx{11}}, \\
r_{\bx{10}} &= c_{10} q_{\bx{11}}.
}
Therefore, imposing $r_{\bx{30}}=0$ forces $q_{\bx{11}}=0$ (since $c_{30}\neq 0$), and then $r_{\bx{10}}=0$ follows automatically.
Consequently, once the coefficient of $\sums (\bd_1)^3$ in $Q_k \co H$ is set to zero, the coefficient of $\sums \bd_1$ in $Q_k \co H$ vanishes as well.

Thus, we have established the existence of all $k$-local charges $Q_k$ of the Hamiltonian \eqref{model3-Ca}, which are eigenoperators of spatial inversion by construction.
\end{proof}

\section{Detailed argument on Type~C$^-$ (Proof of \lref{degree3-C-})}
\lbl{s:proof-degree3-C-}
We now complete the proof that the Hamiltonian \eqref{model3} with $c_{01} \neq 0$, $c_{11} \neq 1$, and $c_{30} \neq 0$ belongs to Type~C$^-$, namely, that it possesses $k$-local charges for all $5 \leq k \leq N/2$ and $k=3$, but lacks a 4-local charge.
In the main text, we have already shown the existence of a 3-local charge (\eref{explicit-3}) and the absence of 4-local charges (\lref{degree3-4local}).
In what follows, we make the heuristic discussion in \ssref{model-4} precise by establishing the absence of $k$-local charges for general $k \geq 5$.

\begin{proof}
By an argument similar to the proof of \lref{degree3-Cb}, we present an algorithm for constructing $k$-local charges $Q_k$ for $k \geq 5$.
We construct $Q_k$ as a linear combination of $\bx{30}$, $\bx{11}$, $\bx{20}$, $\bx{10}$, and $\bx{01}$ quantities.

The algorithm for constructing $Q_k$ proceeds in essentially the same manner as in \lref{degree3-Cb}, yielding a solution up to Step $k$, namely a $\tilde{Q}_k$ that satisfies $\tilde{Q}_k \co H = \lo{1}$.
The only difference between the present algorithm and that in \lref{degree3-Cb} lies in the set of conditions imposed at each step.
In \lref{degree3-Cb}, to solve the system of equations it suffices to use only the conditions $r_{\bB}=0$ associated with certain $\bx{30}$ quantities $\bB$, which consist of the coefficients $q_{\bA}$ of the $\bx{30}$ and $\bx{11}$ quantities $\bA$.
By contrast, in the present lemma, the system consists of the conditions $r_{\bB}=0$ not only for the $\bx{30}$ quantities but also for the $\bx{20}$ quantities of $\bB$, with the variables being the coefficients $q_{\bA}$ of the $\bx{30}$, $\bx{11}$, and $\bx{20}$ quantities $\bA$.
In spite of this difference, the underlying reason why this algorithm works is the same: the system of linear equations has more variables than equations, and thus admits a solution.

We now introduce a shorthand notation.
Let $q_{\bx{xy}}$ denote the sum of the coefficients of all $\bx{xy}$ quantities appearing in the expansion of $Q$, and let $r_{\bx{xy}}$ denote the corresponding sum of coefficients in the expansion of $Q \co H$. Namely, we define
\eq{
q_{\bx{xy}} &:= \sum_{\bA:\,\bA\text{ is a }\bx{xy}\text{ quantity}} q_{\bA},\nx
r_{\bx{xy}} &:= \sum_{\bB:\,\bB\text{ is a }\bx{xy}\text{ quantity}} r_{\bB}.
}
With this notation, the following identities hold:
\eq{
r_{\bx{30}} &= 3 c_{30} q_{\bx{11}} - 3(c_{11}+2) q_{\bx{30}} 
\lbl{r-bx-op1},\\
r_{\bx{20}} &= -3 c_{01} q_{\bx{30}} -2(c_{11}+2) q_{\bx{20}} + 3 c_{30} q_{\bx{01}} 
\lbl{r-bx-op2},\\
r_{\bx{00}} &= -c_{01} q_{\bx{10}} + c_{10} q_{\bx{01}} 
\lbl{r-bx-op3},\\
r_{\bx{10}} &= c_{10} q_{\bx{11}} - 2c_{01} q_{\bx{20}} - (c_{11}+2) q_{\bx{10}} 
\lbl{r-bx-op4},\\
r_{\bx{01}} &= -c_{01} q_{\bx{11}} + (c_{11}+2) q_{\bx{01}}
\lbl{r-bx-op5}.
}
These identities hold because the Hamiltonian of the model is given by \eref{model3}, and because a commutator of $\bx{xy}$ quantities with $\bx{x'y'}$ quantities consists of $\bx{x+x'-z,y+y'-z}$ quantities for positive integers $z$, with expansion coefficients identical to those specified in \eref{comxy-formula}.

At the end of Step~$k$, all coefficients $r_{\bB}$ with $|\bB|\ge 2$ vanish, and any obstruction to $Q_k \co H=0$ is 1-local.
Thus, for the Step-$k$ solution $\tilde{Q}_k$, we have
\eq{
\tilde{Q}_k \co H \in \operatorname{span} \left\{
\und{\sums \op{30}_1}{\text{the unique}\\\bx{30}\text{ quantity}},
\und{\sums \op{20}_1}{\text{the unique}\\\bx{20}\text{ quantity}},
\und{\sums \op{00}_1}{\text{the unique}\\\bx{00}\text{ quantity}},
\und{\sums \op{10}_1}{\text{the unique}\\\bx{10}\text{ quantity}},
\und{\sums \op{01}_1}{\text{the unique}\\\bx{01}\text{ quantity}}
\right\}.
}
Therefore, $Q_k \co H=0$ is equivalent to the requirement that Eqs.~\eqref{r-bx-op1}--\eqref{r-bx-op5} vanish.

Here, by adding to $\tilde{Q}_k$ (the solution up to Step~$k$, as defined at the beginning of the proof) a suitable linear combination $\Delta$ of $\sums n_1$, $\sums b_1$, and $\sums \bd_1$, we can eliminate the coefficients $r_{\bx{30}}, r_{\bx{20}},$ and $r_{\bx{00}}$.
The reason is that $\sums n_1 \co H$, $\sums b_1 \co H$, and $\sums \bd_1 \co H$ are all 1-local, so they do not affect any steps prior to Step~$k$ and can therefore be used to freely adjust $q_{\bx{11}}, q_{\bx{01}}$, and $q_{\bx{10}}$, respectively.

In particular, by setting
\eq{
q_{\bx{11}} &= \frac{c_{11}+2}{c_{30}} q_{\bx{30}}, \\
q_{\bx{01}} &= \frac{c_{01}}{c_{30}} q_{\bx{30}} + \frac{2(c_{11}+2)}{3c_{30}} q_{\bx{20}}, \\
q_{\bx{10}} &= \frac{c_{10}}{c_{30}} q_{\bx{30}} + \frac{2(c_{11}+2)c_{10}}{3c_{30}c_{01}} q_{\bx{20}},
}
the coefficients $r_{\bx{30}}, r_{\bx{20}},$ and $r_{\bx{00}}$ all vanish.

For the adjusted operator $\tilde{Q}_k' := \tilde{Q}_k + \Delta$, its commutator with the Hamiltonian reads
\eq{
\tilde{Q}_k' \co H
= q_{\bx{20}} \sums -2 \left( c_{01} + \frac{(c_{11}+2)^2c_{10}} {3c_{30}c_{01}} \right) \bd_1  + \frac{2(c_{11}+2)^2}{3c_{30}} b_1.
}
From this expression, we see that $Q \co H$ reduces to a 1-local operator independent of $k$.
As shown in \lref{degree3-4local}, if $c_{01} \neq 0$ and $c_{11} \neq 1$ hold, then $Q'_4 \co H$ never vanishes.
Therefore, for any $k \geq 5$, by adding an appropriate multiple of $\tilde{Q}_4'$ to $\tilde{Q}_k'$, we obtain a $k$-local charge.

In conclusion, even in the case where no 4-local charge exists, the Hamiltonian~\eqref{model3} possesses $k$-local charges for every integer $k \geq 5$.
Indeed, the very nonexistence at $k=4$ is what enables the construction above.
\end{proof}

\section{Derivation of formulas used in \eref{cal-absence34}}
\lbl{s:useful}
In this section, we provide proofs of the formulas used in \eref{cal-absence34}.

For the proofs of these formulas, we make use of what we call the {\it Hokkyo identity}
\cite{hokkyo2025rigorous}:
\eq{
\und{(\alpha \co \beta \co \gamma)}{=(\alpha \co \beta) \co \gamma = \alpha \co (\beta \co \gamma)} \co \beta
&= \half ((\alpha \co \beta) \co \gamma) \co \beta + \half (\alpha \co (\beta \co \gamma)) \co \beta \nx
&= \half ((\alpha \co \beta) \co \beta) \co \gamma + \half \cancel{(\alpha \co \beta) \co (\gamma \co \beta)}
+ \half \cancel{(\alpha \co \beta) \co (\beta \co \gamma)} + \half \alpha \co ((\beta \co \gamma) \co \beta) \nx
&= \half ((\alpha \co \beta) \co \beta) \co \gamma + \half \alpha \co ((\beta \co \gamma) \co \beta).
}
Here, the second equality uses the Jacobi identity
$(\alpha \co \beta) \co \gamma = (\alpha \co \gamma) \co \beta + \alpha \co (\beta \co \gamma)$.
In what follows, the use of the Hokkyo identity will be indicated by $\akkun$.

With this notation, we obtain the formulas used in \eref{cal-absence34} as
\eq{
( g_1 \co h_{1:3}) \co h_{1:2}
&= (g_1 \co h_{1:2} \co h_{2:3}) \co h_{1:2} \nx
&\akkun \half ((g_1 \co h_{1:2}) \co h_{1:2}) \co h_{2:3} + \half g_1 \co 
\und{((h_{1:2} \co h_{2:3}) \co h_{1:2})}{=-h_{2:3}} \nx
&= \half ((g_1 \co h_{1:2}) \co h_{1:2}) \co h_{2:3} \lbl{useful-1},\\
(h_{1:3} \co g_3) \co h_{2:3}
&\akkun \half h_{1:2} \co ((h_{2:3} \co g_3) \co h_{2:3} \lbl{useful-2},\\
(h_{1:2} \co g_2 \co h_{2:3}) \co g_2 
&\akkun \half ((h_{1:2} \co g_2) \co g_2) \co h_{2:3} + \half h_{1:2} \co ((g_2 \co h_{2:3}) \co g_2) \lbl{useful-3},\\
(h_{1:2} \co g_2 \co h_{2:3}) \co h_{1:2} 
&= ((h_{1:2} \co g_2) \co h_{2:3}) \co h_{1:2} \nx
&\akkun ((h_{1:2} \co g_2) \co h_{1:2}) \co h_{2:3}
+ \und{(h_{1:2} \co g_2) \co (h_{2:3} \co h_{1:2})}
{=g_2 \co h_{2:3}} \nx
&= ((h_{1:2} \co g_2) \co h_{1:2}) \co h_{2:3} + \lo{2} \lbl{useful-4},\\
(h_{1:2} \co g_2 \co h_{2:3}) \co h_{2:3}
&\akkun h_{1:2} \co ((g_2 \co h_{2:3}) \co h_{2:3}) + \lo{2} \lbl{useful-5}.
}
Here, \eref{useful-2} is the reflection of \eref{useful-1}, and \eref{useful-5} is the reflection of \eref{useful-4}.

\section{Detailed argument on Type~N (Proof of \lref{absence3k})}
\lbl{s:proof-GMC1}
We now present the detailed proof of \lref{absence3k}.
As stated in the proof in the main text, the key point is the Step~3 descent: from a $k$-local quantity $Q_k$ satisfying Step~3, we construct a $(k-1)$-local quantity $Q_{k-1}$ that satisfies Step~3.
The main text gave only a rough argument for this descent. In this appendix, we provide the full derivation and thereby complete the proof for general $k$.

Based on the analyses for general $k$ in Step~1 and Step~2, the operator $Q_k$ can be expressed as follows:
\eq{
Q_k
&= \sums h_{1:k}
+ \sum_{1 \leq m \leq k-1} h_{1:m} \co g_m \co h_{m:k-1}
+ \tk \hb_{1:k-1} + \sk_{1:k-2}
+ \lo{k-3}
}
(a restatement of \eref{prep-absence3k}).

In the analysis of Step~3, we focus on the $k-1$-local part of $Q_k \co H$, which can be written as
\eq{
Q_k \co H
&= \sums (g_1 \co h_{1:k-1}) \co (g_1 + h_{1:2}) \nx
\whb + \sum_{2 \leq m \leq k-2} (h_{1:m} \co g_m \co h_{m:k-1}) \co (g_m + h_{m-1:m} + h_{m:m+1}) \nx
\whb + (h_{1:k-1} \co g_{k-1}) \co (g_{k-1} + h_{k-1:k-2}) \nx
\whb + \tk (- g_1 \co \hb_{1:k-1} + h_{1:k-1} \co g_{k-1}) \nx
\whb + \sk_{1:k-2} \co h_{k-1:k-2} - h_{1:2} \co \sk_{2:k-1} +\lo{k-2} \nx
&= \sums \dk_{1:k-1} + \ek_{1:k-1} + \fk_{1:k-1} \nx
\whb +\tk (- g_1 \co \hb_{1:k-1} + h_{1:k-1} \co g_{k-1}) + \sk_{1:k-2} \co h_{k-1:k-2} - h_{1:2} \co \sk_{2:k-1} + \lo{k-2} ,
\lbl{GMC1-QH}
}
where we have introduced
\eq{
\dk_{1:k-1} &:= 
 \sum_{2 \leq m \leq k-1} (h_{1:m} \co g_m \co h_{m:k-1}) \co g_1
+ \sum_{1 \leq m \leq k-2} (h_{1:m} \co g_m \co h_{m:k-1}) \co g_{k-1} \lbl{dk-def},\\
\ek_{1:k-1} &:=
\sum_{1 \leq m \leq k-1} (h_{1:m} \co g_m \co h_{m:k-1}) \co g_m \lbl{ek-def},\\
\fk_{1:k-1} &:=
\sum_{2 \leq m \leq k-1} (h_{1:m} \co g_m \co h_{m:k-1}) \co h_{m-1:m} 
+ \sum_{1 \leq m \leq k-2} (h_{1:m} \co g_m \co h_{m:k-1}) \co h_{m:m+1} \lbl{fk-def}.
}

We now ``split'' the right-hand side of \eref{GMC1-QH} into the form of \eref{chi-psi-h}.
The splitting is not unique, and some care is required to make the subsequent arguments work properly.

First, we split $\dk_{1:k-1}$ in \eref{dk-def}.
Defining $\nabla_{i_1,i_2,i_3,i_4} := h_{i_1:i_2} \co g_{i_2} \co h_{i_2:i_3} \co g_{i_3} \co h_{i_3:i_4}$, we obtain
\eq{
\dk_{1:k-1}
&= - \sum_{1 < m' \leq k-1} \nabla_{1,1,m',k-1}
+ \sum_{1 \leq m < k-1} \nabla_{1,m,k-1,k-1} \nx
&= - \sum_{1 = m < m' \leq k-1} \nabla_{1,m,m',k-1}
+ \sum_{1 \leq m < m' = k-1} \nabla_{1,m,m',k-1} \nx
&= - \cancel{\sum_{1 \leq m < m' \leq k-1} \nabla_{1,m,m',k-1}}
+ \sum_{1 < m < m' \leq k-1} \nabla_{1,m,m',k-1}
+ \cancel{\sum_{1 \leq m < m' \leq k-1} \nabla_{1,m,m',k-1}}
- \sum_{1 \leq m < m' < k-1} \nabla_{1,m,m',k-1} \nx
&= h_{1:2} \co \sum_{2 \leq m < m' \leq k-1} \nabla_{2,m,m',k-1}
- \sum_{1 \leq m < m' \leq k-2} \nabla_{1,m,m',k-2} \co h_{k-2:k-1} \nx
&= \addl{\ddk_{1:k-2}} \co h_{k-2:k-1} + h_{1:2} \co \addr{-\ddk_{2:k-1}} ,
}
where we define
\eq{
\ddk_{1:k-2} := - \sum_{1 \leq m < m' \leq k-2} \nabla_{1,m,m',k-2} .
}

Next, we split $\ek_{1:k-1}$ in \eref{dk-def}.
To carry out this splitting, we first prepare some formulas required for the decomposition:
\eq{
(h_{1:2} \co g_2 \co h_{2:k-1}) \co g_2
&\akkun \half h_{1:2} \co ( ( g_2 \co h_{2:k-1} ) \co g_2 )
+ \half ( ( h_{1:2} \co g_2 ) \co g_2 ) \co h_{2:k-1} \nx
&= \half h_{1:2} \co ( ( g_2 \co h_{2:k-1} ) \co g_2 ) 
+ \half (h_{1:k-1} \co g_{k-1}) \co g_{k-1}
\nx
\wha + \half \sum_{2 \leq m \leq k-2} ( ( h_{1:m} \co g_m ) \co g_m ) \co h_{m:k-1}
- \half \sum_{3 \leq m \leq k-1} ( ( h_{1:m} \co g_m ) \co g_m ) \co h_{m:k-1} \nx
&= \half h_{1:2} \co ( (g_2 \co h_{2:k-1}) \co g_2 + (h_{2:k-1} \co g_{k-1}) \co g_{k-1}) + \erk_{1:k-2} \co h_{k-2:k-1} - h_{1:2} \co \erk_{2:k-1} 
\lbl{erk-elk-1},
\\
(h_{1:k-2} \co g_{k-2} \co h_{k-2:k-1}) \co g_{k-2}
&\akkun \half (( g_1 \co h_{1:k-2} ) \co g_1 + (h_{1:k-2} \co g_{k-2}) \co g_{k-2} ) \co h_{k-2:k-1}
+ \elk_{1:k-2} \co h_{k-2:k-1} - h_{1:2} \co \elk_{2:k-1}
\lbl{erk-elk-2}},
where we define
\eq{
\erk_{1:k-2} &:= \half \sum_{2 \leq m \leq k-2} (( h_{1:m} \co g_m ) \co g_m ) \co h_{m:k-2}, \\
\elk_{1:k-2} &:= \tfrac{-1}{2} \sum_{1 \leq m \leq k-3} h_{1:m} \co ( ( g_m \co h_{m:k-2} ) \co g_m) .
}
Here, \eref{erk-elk-2} is the reflection of \eref{erk-elk-1}.
Applying \eref{erk-elk-1} at $m = 2$ and \eref{erk-elk-2} at $m = k - 2$,
the term $\ek_{1:k-1}$ in \eref{ek-def} can be split, for each $m$, as follows:
\eq{
& m=1 & (g_1 \co h_{1:k-1}) \co g_1
&= \addl{((g_1 \co h_{1:k-2}) \co g_1)} \co h_{k-2:k-1},\\
& m=2 & (h_{1:2} \co g_2 \co h_{2:k-1}) \co g_2
&= (\tfrac{k-3}{k-2} + \tfrac{1}{k-2}) (h_{1:2} \co g_2 \co h_{2:k-1}) \nx
&&&= \addl{\tfrac{k-3}{k-2} ((h_{1:2} \co g_2 \co h_{2:k-2}) \co g_2)} \co h_{k-2:k-1} \nx
&&\wha + h_{1:2} \co \addr{\tfrac{1}{2(k-2)} ((g_2 \co h_{2:k-1}) \co g_2 + (h_{2:k-1} \co g_{k-1}) \co g_{k-1})} \nx
&&\wha + \addl{\tfrac{1}{k-2} \erk_{1:k-2} \co h_{k-2:k-1}}
+ h_{1:2} \co \addr{-\tfrac{1}{k-2} \erk_{2:k-1}},\\
& 3 \leq m \leq k-3
& (h_{1:m} \co g_m \co h_{m:k-1}) \co g_m
&= (\tfrac{k-m-1}{k-2} + \tfrac{m-1}{k-2}) (h_{1:m} \co g_m \co h_{m:k-1}) \co g_m \nx
&&&= \addl{\tfrac{k-m-1}{k-2} ((h_{1:m} \co g_m \co h_{m:k-2}) \co g_m)} \co h_{k-2:k-1} \nx
&&\wha +h_{1:2} \co \addr{\tfrac{m-1}{k-2} ((h_{2:m} \co g_m \co h_{m:k-2}) \co g_m)},\\
& m=k-2 & (h_{1:k-2} \co g_{k-2} \co h_{k-2:k-1}) \co g_{k-2}
&= (\tfrac{1}{k-2} + \tfrac{k-3}{k-2}) (h_{1:k-2} \co g_{k-2} \co h_{k-2:k-1}) \co g_{k-2} \nx
&&&= \addl{\tfrac{1}{2(k-2)} ((g_1 \co h_{1:k-2}) \co g_1 + (h_{1:k-2} \co g_{k-2}) \co g_{k-2})} \co h_{k-2:k-1} \nx
&&\wha + h_{1:2} \co \addr{\tfrac{k-3}{k-2} ((h_{2:k-2} \co g_{k-2} \co h_{k-2:k-1}) \co g_{k-2})},\\
& m=k-1 & (h_{1:k-1} \co g_{k-1}) \co g_{k-1}
&= h_{1:2} \co \addr{((h_{2:k-1} \co g_{k-1}) \co g_{k-1})}.
}

Finally, we split $\fk_{1:k-1}$ in \eref{fk-def}.
The first term of the right-hand side of \eref{fk-def} can be split, for each $m$, as
\eq{
& m=2  &
(h_{1:2} \co g_2 \co h_{2:k-1}) \co h_{1:2}
&= \addl{((h_{1:2} \co g_2 \co h_{2:k-2}) \co h_{1:2})} \co h_{k-2:k-1},\\
& m=3  &
(h_{1:3} \co g_3 \co h_{3:k-1}) \co h_{2:3}
&= (\tfrac{k-4}{k-2} + \tfrac{2}{k-2}) (h_{1:3} \co g_3 \co h_{3:k-1}) \co h_{2:3} \nx
&&&\akkun \addl{\tfrac{k-4}{k-2} ((h_{1:3} \co g_3 \co h_{3:k-2}) \co h_{2:3})} \co h_{k-2:k-1} \nx
&&\wha + h_{1:2} \co \addr{\tfrac{1}{k-2} ((h_{2:3} \co g_3 \co h_{3:k-1}) \co h_{2:3})},\\
& 4 \leq m \leq k-3  &
(h_{1:m} \co g_m \co h_{m:k-1}) \co h_{m-1:m}
&= (\tfrac{k-m-1}{k-2} + \tfrac{m-1}{k-2}) (h_{1:m} \co g_m \co h_{m:k-1}) \co h_{m-1:m} \nx
&&&= \addl{\tfrac{k-m-1}{k-2} ((h_{1:m} \co g_m \co h_{m:k-2}) \co h_{m-1:m})} \co h_{k-2:k-1} \nx
&&\wha + h_{1:2} \co \addr{\tfrac{m-1}{k-2} ((h_{2:m} \co g_m \co h_{m:k-2}) \co h_{m-1:m})},\\
& m=k-2  & (h_{1:k-2} \co g_{k-2} \co h_{k-2:k-1}) \co h_{k-3:k-2}
&= (\tfrac{1}{k-2} + \tfrac{k-3}{k-2}) (h_{1:k-2} \co g_{k-2} \co h_{k-2:k-1}) \co h_{k-3:k-2} \nx
&&&\akkun \addl{\tfrac{1}{k-2} ((h_{1:k-2} \co g_{k-2}) \co h_{k-3:k-2})} \co h_{k-2:k-1} \nx
&&\wha + h_{1:2} \co \addr{\tfrac{k-3}{k-2} ((h_{2:k-2} \co g_{k-2} \co h_{k-2:k-1}) \co h_{k-3:k-2})},\\
& m=k-1  & (h_{1:k-1} \co g_{k-1}) \co h_{k-2:k-1} &= h_{1:2} \co \addr{((h_{2:k-1} \co g_{k-1}) \co h_{k-2:k-1})} ,}
and the second term can be split as
\eq{
& m=1  &
(g_1 \co h_{1:k-1}) \co h_{1:2}
&= \addl{((g_1 \co h_{1:k-2}) \co h_{1:2})} \co h_{k-2:k-1},\\
& m=2  &
(h_{1:2} \co g_2 \co h_{2:k-1}) \co h_{2:3}
&= (\tfrac{k-3}{k-2} + \tfrac{1}{k-2}) (h_{1:2} \co g_2 \co h_{2:k-1}) \co h_{2:3} \nx
&&&\akkun \addl{\tfrac{k-3}{k-2} ((h_{1:2} \co g_2 \co h_{2:k-2}) \co h_{2:3})} \co h_{k-2:k-1} \nx
&&\wha + h_{1:2} \co \addr{\tfrac{1}{k-2} (g_2 \co h_{2:k-1}) \co h_{2:3})},\\
& 3 \leq m \leq k-4  &
(h_{1:m} \co g_m \co h_{m:k-1}) \co h_{m:m+1}
&= (\tfrac{k-m-1}{k-2} + \tfrac{m-1}{k-2}) (h_{1:m} \co g_m \co h_{m:k-1}) \co h_{m:m+1} \nx
&&&= \addl{\tfrac{k-m-1}{k-2} ((h_{1:m} \co g_m \co h_{m:k-2}) \co h_{m:m+1})} \co h_{k-2:k-1} \nx
&&\wha + h_{1:2} \co \addr{\tfrac{m-1}{k-2} ((h_{2:m} \co g_m \co h_{m:k-2}) \co h_{m:m+1})},\\
& m=k-3  & (h_{1:k-3} \co g_{k-3} \co h_{k-3:k-1}) \co h_{k-3:k-2}
&= (\tfrac{2}{k-2} + \tfrac{k-4}{k-2}) (h_{1:k-3} \co g_{k-3} \co h_{k-3:k-1}) \co h_{k-3:k-2} \nx
&&&\akkun \addl{\tfrac{1}{k-2} ((h_{1:k-3} \co g_{k-3}) \co h_{k-3:k-2})} \co h_{k-2:k-1} \nx
&&\wha + h_{1:2} \co \addr{\tfrac{k-4}{k-2} ((h_{2:k-3} \co g_{k-3} \co h_{k-3:k-1}) \co h_{k-3:k-2})},\\
& m=k-2  & (h_{1:k-2} \co g_{k-2} \co h_{k-2:k-1}) \co h_{k-2:k-1} &= h_{1:2} \co \addr{((h_{2:k-2} \co g_{k-2} \co h_{k-2:k-1}) \co h_{k-2:k-1})} .
}

We now summarize the above splittings.
Defining $\chi^{(k)}_{1:k-2}$ and $\psi^{(k)}_{1:k-2}$ as
\eq{
\chi^{(k)}_{1:k-2} &:= 
\ddk_{1:k-2} \nx
\wha +\sum_{1 \leq m \leq k-3} \tfrac{k-m-1}{k-2} (h_{1:m} \co g_m \co h_{m:k-2}) \co g_m \nx
\wha + \tfrac{1}{2(k-2)} ((g_1 \co h_{1:k-2}) \co g_1 + (h_{1:k-2} \co g_{k-2}) \co g_{k-2}) \nx
\wha + \tfrac{1}{k-2} ( \erk_{1:k-2} + \elk_{1:k-2}) \nx
\wha + (h_{1:2} \co g_2 \co h_{2:k-2}) \co h_{1:2} \nx
\wha + \sum_{3 \leq m \leq k-2} \tfrac{k-m-1}{k-2} (h_{1:m} \co g_m \co h_{m:k-2}) \co h_{m-1:m} \nx
\wha + \sum_{1 \leq m \leq k-4} \tfrac{k-m-1}{k-2} (h_{1:m} \co g_m \co h_{m:k-2}) \co h_{m:m+1} \nx
\wha + \tfrac{1}{k-2} (h_{1:k-3} \co g_{k-3} \co h_{k-3:k-2}) \co h_{k-3:k-2} \nx
\wha - \tk g_1 \co \hb_{1:k-2} \nx
\wha + \sk_{1:k-2}, \\
\psi^{(k)}_{2:k-1} &:=
- \ddk_{2:k-1} \nx
\wha + \tfrac{1}{2(k-2)} ((g_2 \co h_{2:k-1}) \co g_2 + (h_{2:k-1} \co g_{k-1}) \co g_{k-1}) \nx
\wha +\sum_{3 \leq m \leq k-1} \tfrac{m-1}{k-2} (h_{2:m} \co g_m \co h_{m:k-1}) \co g_m \nx
\wha - \tfrac{1}{k-2} ( \erk_{1:k-2} + \elk_{1:k-2}) \nx
\wha + \tfrac{1}{k-2} (h_{2:3} \co g_3 \co h_{3:k-1}) \co h_{2:3} \nx
\wha + \sum_{4 \leq m \leq k-1} \tfrac{m-1}{k-2} (h_{2:m} \co g_m \co h_{m:k-1}) \co h_{m-1:m} \nx
\wha + \sum_{2 \leq m \leq k-3} \tfrac{m-1}{k-2} (h_{2:m} \co g_m \co h_{m:k-1}) \co h_{m:m+1} \nx
\wha + (h_{2:k-2} \co g_{k-2} \co h_{k-2:k-1}) \co h_{k-2:k-1} \nx
\wha + \tk \hb_{2:k-1} \co g_{k-1} \nx
\wha - \sk_{2:k-1} ,
}
we obtain
\eq{
\chi_{1:k-2}\pk \co h_{k-2:k-1} + h_{1:2} \co \psi_{2:k-1}\pk = \lo{k-2}
}
(a restatement of \eref{chi-psi-h}).

Here, we compute the following:
\eq{
\chi^{(k)}_{1:k-2} + \psi^{(k)}_{1:k-2} - h_{1:2} \co \omega^{(k)}_{2:k-2} + \omega^{(k)}_{1:k-3} \co h_{k-3:k-2} 
&= \sum_{1 \leq m \leq k-1} \tfrac{k-1}{k-2} (h_{1:m} \co g_m \co h_{m:k-2}) \co g_m \nx
\wha + \sum_{1 \leq m \leq k-1} \tfrac{k-1}{k-2} (h_{1:m} \co g_m \co h_{m:k-2}) \co h_{m-1:m} \nx
\wha + \sum_{1 \leq m \leq k-1} \tfrac{k-1}{k-2} (h_{1:m} \co g_m \co h_{m:k-2}) \co h_{m:m+1} \nx
\wha + \tk (-g_1 \co \hb_{1:k-2} + \hb_{1:k-2} \co g_2) \nx
\wha + \omega^{(k)}_{1:k-3} \co h_{k-3:k-2} - h_{1:2} \co \omega^{(k)}_{2:k-2} \nx
&= \tfrac{k-1}{k-2} ( d^{(k-1)}_{1:k-2} + e^{(k-1)}_{1:k-2} + f^{(k-1)}_{1:k-2}) \nx
\wha + \tk (-g_1 \co \hb_{1:k-2} + \hb_{1:k-2} \co g_2) \nx
\wha + (\omega^{(k)}_{1:k-3} - \tfrac{k-1}{k-2} \delta^{(k-1)}_{1:k-3}) \co h_{k-3:k-2} - h_{1:2} \co (\omega^{(k)}_{2:k-2} - \tfrac{k-1}{k-2} \delta^{(k-1)}_{2:k-2}) .
\lbl{GMC1-chi-psi}
}
As seen in the main text (\eref{chi-psi-3k}),
when a solution to Step~3 exists at $k$ (i.e., $Q_k \co H = \lo{k-2}$),
both sides of \eref{chi-psi-3k} are at most $k-3$-local.
From a comparison with \eref{GMC1-QH}, we find that if the right-hand side of \eref{GMC1-chi-psi} is at most $k-3$-local,
then a solution to Step~3 at $k-1$ exists (in fact, these two are equivalent).
Precisely, by setting
\eq{
t^{(k-1)} &= \tfrac{k-2}{k-1} \tk \nx
s^{(k-1)}_{1:k-3} &= \tfrac{k-2}{k-1} \omega^{(k)}_{1:k-3} - \delta^{(k-1)}_{1:k-3},
}
a solution to Step~3 at $k-1$ can be constructed.
In other words, the existence of a Step~3 solution at $k$ implies the existence of a Step~3 solution at $k-1$.
Combined with the discussion in the main text, this shows that the absence of a 3-local charge
implies the absence of any $k$-local charge for all $3 \leq k \leq N/2$.

\section{Models of degree 2 and less}
\lbl{s:degree2}
In this section, we show that when the degree of $g$ is two or less, that is, when the Hamiltonian takes the form
\eq{
H = \sums \bd_1 b_2 + b_1 \bd_2 + c_{20} (\bd_1)^2 + c_{10} \bd_1 + c_{11} n_1 + c_{01} b_1 + c_{02} b_1^2 ,
\lbl{deg2andless}
}
the model belongs to Type~C, meaning that it possesses $k$-local charges in the uniform sector for all $k$.
While \lref{step1} remains valid for $\deg(g) \leq 2$, \lref{step2} no longer applies in this case.
In contrast to the other Type~C models discussed so far in this paper, for each $k$ there exist $k$-local charges in both the (uniform) symmetric sector ($q_{+-} = q_{-+}$) and the antisymmetric sector ($q_{+-} = -q_{-+}$), as shown below by explicit construction.
\begin{lem}
The Hamiltonian \eqref{deg2andless} has (uniform) $k$-local charges for any $3 \leq k \leq N/2$.
\end{lem}
\begin{proof}
Two $k$-local charges can be written explicitly as
\eq{
Q_k^{\text{sym}} &= \sums \bd_1 b_k + b_1 \bd_k + 2c_{20} \bd_1 \bd_{k-1} + 2c_{02} b_1 b_{k-1} + c_{11} \left(\bd_1 b_{k-1} + b_1 \bd_{k-1} \right) + \bd_1 b_{k-2} + b_1 \bd_{k-2}, \\
Q_k^{\text{anti}} &= \sums  \bd_1 b_k - b_1 \bd_k .
}
A direct calculation confirms that both of these commute with the Hamiltonian~\eqref{deg2andless}.
\end{proof}

\section{Staggered local charges}
\lbl{s:staggered}
In this section, we present a complete analysis of the \textit{staggered} local charges of the Hamiltonian~\eqref{H-gen}, namely those local charges that are eigenoperators of the one-site shift superoperator $T$ with eigenvalue $-1$.

From \lref{step1}, it follows that all local charges of the Hamiltonian~\eqref{H-gen} are either uniform or staggered.
Since the uniform ones have been fully analyzed in the preceding sections, analyzing the staggered case completes the overall classification.
Moreover, \lref{step1} implies that the staggered $k$-local charges take the form
\eq{
Q_k = \sumi (-1)^i \left( q_{++} \bd_i \bd_{i+k-1} + q_{--} b_i b_{i+k-1} \right) + \lo{k-1} .
\lbl{Stag-st}
}

In what follows, we show that depending on the staggered local charges, the Hamiltonian~\eqref{H-gen}, when analyzed in terms of its staggered local charges, falls into three distinct types:
Type~SC, which admits $k$-local charges for every $3 \leq k \leq N/2$;
Type~SN, which admits none;
and Type~SN$^+$, which admits only the 3-local one.
These classifications are listed in Table~\ref{table-classify-II}.

The structure of this section is as follows.
We show that a certain class of models belongs to Type~SN in \lref{stag-N}.
Next, in \lref{stag-C1} and \lref{stag-C2}, we demonstrate that another class of models belongs to Type~SC by explicit construction.
Finally, we show that the remaining models fall into Type~SN$^+$ in \lref{stag-N+}.

\begin{lem}
\label{t:stag-N}
If the on-site term $g$ is neither of the form
\eq{
g = \sum_{x=1}^\infty c_{x0} (\bd)^x + c_{01} b + c_{02} b^2,
\label{Stag-N}
}
nor its Hermitian conjugate, then the Hamiltonian~\eqref{H-gen}
has no staggered $k$-local charges for any $3 \leq k \leq N/2$.
\end{lem}

\begin{proof}
We first outline the proof. 
First we show that $q_{++}$ in \eref{Stag-st} is zero if $g$ is not of the form \eref{Stag-N}.
From this, in the case where $g$ is neither of the form \eref{Stag-N} nor its Hermitian conjugate, we directly have $q_{++}=q_{--}=0$, implying that no staggered $k$-local charges exist.

We prove $q_{++}=0$ by considering the following three cases:
(i) $c_{xy} \neq 0$ for some $(x,y)$ with $x\geq 1$, $y\geq 1$, $(x,y)\neq(1,1)$,
(ii) $c_{11}\neq 0$, and
(iii) $c_{0y}\neq 0$ for some $y\geq 2$.

First, we consider the case that there exists a pair $(x,y)$ satisfying $x\geq 1$, $y\geq 1$, $(x,y)\neq (1,1)$, and $c_{xy}\neq 0$.
For such $(x,y)$, we consider all the contribution to
\eq{
\bB_{(m)} = \begin{cases}
\op{x,y-1}_0 \op{10}_{k-1} & (m=0) \\
\op{10}_{-m} \op{x-1,y-1}_0 \op{10}_{k-m-1} & (1 \leq m \leq k-2) \\
\op{10}_{-k+1} \op{x,y-1}_0 & (m=k-1)
\end{cases} .
}
By defining
\eq{
\bA_{(m)} = \begin{cases}
\op{10}_0 \op{10}_{k-1} & (m=-1) \\
\op{x,y-1}_0 \op{10}_{k-2} & (m=0) \\
\op{10}_{-m} \op{x-1,y-1}_0 \op{10}_{k-m-2} & (1 \leq m \leq k-3) \\
\op{10}_{-k+2} \op{x,y-1}_0 & (m=k-2) \\
\op{10}_{-k+1} \op{10}_0 & (m=k-1)
\end{cases} ,
}
we find that
\eq{
r_{\bB_{(m)}} = \begin{cases}
-y c_{xy} q_{\bA_{(-1)}} - q_{\bA_{(0)}} & (m=0) \\
-xq_{\bA_{(0)}} - q_{\bA_{(1)}} & (m=1) \\
-q_{\bA_{(m-1)}} - q_{\bA_{(m)}} & (2 \leq m \leq k-3) \\
-q_{\bA_{(k-3)}} - xq_{\bA_{(k-2)}} & (m=k-2) \\
-q_{\bA_{(k-2)}} - xc_{xy}q_{\bA_{(k-1)}} & (m=k-1)
\end{cases}
\lbl{Stag-rq1}
}
holds.
Since $r_{\bB}=0$ hold for all $\bB$, the following relation immediately holds:
\eq{
xr_{\bB_{(0)}} + \sum_{1\leq m \leq k-2} (-1)^m r_{\bB_{(m)}} + (-1)^{k-1} x r_{\bB_{(k-1)}} = 0
\lbl{Stag-rq2}.
}
By substituting \eref{Stag-rq1} into \eref{Stag-rq2}, we obtain
\eq{
-xy c_{xy} ( q_{\op{10}_0\op{10}_{k-1}} + (-1)^{k-1} q_{\op{10}_{-k+1}\op{10}_0} ) = 0.
}
In the staggered operator subspace, one has $q_{\op{10}_{-k+1}\op{10}_0} = (-1)^{k-1} q_{\op{10}_0\op{10}_{k-1}}$.
Hence, we finally obtain
\eq{
-2xy c_{xy} q_{\op{10}_0\op{10}_{k-1}} = 0.
}
Therefore, if there exists $(x,y)$ satisfying $x\geq1$, $y\geq1$, $(x,y)\neq(1,1)$, and $c_{xy}\neq0$, then $q_{++} = q_{\op{10}_0\op{10}_{k-1}} = 0$ holds.

Next, we consider the case of $(x,y) = (1,1)$.
In this case, we have
\eq{
r_{\op{10}_0 \op{10}_{k-1}} &=
-c_{11} q_{\op{10}_0 \op{10}_{k-1}}
-c_{11} q_{\op{10}_0 \op{10}_{k-1}}
- \cancel{q_{\op{10}_1 \op{10}_{k-1}}}
- \cancel{q_{\op{10}_0 \op{10}_{k-2}}} \nx
&= -2c_{11} q_{\op{10}_0 \op{10}_{k-1}}
}
Hence, $c_{11}\neq0$ also leads to $q_{++}=0$.

Finally, we consider the case that there exists an operator $\op{0y}$ with $x=0$ and $y\geq2$ that has a nonzero coefficient $c_{\op{0y}}$.
In this case, we consider all the contributions to
\eq{
\bB_{(m)}' = \begin{cases}
\op{0,y-1}_0 \op{10}_{k-1} & (m=0) \\
\op{01}_{-m} \op{0,y-2}_0 \op{10}_{k-m-1} & (1 \leq m \leq k-2) \\
\op{01}_{-k+1} \op{1,y-2}_0 & (m=k-1)
\end{cases} .
}
By writing
\eq{
\bA_{(m)}' = \begin{cases}
\op{10}_0 \op{10}_{k-1} & (m=-1) \\
\op{0,y-1}_0 \op{10}_{k-2} & (m=0) \\
\op{01}_{-m} \op{0,y-2}_0 \op{10}_{k-m-2} & (1 \leq m \leq k-3) \\
\op{10}_{-k+2} \op{1,y-2}_0 & (m=k-2)
\end{cases} ,
}
we find that
\eq{
r_{\bB_{(m)}} = \begin{cases}
-y c_{xy} q_{\bA_{(-1)}} - q_{\bA_{(0)}} & (m=0) \\
(y-1) q_{\bA_{(0)}} - q_{\bA_{(1)}} & (m=1) \\
q_{\bA_{(m-1)}} - q_{\bA_{(m)}} & (2 \leq m \leq k-2) \\
q_{\bA_{(k-2)}} & (m=k-1)
\end{cases}
}
holds.
Substituting it into
\eq{
(y-1)r_{\bB_{(0)}'} +  \sum_{1 \leq m \leq k-2} r_{\bB_{(m)}'} = 0,
}
we have
\eq{
-(y-1)yc_{xy} q_{\op{10}_0 \op{10}_{k-1}} = 0 ,
}
which implies $q_{++} = 0$.

Thus, we have shown that if there exists an operator $\op{xy}$ with a nonzero coefficient $c_{xy}$ such that both $x$ and $y$ are more than or equal to one, or if both $x=0$ and $y\geq 2$, hold, then $q_{++}=0$ holds.
This completes the proof.
\end{proof}

\begin{lem}
\lbl{t:stag-C1}
If the on-site term $g$ is of the form
\eq{
g = \sum_{x=1}^\infty c_{x0} (\bd)^x + c_{01} b
\lbl{Stag-C1}
}
or its Hermitian conjugate,
then the Hamiltonian \eqref{H-gen} has a staggered $k$-local charge for any $3 \leq k \leq N/2$.
\end{lem}

\begin{proof}
The staggered local charge can be written explicitly as
\eq{
Q_k^{++} =\begin{cases} 
\displaystyle \sumi (-1)^i \left( \bd_i \bd_{i+k-1} - (\bd_i)^2\right)  & (m \text{ is odd}) \\
\displaystyle \sumi (-1)^i \bd_i \bd_{i+k-1} & (m \text{ is even})
\end{cases} .
}
The verification of its commutation with the Hamiltonian is straightforward.
\end{proof}

\begin{lem}
\lbl{t:stag-C2}
If the on-site term $g$ is of the form
\eq{
g = c_{20} (\bd)^2 + c_{10} \bd + c_{01} b + c_{02} b^2,
\lbl{Stag-C2}
}
then the Hamiltonian \eqref{H-gen} has staggered $k$-local charges for every $3 \leq k \leq N/2$.
\end{lem}
\begin{proof}
Two staggered local charges can be written explicitly as
\eq{
Q_k^{++} &= \begin{cases} 
\displaystyle
\sumi (-1)^i \left( 
\bd_i \bd_{i+k-1} 
+ c_{02} \left( \bd_i b_{i+k-2} - b_i \bd_{i+k-2} \right)
+ \bd_i \bd_{i+k-3} 
- c_{02} g_i 
+ 2c_{01} \bd_i \right)
& (m \text{ is odd}) \\
\displaystyle
\sumi (-1)^i \left( 
\bd_i \bd_{i+k-1} 
+ c_{02} \left( \bd_i b_{i+k-2} - b_i \bd_{i+k-2} \right)
+ \bd_i \bd_{i+k-3} \right)
& (m \text{ is even}) \\
\end{cases}
,\\
Q_k^{--} &= \begin{cases} 
\displaystyle
\sumi (-1)^i \left( 
b_i b_{i+k-1} 
+ c_{20} \left( b_i \bd_{i+k-2} - \bd_i b_{i+k-2} \right)
+ b_i b_{i+k-3} 
- c_{20} g_i 
+ 2c_{10} b_i
\right)
& (m \text{ is odd})\\
\displaystyle
\sumi (-1)^i \left( 
b_i b_{i+k-1} 
+ c_{20} \left( b_i \bd_{i+k-2} - \bd_i b_{i+k-2} \right)
+ b_i b_{i+k-3}
\right)
& (m \text{ is even}) \\
\end{cases}.
}
The verification of their commutation with the Hamiltonian is straightforward.
\end{proof}

\begin{lem}
\lbl{t:stag-N+}
If the on-site term $g$ is of the form
\eq{
g = \sum_{x=1}^{\infty} c_{x0} (\bd)^x + c_{01} b + c_{02} b^2,
\lbl{Stag-N+}
}
with $c_{02} \neq 0$ and $c_{x0} \neq 0$ for some $x \geq 3$,
or its Hermitian conjugate,
then the Hamiltonian \eqref{H-gen} has a staggered 3-local charge,
and no staggered $k$-local charges for any $4 \leq k \leq N/2$.
\end{lem}

\begin{proof}
The staggered 3-local charge can be written explicitly as
\eq{
Q_3 = \sumi (-1)^i \left( 
\bd_i \bd_{i+2} 
+ c_{02} \left( \bd_i b_{i+1} - b_i \bd_{i+1} \right)
+ (\bd_i)^2 - c_{02} g_i + 2c_{01} \bd_i
\right).
}

For $k \geq 4$, the normalized general solution of Step~4 (i.e., the operator $\tilde{Q}_k$ satisfying $\tilde{Q}_k \co H = \lo{k-2}$) can be written as
\eq{
\tilde{Q}_k = \sumi (-1)^i \Bigg(
& \bd_i \bd_{i+k-1} 
+ c_{02} \left( \bd_i b_{i+k-2} - b_i \bd_{i+k-2} \right) \nx
&- \sum_{x=3}^\infty c_{02} c_{x0} \left( 
x (\bd_i)^{x-1} \bd_{i+k-3} + x(x-1) \sum_{i+1\leq m\leq i+k-4} \bd_i (\bd_m)^{x-2} \bd_{i+k-3} + x \bd_i (\bd_{i+k-3})^{x-1} \right) \nx
&+ t_1 \tilde{Q}_{k-1} + t_2 \bd_i \bd_{i+k-3} + t_3 b_i b_{i+k-3}
\Bigg) + \lo{k-3} ,
\lbl{Stag-step4}
}
where $t_1$, $t_2$, and $t_3$ are free parameters.
It can be verified by direct computation that \eref{Stag-step4} indeed satisfies the condition for Step~4.
Conversely, as will be shown below, the general solution of Step~4 for a given $k$ is given by \eref{Stag-step4}, up to an overall constant factor.

Let us express the general solution of Step~4 at $k$ as $\tilde{Q}_k+\Delta$.
Then
\eq{
\Delta \co H = \lo{k-2}
\lbl{Stag-delta}
}
holds.
Since both $\tilde{Q}_k$ and $\tilde{Q}_k+\Delta$ are $k$-local quantities, $\Delta$ is at most $k$-local.
First, when $\Delta$ is $k$-local, \eref{Stag-delta} imposes the Step~1 and Step~2 constraints on $\Delta$.
By the argument in the proof of \lref{stag-N}, the $k$-local part of $\Delta$ can only be $\sumi (-1)^i \bd_i \bd_{i+k-1}$,
which merely changes the overall normalization of \eref{Stag-step4}.
Next, when $\Delta$ is $k-1$-local, \eref{Stag-delta} likewise imposes the Step~1 and Step~2 constraints on $\Delta$.
By the argument in the proof of \lref{stag-N}, the $k-1$-local part of $\Delta$ can only be $\sumi (-1)^i \bd_i \bd_{i+k-2}$, and this contribution is absorbed into the free parameter $t_1$.
Finally, when $\Delta$ is $k-2$-local, \eref{Stag-delta} is equivalent to the Step~1 condition for $\Delta$.
By \lref{step1}, the $k-2$-local part of $\Delta$ is a linear combination of $\sumi (-1)^i \bd_i \bd_{i+k-3}$ and $\sumi (-1)^i b_i b_{i+k-3}$, whose contributions are absorbed into the free parameters $t_2$ and $t_3$, respectively.
Consequently, the form \eref{Stag-step4} is both necessary and sufficient for being a solution of Step~4 (up to an overall constant factor).

To conclude the argument, we show that the quantity $\tilde{Q}_k$ in \eref{Stag-step4} cannot serve as a solution satisfying the next consistency condition (Step~5), thereby proving that no $k$-local charge exists.
For $k \geq 6$, consider all the contributions to 
\eq{
\bB_{(m)} = \begin{cases}
\op{x-2,1}_0 \op{10}_{k-3} & (m=0) \\
\op{10}_0 \op{x-3,1}_m \op{10}_{k-3} & (1 \leq m \leq k-4) \\
\op{10}_0 \op{x-2,1}_{k-3} & (m=k-3)
\end{cases} .
}
By writing
\eq{
\bA_{(m)} = \begin{cases}
\op{x-1,0}_0 \op{10}_{k-3} & (m=-1) \\
\op{x-2,1}_0 \op{10}_{k-4} & (m=0) \\
\op{10}_0 \op{x-3,1}_m \op{10}_{k-4} & (1 \leq m \leq k-5) \\
\op{10}_0 \op{x-2,1}_{k-4} & (m=k-4) \\
\op{10}_0 \op{x-1,0}_{k-3} & (m=k-3)
\end{cases} ,
}
we find that
\eq{
r_{\bB_{(m)}} = \begin{cases}
-2(x-1) c_{02} q_{\bA_{(-1)}} - q_{\bA_{(0)}} & (m=0) \\
(x-2) q_{\bA_{(0)}} - q_{\bA_{(1)}} & (m=1) \\
q_{\bA_{(m-1)}} - q_{\bA_{(m)}} & (2 \leq m \leq k-5) \\
q_{\bA_{(k-5)}} - (x-2) q_{\bA_{(k-4)}} & (m=k-4) \\
q_{\bA_{(k-4)}} - 2(x-1) c_{02} q_{\bA_{(k-3)}} & (m=k-3)
\end{cases}
}
holds.
Using this relation, we obtain
\eq{
 (x-2) r_{\bB_{(0)}} + 2(x-1) c_{02} \sum_{1 \leq m \leq k-4} r_{\bB_{(m)}} + (x-2) r_{\bB_{(k-3)}} 
&= -2 (x-2) (x-1) c_{02} (q_{\bA_{(-1)}} - q_{\bA_{(k-3)}}) \nx
&= -4 (x-2) (x-1) c_{02} q_{\bA_{(-1)}} \nx
&= -4 (x-2) (x-1) x c_{02}^2 c_{x0} \nx
&\neq 0.
}
Here we used $x\geq 3$ together with $c_{02}\neq 0$ and $c_{x0}\neq 0$.
Since $r_{\bB}=0$ must hold for all $\bB$ in a charge, we conclude that $\tilde{Q}_k$ cannot be a $k$-local charge.

For $k = 5$, since
\eq{
\begin{cases}
r_{\bB_{(0)}} = -2(x-1) c_{02} q_{\bA_{(-1)}} - q_{\bA_{(0)}} \\
r_{\bB_{(1)}} = (x-2) q_{\bA_{(0)}} - (x-2) q_{\bA_{(1)}} \\
r_{\bB_{(2)}} = q_{\bA_{(1)}} - 2(x-1) c_{02} q_{\bA_{(2)}}
\end{cases}
}
holds, we obtain
\eq{
(x-2) r_{\bB_{(0)}} + 2(x-1)c_{02} r_{\bB_{(1)}} + (x-2) r_{\bB_{(2)}}
= -4(x-2)(x-1)xc_{02}^2c_{x0} \neq 0,
}
which again shows that $\tilde{Q}_k$ is not a charge.

For $k = 4$, from
\eq{
\begin{cases}
r_{\op{x-2,1}_0 \op{10}_1} = -2(x-1) c_{02} q_{\op{x-1,0}_0\op{10}_1} - (x-1)q_{\op{x-1,1}_0} \\
r_{\op{10}_0 \op{x-2,1}_1} = (x-1)q_{\op{x-1,1}_0} -2(x-1) c_{02} q_{\op{10}_0\op{x-1,0}_1}\\
\end{cases},
}
we obtain
\eq{
r_{\op{x-2,1}_0 \op{10}_1} + r_{\op{10}_0 \op{x-2,1}_1}
= -4 (x-1) x c_{02}^2 c_{x0}
\neq 0,
}
which likewise violates the condition for a charge.

Consequently, \eref{Stag-N+} is identified as a Type~N$^+$ model, namely, a system that possesses a 3-local charge but no $k$-local charges for any $4 \leq k \leq N/2$.
\end{proof}

%\end{CJK*}
\end{document}